\documentclass[12pt]{article}

\usepackage{amsmath, amssymb, amsthm}
\usepackage{graphicx}
\usepackage{booktabs}
\usepackage{float}
\usepackage{placeins}
\usepackage{longtable}
\usepackage{geometry}
\usepackage{natbib}
\usepackage{setspace}
\usepackage{hyperref}
\usepackage{caption}
\usepackage{microtype}
\usepackage{multirow}
\usepackage{algorithm}
\usepackage{algpseudocode}

\numberwithin{equation}{section}

\theoremstyle{plain}
\newtheorem{theorem}{Theorem}[section]
\newtheorem{assumption}{Assumption}
\newtheorem{lemma}{Lemma}[section]
\newtheorem{proposition}{Proposition}[section]
\newtheorem{corollary}{Corollary}[section]

\theoremstyle{remark}
\newtheorem{remark}{Remark}[section]

\newcommand{\assMixing}{Assumption~S1}
\newcommand{\assRegularity}{Assumption~S2}
\newcommand{\assKernel}{Assumption~S3}
\newcommand{\assPreEst}{Assumption~S4}

\newcommand{\assLocal}{Assumption~S6}

\title{\textbf{Uniform Inference for Parameters Identified by
Conditional Quantile Restrictions}}

\author{  Xuqing Lin
  \and
  Xiaojun Song\thanks{Corresponding author.
  Email address: \texttt{sxj@gsm.pku.edu.cn}.}
}

\date{
  Guanghua School of Management, Peking University
  \\[0.5em]
  September 21, 2026
}

\begin{document}
\maketitle

\begin{abstract}
Many structural and dynamic economic models imply that key parameters are identified by conditional quantile restrictions. Building on the exponential weighting approach of \citet{bierens1990} and the penalized maximum statistics of \citet{chen2025}, we develop a uniform inference method for such parameters. We construct an adaptive $\ell_1$-penalized statistic from a process indexed by quantiles and weighting directions. The statistic aggregates evidence across quantiles and regularizes maximization over directions. A shifted bootstrap is used to estimate local power for penalty selection. Under the stated conditions, adaptive selection asymptotically preserves or improves worst case local power relative to the unpenalized test over a prespecified set of alternatives. The improvement is strict when a positive candidate penalty raises the population level criterion above its value at zero. We also extend the process and bootstrap theory to the setting with pre-estimated nuisance parameters. We characterize their contribution to the limiting process, correct the variance estimator, and establish Gaussian multiplier bootstrap validity for ites null inference. Monte Carlo experiments favor CvM--KS method for size control under pre-estimation and show that adaptive penalization can improve empirical power relative to the unpenalized test.
\\[1em]
\end{abstract}

\section{Introduction}
\label{sec:introduction}

Conditional quantiles describe heterogeneity and asymmetry that conditional
means do not capture. Since \citet{koenker1978}, they have become standard in
empirical finance and macroeconomics. And \citet{koenker2005} gives a general
treatment. In time series, dynamic conditional quantile models and quantile
autoregression model allow persistence and adjustment to vary across the conditional
distribution \citep{koenkerxiao2006}. Empirical work with these models often
requires tests and confidence sets over a range of
quantiles. We study this problem when the parameter is identified by the conditional
moment restriction
\begin{equation}\label{eq:cond_moment}
    \mathbb{E}\big[\mathbf{1}\{Y_t \le m(I_{t-1}, \theta_0(\alpha))\} - \alpha \mid I_{t-1}\big] = 0 \quad \forall \alpha \in \mathcal{T},
\end{equation}
where $Y_t$ is the response, $I_{t-1}$ is the conditioning information, and
$m(\cdot,\theta)$ is a parametric conditional quantile function. We test whether $\theta_0(\alpha)$ equals to $\bar\theta(\alpha)$, at one quantile level or uniformly over
$\alpha\in\mathcal T\subset(0,1)$. The indicator makes the restriction
nonsmooth.

Uniform inference for quantile regression model is well developed.
\citet{koenkerxiao2002} use a Khmaladze transformation,
\citet{chernozhukov2005} use subsampling, and \citet{chen2026fixed} develop
a fixed smoothing inference for time series quantile regression that is uniform
over quantile levels. Related work also allows structural change in quantile models \citep{okaqu2011}. These papers provide inference method for
quantile regression models. We want to test whether a fixed parameter
restriction satisfies the conditional quantile restriction.

There is also a large literature on specification tests for conditional
quantile models and conditional distributions. Kernel methods provide omnibus tests
of parametric quantile restrictions \citep{zheng1998}. ICM and marked process
methods test dynamic and linear quantile specifications over conditioning
directions or quantile indices \citep{escanciano2010,escancianogoh2014}.
Other approaches estimate the conditional moment by series regression
\citep{horvath2022}, permit flexible quantile dependence and nonlinear effects
\citep{kutzker2024}, or project out estimated nuisance components in partially
linear quantile models \citep{song2025}. \citet{rothe2013} study a broader
class of conditional distributional models. These contributions focus on
overall model adequacy or on different estimands and aggregation rules. They
also make clear that estimation effects cannot be generally ignored.

Our method starts from the exponential weighting ICM transformation of
\citet{bierens1990}. We index the studentized process jointly by the quantile
level $\alpha$ and an exponential weighting direction $\gamma\in\Gamma$. And then integrate
its square over $\alpha$, and maximize over $\gamma$. The resulting CvM--KS
statistic targets a parameter restriction while retaining a rich search over
the conditioning information. We establish uniform convergence theorem over
both indices and use a Gaussian multiplier bootstrap method to approximate the
non-pivotal null distribution. Recent ICM work also pursues pivotal and
computationally inexpensive statistics, including directionally targeted power
improvements \citep{jiang2026}. That work and the present statistic address
different goals.

We also derive the effect of estimating nuisance parameters under the null.
The influence function of the restricted estimator contributes an
additional term to the limiting process and changes its covariance kernel. We account
for this term in both the variance estimator and the multiplier bootstrap procedures,
and derive the corresponding drift term under local alternatives. The argument is
related to classical empirical process results with estimated parameters
\citep{durbin1973}, and treats the estimation effect uniformly over quantiles
and conditioning directions. The Monte Carlo section compares aggregation
orders under pre-estimation.

The closest methodological benchmark is the penalized generalized Bierens
maximum statistic of \citet{chen2025}, which selects tuning parameters by a
data driven maxmin rule for conditional moment restrictions. We
adapt this approach to the quantile regression model. The $\ell_1$-penalty enters
inside the directional supremum, and we use bootstrap method to estimate the maximin
criterion over a collection of local directions. Under the stated
conditions, we establish uniform consistency of this criterion and convergence
of the selector to its population maximizer. These results connect feasible
penalty selection to the limiting rejection probabilities. Because the candidate
set contains $\lambda=0$, the local power of our method is not lower than the unpenalized procedure. Strict improvement requires a positive
penalty to raise the population maximin criterion.

The remainder of the paper is organized as follows.
Section~\ref{sec:test_stat} introduces the test statistic and its asymptotic null theory.
Section~\ref{sec:bootstrap_power} describes multiplier bootstrap inference and analysis of local power.
Section~\ref{sec:pre_estimation} develops the pre-estimation extension.
Sections~\ref{sec:simulation} and \ref{sec:empirical} report Monte Carlo and empirical results.
Section~\ref{sec:conclusion} concludes. Proofs are collected in the Supplementary Appendix.

\section{Test Statistics and Asymptotic Theory}
\label{sec:test_stat}

This section introduces the penalized Bierens-type statistic for conditional
quantile restrictions and the objects used for inference. The formal
regularity conditions are stated in the Supplementary Appendix. Under these
conditions, we establish weak convergence of the empirical process and the
resulting null limit distribution.

\subsection{The Penalized Bierens Test Statistic}

We are interested in testing whether the true parameter function $\theta_0(\cdot)$ equals a prespecified value $\bar{\theta}(\cdot)$:
\[
    H_0: \theta_0(\alpha) = \bar{\theta}(\alpha) \;\text{ for all }\; \alpha \in \mathcal{T}
    \qquad \text{vs.} \qquad
    H_1: \theta_0(\alpha) \neq \bar{\theta}(\alpha) \;\text{ for some }\; \alpha \in \mathcal{T}.
\]
In Sections~\ref{sec:test_stat}--\ref{sec:bootstrap_power},
$\bar{\theta}(\cdot)$ is a fixed curve, such as
$\bar{\theta}(\cdot)\equiv0$ in a significance test. Testing $H_0$ via \eqref{eq:cond_moment} is
challenging because it corresponds to a continuum of unconditional moment restrictions. ICM methods address this via weighting and integration. Following the
exponential weighting approach of \citet{bierens1990} and the generalized
penalization framework of \citet{chen2025}, we transform the conditional
restriction into an equivalent continuum of unconditional ones using the weighting
function $w(I_{t-1},\gamma)=\exp(\gamma^\top\Phi(I_{t-1}))$. Here,
$\Phi:\mathbb{R}^{d_I}\to\mathbb{R}^{d_I}$ is a one-to-one transformation
so that $\exp(\gamma^\top\Phi(I_{t-1}))$ has finite moments uniformly over
$\gamma\in\Gamma$ (e.g., identity for well-behaved covariates, or a bounded map
such as componentwise $\tan^{-1}(\cdot)$ for heavy-tailed covariates). For
notational simplicity, we write $\exp(\gamma^\top I_{t-1})$ for
$\exp\{\gamma^\top\Phi(I_{t-1})\}$ below.

Define the empirical process indexed by the quantile level $\alpha$ and the
weighting direction $\gamma$:
\begin{equation}
    M_n(\alpha, \gamma)
    :=
    \frac{1}{n}\sum_{t=1}^{n}
    \exp(\gamma^{\top} I_{t-1})
    \big[\mathbf{1}\{Y_t \le m(I_{t-1}, \bar{\theta}(\alpha))\} - \alpha\big],
    \label{eq:Mn}
\end{equation}
where $\gamma \in \Gamma$, a compact subset of $\mathbb{R}^{d_I}$ containing the origin. To studentize the process, we scale $\sqrt n\,M_n(\alpha,\gamma)$ by its estimated asymptotic standard deviation $s_n(\alpha,\gamma)$, which is available in closed form under \eqref{eq:cond_moment}. The normalization term is given by the sample analog of the theoretical variance:
\begin{equation}
    s_n^2(\alpha, \gamma)
    :=
    \frac{\alpha(1-\alpha)}{n}\sum_{t=1}^{n}\exp(2\gamma^{\top} I_{t-1}).
    \label{eq:sn}
\end{equation}

The closed-form expression in \eqref{eq:sn} is derived using the martingale difference
sequence (MDS) property of the centered indicator $\mathbf{1}\{Y_t \le m(I_{t-1},
\bar{\theta}(\alpha))\} - \alpha$ under $H_0$ (\assMixing(iii)).
This MDS condition is stronger than the conditional restriction
in \eqref{eq:cond_moment} when $I_{t-1}$ contains less information than
$\mathcal F_{t-1}$. This condition is maintained here to eliminate cross-lag covariances.
Specifically, the conditional variance of this term given $I_{t-1}$ equals
$\alpha(1-\alpha)$ regardless of $I_{t-1}$, so the conditional
variance of each summand in $\sqrt{n}\, M_n(\alpha, \gamma)$ is
$\alpha(1-\alpha)\exp(2\gamma^\top I_{t-1})$, and \eqref{eq:sn} is its sample
analog.

A more general alternative is the sample second moment estimator
\begin{equation}
    \tilde{s}_n^2(\alpha, \gamma)
    := \frac{1}{n}\sum_{t=1}^n \exp(2\gamma^\top I_{t-1})
    \big[\mathbf{1}\{Y_t \le m(I_{t-1}, \bar{\theta}(\alpha))\} - \alpha\big]^2,
    \label{eq:sn-robust}
\end{equation}
constructed directly from the squared summands of $\sqrt{n}\, M_n(\alpha, \gamma)$.
Under the MDS property, the cross terms in the long run variance vanish, so
$\tilde{s}_n^2(\alpha, \gamma)$ and $s_n^2(\alpha, \gamma)$ share the same
probability limit $\sigma^2(\alpha, \gamma) = \alpha(1-\alpha)\,
E[\exp(2\gamma^\top I_{t-1})]$, and either can be used to studentize the empirical
process. We use $s_n^2$ in \eqref{eq:sn} because it is simple and stable in
finite samples. The same first order results hold for $\tilde{s}_n^2$ when its analog is also used in the bootstrap.

The standardized process is
\[
Q_n(\alpha, \gamma)
:=
\frac{\sqrt{n}|M_n(\alpha, \gamma)|}{s_n(\alpha, \gamma)}.
\]
To conduct inference uniformly over \(\alpha\in\mathcal{T}\), we aggregate evidence
across quantile indices through a fixed probability measure \(\Pi\) on
\(\mathcal T\). Throughout the theoretical analysis, \(\mathcal T\) is a compact
subset of \((0,1)\), and \(\Pi\) is nonrandom. For instance, when
\(\mathcal T=[\underline\alpha,\overline\alpha]\), \(\Pi\) may be the uniform
probability measure on this interval. The finite grid implementation is obtained
as the special case \(\Pi_A=A^{-1}\sum_{j=1}^A\delta_{\alpha_j}\), where
\(\mathcal T_A=\{\alpha_1,\ldots,\alpha_A\}\).
For a continuous $\Pi$, the parameter curves have Borel measurable paths in
$\alpha$, so the sample processes can be integrated against $\Pi$.

The standardized process is aggregated by a CvM functional over \(\alpha\)
and a penalized supremum over \(\gamma\in\Gamma\). We define the
\textbf{penalized Bierens statistic} as
\begin{equation}
    T_{n,\Pi}(\lambda)
    :=
    \sup_{\gamma \in \Gamma}
    \left[
        \int_{\mathcal T} Q_n^2(\alpha,\gamma)\,d\Pi(\alpha)
        -
        \lambda\|\gamma\|_1
    \right].
    \label{eq:Tn}
\end{equation}
For the fixed penalty asymptotic results below, \(\lambda\ge0\) is treated as a
deterministic index. The data driven selector \(\widehat\lambda\) is introduced in Section~\ref{subsec:calibration}. When no confusion arises, we write
\(T_n(\lambda)\) for \(T_{n,\Pi}(\lambda)\). The statistic combines CvM
aggregation over \(\alpha\) with a penalized supremum over \(\gamma\). The
penalty discourages weighting directions with large \(\ell_1\)-norms.

\begin{remark}[Finite grid implementation and aggregation schemes]
\label{rem:alternative-schemes}
The formulation in \eqref{eq:Tn} covers continuous and discrete aggregation
over the quantile index. If \(\Pi\) is the uniform probability
measure on \([\underline\alpha,\overline\alpha]\), then
\[
    \int_{\mathcal T} Q_n^2(\alpha,\gamma)\,d\Pi(\alpha)
    =
    \frac{1}{\overline\alpha-\underline\alpha}
    \int_{\underline\alpha}^{\overline\alpha}
    Q_n^2(\alpha,\gamma)\,d\alpha.
\]
If instead
\(\Pi_A=A^{-1}\sum_{j=1}^A\delta_{\alpha_j}\), then
\[
    T_{n,\Pi_A}(\lambda)
    =
    \sup_{\gamma\in\Gamma}
    \left[
        \frac1A\sum_{j=1}^A Q_n^2(\alpha_j,\gamma)
        -
        \lambda\|\gamma\|_1
    \right].
\]
This replaces quantile integration by an average over the chosen grid points.
The numerical implementation also replaces the supremum over \(\Gamma\)
by the maximum over a finite grid \(\Gamma_G\).
We also write
\[
    T_{n,\Pi}^{CvM\text{-}KS}(\lambda)
    :=
    T_{n,\Pi}(\lambda),
\]
to denote CvM integration over \(\alpha\) followed by a KS supremum over
\(\gamma\).

The numerical experiments compare this statistic with two alternatives.
Writing \(\eta\ge0\) for an optional penalty on the quantile index, define
\begin{equation*}
    T_n^{KS\text{-}KS}(\eta,\lambda)
    :=\sup_{\alpha\in\mathcal T,\gamma\in\Gamma}
    \left[|Q_n(\alpha,\gamma)|-\eta|\alpha|-\lambda\|\gamma\|_1\right].
\end{equation*}
At \(\eta=\lambda=0\), this is the fully supremum based counterpart.
The other comparator maximizes over \(\gamma\) at each quantile before
squaring and integrating:
\begin{equation*}
    T_{n,\Pi}^{KS\text{-}CvM}(\lambda)
    :=\int_{\mathcal T}
    \left\{\sup_{\gamma\in\Gamma}
    \left[|Q_n(\alpha,\gamma)|-\lambda\|\gamma\|_1\right]\right\}^{2}
    d\Pi(\alpha).
\end{equation*}
In the second comparator, the
penalty enters before squaring, so penalized comparisons depend on both
aggregation and penalty placement.

We choose the CvM--KS statistic as our main target because it has rejection rates relatively close to nominal
levels in our simulation designs. This finding is specific to
those simulations.
\end{remark}

\subsection{Asymptotic Null Distribution}

The assumptions for the null theory are stated in the Supplementary Appendix and the following theorem gives the null limit of the empirical process and test
statistic.

\begin{theorem}[Asymptotic null distribution]
\label{thm:null_main}
Suppose Assumptions~S1--S2 hold. Under the null hypothesis $H_0: \theta_0(\alpha) = \bar{\theta}(\alpha)$ for all $\alpha \in \mathcal{T}$, the following results hold:

\begin{enumerate}
    \item The empirical process converges weakly to a Gaussian process with mean zero:
    \begin{equation*}
    \sqrt{n} M_n(\alpha, \gamma) \rightsquigarrow \mathcal{M}(\alpha, \gamma) \quad \text{in } \ell^{\infty}(\mathcal{T} \times \Gamma).
    \end{equation*}
    The covariance kernel of the limiting process $\mathcal{M}(\alpha, \gamma)$ is given by:
    \begin{equation*}
    \mathrm{Cov}\big(\mathcal{M}(\alpha_1, \gamma_1), \mathcal{M}(\alpha_2, \gamma_2)\big)
    = \big(\min(\alpha_1, \alpha_2) - \alpha_1 \alpha_2\big)
    E\!\left[\exp\!\big((\gamma_1 + \gamma_2)^{\top} I_{t-1}\big)\right].
    \end{equation*}

    \item The variance estimator converges uniformly in probability:
    \begin{equation*}
    \sup_{\alpha\in \mathcal{T}, \gamma\in \Gamma}
    \Big|
    s_n^2(\alpha, \gamma)
    - \sigma^2(\alpha, \gamma)
    \Big|
    \overset{p}{\longrightarrow} 0,
    \end{equation*}
    where $\sigma^2(\alpha, \gamma) := \alpha(1-\alpha) E[\exp(2\gamma^{\top} I_{t-1})]$.

    \item The test statistic \(T_{n,\Pi}(\lambda)\) converges in distribution to
    \begin{equation*}
    T_{\Pi}(\lambda)
    :=
    \sup_{\gamma\in \Gamma}
    \Bigg[
    \int_{\mathcal T}
    \left( \frac{\mathcal{M}(\alpha, \gamma)}{\sigma(\alpha, \gamma)} \right)^2
    d\Pi(\alpha)
    - \lambda \|\gamma\|_1
    \Bigg].
    \end{equation*}
\end{enumerate}
\end{theorem}

\subsection{Consistency under Global Alternatives}

We next study consistency against fixed alternatives. Define the population analog of the studentized moment:
\begin{equation}
    h(\alpha, \gamma, \theta)
    := \frac{E\big[\exp(\gamma^\top I_{t-1})
    \big(\mathbf{1}\{Y_t \le m(I_{t-1}, \theta(\alpha))\} - \alpha\big)\big]}
    {\sigma(\alpha, \gamma)},
    \label{eq:h-pop}
\end{equation}
where $\sigma(\alpha, \gamma) = \sqrt{\alpha(1-\alpha)\, E[\exp(2\gamma^\top I_{t-1})]}$ is the null standard deviation.

Under \(H_0\), the conditional expectation gives
\(h(\alpha,\gamma,\bar\theta)=0\) for all
\((\alpha,\gamma)\in\mathcal T\times\Gamma\). Under the maintained conditions,
the completeness argument of \citet{bierens1990} gives the converse at each
quantile, as shown in Lemma~\ref{lem:bierens} in the Supplementary Appendix.
Consistency under a fixed alternative follows if
\(h(\cdot,\gamma,\bar\theta)\) has a positive \(L^2(\Pi)\) norm for some
\(\gamma\in\Gamma\), as stated in the following theorem.

\begin{theorem}[Consistency under fixed alternatives]
\label{thm:consistency_main}
Suppose Assumptions~S1--S2 hold. If
\[
    \sup_{\gamma\in\Gamma}
    \int_{\mathcal T}
    h^2(\alpha,\gamma,\bar\theta)
    d\Pi(\alpha)
    >0,
\]
then for any fixed \(\lambda \geq 0\),
\[
    T_{n,\Pi}(\lambda) \xrightarrow{p} +\infty.
\]
\end{theorem}

To see this, note that
\[
    \frac{1}{n}\,T_{n,\Pi}(\lambda)
    =
    \sup_{\gamma \in \Gamma}
    \left[
        \int_{\mathcal T}
        \frac{M_n^2(\alpha,\gamma)}
        {s_n^2(\alpha,\gamma)}
        d\Pi(\alpha)
        -
        \frac{\lambda}{n}\|\gamma\|_1
    \right].
\]
Under the alternative, the uniform law of large numbers yields
\[
    \sup_{\alpha\in\mathcal T,\gamma\in\Gamma}
    \left|
        \frac{M_n^2(\alpha,\gamma)}
        {s_n^2(\alpha,\gamma)}
        -
        h^2(\alpha,\gamma,\bar\theta)
    \right|
    \xrightarrow{p}0,
\]
while \(\lambda/n\to0\). Hence, by the continuous mapping theorem,
\[
    \frac{1}{n}\,T_{n,\Pi}(\lambda)
    \xrightarrow{p}
    \sup_{\gamma\in\Gamma}
    \int_{\mathcal T}
    h^2(\alpha,\gamma,\bar\theta)
    d\Pi(\alpha)
    >0.
\]
Since the limit is a positive constant, \(T_{n,\Pi}(\lambda)\to+\infty\) in
probability. And if \(\Pi\) has full support on \(\mathcal T\) and
\(\alpha\mapsto h(\alpha,\gamma,\bar\theta)\) is continuous, the condition above
is implied by the existence of a pair \((\alpha',\gamma')\) such that
\(h(\alpha',\gamma',\bar\theta)\neq0\).

\section{Bootstrap Inference and Local Power Analysis}
\label{sec:bootstrap_power}

\subsection{Multiplier Bootstrap Implementation}

Because the covariance kernel of $\mathcal M$ depends on the data generating
process, $T_{n,\Pi}(\lambda)$ is non-pivotal. We use a Gaussian multiplier
bootstrap method to estimate the null distribution. Conditional on the data, the bootstrap matches the
covariance of the scores and the MDS assumption makes the cross-lag covariances
zero.

Let $\{\omega_t\}_{t=1}^n$ be i.i.d.\ standard normal multipliers,
independent of the original sample $\mathcal Z_n$. We define the bootstrap
empirical process as
\begin{equation}
    M_{n*}(\alpha, \gamma)
    :=
    \frac{1}{n} \sum_{t=1}^{n} \omega_t \exp(\gamma^{\top} I_{t-1})
    \big[\mathbf{1}\{Y_t \le m(I_{t-1}, \bar{\theta}(\alpha))\} - \alpha\big].
    \label{eq:Mn_star}
\end{equation}
The corresponding bootstrap variance estimator is given by:
\begin{equation}
    s_{n*}^{2}(\alpha, \gamma)
    :=
    \frac{\alpha(1-\alpha)}{n}\sum_{t=1}^{n}\omega_t^2\,\exp(2\gamma^{\top} I_{t-1}).
    \label{eq:sn_star}
\end{equation}
Mimicking the construction of the original test, the standardized bootstrap process
is defined as \(Q_{n*}(\alpha,\gamma):=\sqrt n\,|M_{n*}(\alpha,\gamma)|/
s_{n*}(\alpha,\gamma)\). The penalized bootstrap statistic uses the same
aggregation scheme as \(T_{n,\Pi}(\lambda)\):
\begin{equation}
    T_{n*,\Pi}(\lambda)
    :=
    \sup_{\gamma \in \Gamma}
    \left[
        \int_{\mathcal T} Q_{n*}^2(\alpha,\gamma)\,d\Pi(\alpha)
        -
        \lambda\|\gamma\|_1
    \right].
    \label{eq:Tn_star}
\end{equation}

The following theorem establishes the conditional weak convergence needed for bootstrap critical values to be asymptotically valid under $H_0$.

\begin{theorem}[Multiplier bootstrap validity]
\label{thm:bootstrap_main}
Suppose Assumptions~S1--S2 hold. Under the null hypothesis
$H_0:\theta_0(\alpha)=\bar{\theta}(\alpha)$ for all $\alpha\in\mathcal T$,
conditional on the data $\mathcal Z_n$, the following results hold in
probability:
\begin{enumerate}
    \item The bootstrap process converges weakly to the same limiting Gaussian process as the original empirical process:
    \begin{equation*}
        \sqrt{n} M_{n*}(\alpha,\gamma) \rightsquigarrow^* \mathcal{M}(\alpha,\gamma) \quad \text{in } \ell^\infty(\mathcal T\times\Gamma).
    \end{equation*}
    \item The bootstrap variance estimator is uniformly consistent:
    \begin{equation*}
        \sup_{\alpha\in\mathcal{T},\gamma\in\Gamma} \big|\, s_{n*}^2(\alpha,\gamma) - \sigma^2(\alpha, \gamma) \,\big| \xrightarrow{p^*} 0.
    \end{equation*}
    \item The bootstrap test statistic consistently approximates the null distribution:
    \begin{equation*}
        T_{n*,\Pi}(\lambda)
        \rightsquigarrow^*
        \sup_{\gamma \in \Gamma}
        \Bigg[
        \int_{\mathcal T}
        \left(
            \frac{\mathcal{M}(\alpha, \gamma)}
            {\sigma(\alpha, \gamma)}
        \right)^2
        d\Pi(\alpha)
        -
        \lambda\|\gamma\|_1
        \Bigg].
    \end{equation*}
\end{enumerate}
Here, $\rightsquigarrow^*$ and $\xrightarrow{p^*}$ denote weak convergence and convergence in probability conditional on the data, respectively.
\end{theorem}

The following corollary states explicitly the additional continuity condition
needed to translate bootstrap consistency into the convergence of the
rejection probability.

\begin{corollary}[Fixed penalty size control]
\label{cor:fixed-lambda-size_main}
Fix a deterministic $\lambda\ge0$, and let
$c_{n,1-\tau,\Pi}^*(\lambda)$ be the ideal conditional bootstrap
$(1-\tau)$-quantile. Under Theorem~\ref{thm:bootstrap_main}, if the distribution of
$T_\Pi(\lambda)$ is continuous at its $(1-\tau)$-quantile, then
\[
P\{T_{n,\Pi}(\lambda)>c_{n,1-\tau,\Pi}^*(\lambda)\}\to\tau.
\]
The same conclusion holds for the pre-estimated statistic under
Theorem~\ref{thm:boot-pre_main}, using its plug-in bootstrap critical value and
the corresponding continuity condition.
\end{corollary}

\subsection{Local Power Analysis}
\label{subsec:local-power}

To provide a theoretical foundation for selecting the penalty parameter
\(\lambda\), we study the behavior of the statistic under a
sequence of Pitman local alternatives. We parameterize the local sequence as
\begin{equation}
    H_{1,n}:\quad
    \theta_{n,B}(\alpha)
    =
    \theta_0(\alpha)
    -
    \frac{B(\alpha)}{\sqrt n},
    \qquad \alpha\in\mathcal T,
    \label{eq:local-alt-theta}
\end{equation}
 where \(B(\cdot)\in\mathcal B\) is a bounded deterministic function satisfying
\assLocal, which ensures that \(\theta_{n,B}(\cdot)\) defines a valid
conditional quantile model for all sufficiently large \(n\). Let \(P_{n,B}\) and \(E_{n,B}\) denote probability and expectation under the local sequence \eqref{eq:local-alt-theta}. The data are generated under \(P_{n,B}\), but the implemented statistic is still evaluated at the null.

Define
\[
Z_{nt}(\alpha,\gamma)
:=
\exp(\gamma^\top I_{t-1})
\Big[
    \mathbf 1\{Y_t\le m(I_{t-1},\theta_0(\alpha))\}
    -
    \alpha
\Big],
\]
and let
\[
    M_{n,0}^{B}(\alpha,\gamma)
    :=
    \frac1n\sum_{t=1}^n Z_{nt}(\alpha,\gamma)
\]
denote the empirical process under \(P_{n,B}\). The following
identity separates the random fluctuation from the deterministic local mean
shift:
\begin{equation}
\label{eq:local-centered-decomp}
\sqrt n\,M_{n,0}^{B}(\alpha,\gamma)
=
\mathbb G_{n,B}(\alpha,\gamma)
+
\frac1{\sqrt n}\sum_{t=1}^n E_{n,B}Z_{nt}(\alpha,\gamma),
\end{equation}
where
\[
\mathbb G_{n,B}(\alpha,\gamma)
:=
\frac1{\sqrt n}\sum_{t=1}^n
\Big\{
Z_{nt}(\alpha,\gamma)
-
E_{n,B}Z_{nt}(\alpha,\gamma)
\Big\}.
\]
Thus the expectation inside the centered empirical process is not an additional
drift term. It only centers the stochastic fluctuation. The local alternative
enters through the second term in \eqref{eq:local-centered-decomp}. To see how
the drift arises, note that under $H_{1,n}$ the local quantile at $\alpha$ is
\[
m(I_{t-1},\theta_{n,B}(\alpha)) \approx m(I_{t-1},\theta_0(\alpha)) -
n^{-1/2}\nabla_\theta m(I_{t-1},\theta_0(\alpha))^\top B(\alpha).
\]
Thus $F_{Y_t\mid I_{t-1}}\!\big(m(I_{t-1},\theta_0(\alpha))\big)$ differs from
$\alpha$ by approximately
\[
n^{-1/2}f_{Y_t\mid I_{t-1}}\!\big(m(I_{t-1},\theta_0(\alpha))\big)
\nabla_\theta m(I_{t-1},\theta_0(\alpha))^\top B(\alpha),
\]
uniformly over $(\alpha,\gamma)$. Taking the conditional expectation of $Z_{nt}$ under
$P_{n,B}$ and then the unconditional expectation therefore yields the drift
expansion below, derived in the Supplementary Appendix:
\begin{equation}
\label{eq:local-mean-expansion-main}
\frac1{\sqrt n}\sum_{t=1}^n E_{n,B}Z_{nt}(\alpha,\gamma)
=
d_B(\alpha,\gamma)+o(1),
\end{equation}
uniformly in \((\alpha,\gamma)\), where
\begin{equation}
    d_B(\alpha,\gamma)
    :=
    E\!\left[
    \exp(\gamma^\top I_{t-1})
    f_{Y_t|I_{t-1}}\!\left(m(I_{t-1},\theta_0(\alpha))\right)
    \nabla_\theta m(I_{t-1},\theta_0(\alpha))^\top B(\alpha)
    \right].
    \label{eq:local-drift}
\end{equation}
Under the local stability and stochastic equicontinuity conditions in
\assLocal, the centered process \(\mathbb G_{n,B}\) has the same first order
Gaussian limit as the null. Consequently,
\begin{equation}
\label{eq:local-expansion}
    \sqrt n\,M_{n,0}^{B}(\alpha,\gamma)
    \rightsquigarrow
    \mathcal M(\alpha,\gamma)+d_B(\alpha,\gamma)
    \qquad\text{in }\ell^\infty(\mathcal T\times\Gamma),
\end{equation}
where \(\mathcal M\) is the Gaussian process in
Theorem~\ref{thm:null_main}.

To characterize the limiting distribution of the statistic, define
\[
    \sigma(\alpha,\gamma)
    :=
    \sqrt{\alpha(1-\alpha)E[\exp(2\gamma^\top I_{t-1})]},
\]
and set
\begin{equation}
    Q(\alpha,\gamma)
    :=
    \frac{\mathcal M(\alpha,\gamma)}{\sigma(\alpha,\gamma)},
    \qquad
    R_B(\alpha,\gamma)
    :=
    \frac{d_B(\alpha,\gamma)}{\sigma(\alpha,\gamma)}.
    \label{eq:QR_def}
\end{equation}
By the continuous mapping theorem,
\begin{equation}
    T_{n,\Pi}(\lambda)
    \rightsquigarrow
    T_{\Pi}(\lambda;B)
    :=
    \sup_{\gamma\in\Gamma}
    \Bigg[
        \int_{\mathcal T}
        \big(Q(\alpha,\gamma)+R_B(\alpha,\gamma)\big)^2
        d\Pi(\alpha)
        -
        \lambda\|\gamma\|_1
    \Bigg].
    \label{eq:limit_local}
\end{equation}

We next examine how penalization changes the separation between the null and
local alternative along a given realization of the Gaussian
process.

\begin{theorem}[Effect of penalization on local separation]
\label{thm:power-enhance_main}
For simplicity, suppose the local drift direction is constant,
\(B(\alpha)\equiv B_0\). Define the unpenalized limit objective
\begin{equation*}
    J_\Pi(\gamma;B)
    :=
    \int_{\mathcal T}
    \big(Q(\alpha,\gamma)+R_B(\alpha,\gamma)\big)^2
    d\Pi(\alpha).
\end{equation*}
Fix a realization for which \(J_\Pi(\cdot;B)\) is continuous on the compact
set \(\Gamma\) and has a unique maximizer \(\tilde\gamma_\Pi(B)\) for both
\(B=0\) and some \(B_0\neq0\). If
\begin{equation*}
    \|\tilde\gamma_\Pi(0)\|_1-
    \|\tilde\gamma_\Pi(B_0)\|_1>0
\end{equation*}
holds on this realization, there exists a \(\delta(Q,B_0)>0\) such that
for every \(0<\lambda<\delta(Q,B_0)\), the penalized limit statistic increases
the separation
\[
    T_\Pi(\lambda;B_0)-T_\Pi(\lambda;0)
\]
relative to the unpenalized benchmark \(\lambda=0\).
\end{theorem}

Theorem~\ref{thm:power-enhance_main} compares the null and local alternative along the same realization of the Gaussian process. A small
penalty increases their difference when the unpenalized null maximizer has a
larger \(\ell_1\)-norm than the alternative maximizer. The result is pathwise and
does not compare rejection probabilities because the critical value also depends
on \(\lambda\).

\subsection{Choosing \texorpdfstring{$\lambda$}{lambda} via Local Bootstrap}
\label{subsec:calibration}

The adaptive rule requires an estimate of the local drift
$d_B(\alpha,\gamma)$, given by
\begin{equation}
    d_B(\alpha,\gamma) = E\Big[\exp(\gamma^\top I_{t-1})\, f_{Y_t|I_{t-1}}\big(m(I_{t-1}, \theta_0(\alpha))\big)\, \nabla_\theta m(I_{t-1}, \theta_0(\alpha))^\top B(\alpha)\Big].
    \label{eq:target_drift}
\end{equation}
Evaluating \eqref{eq:target_drift} requires the conditional density at the
fitted quantile, and we estimate it using the method described below.

\subsubsection{Conditional Density Estimation}

For the drift and variance corrections, we use the quantile density
estimator of \citet{escanciano2019}. It smooths over fitted quantile values
rather than the full conditioning vector $I_{t-1}$.

Let $\{\alpha_j\}_{j=1}^J$ be i.i.d.\ draws from $\mathrm{Unif}[a_1,a_2]$, independent of the data, where $\mathcal{T}\subset(a_1,a_2)$ and the distance from $\mathcal T$ to the boundary of $[a_1,a_2]$ is positive. The conditional density estimator is defined as:
\begin{equation}
    \widehat{f}_h(I_{t-1}, \bar\theta(\alpha)) := \frac{a_2-a_1}{J h} \sum_{j=1}^{J} K\left( \frac{m(I_{t-1}, \bar\theta(\alpha)) - m(I_{t-1}, \bar\theta(\alpha_j))}{h} \right),
    \label{eq:density_est}
\end{equation}
where $K(\cdot)$ is a symmetric kernel function and $h$ is the bandwidth.
The quantile curve must be correctly specified on the auxiliary
interval, as required by \assRegularity(iv). This is a maintained condition
for density estimation in addition to the restriction tested on $\mathcal T$.
The following theorem gives the uniform consistency result used below.

\begin{theorem}[Uniform consistency and rate of $\widehat{f}_h$]
\label{thm:density-consistency_main}
Under $H_0$ and Assumptions~S1--S2 and S3(a)--(c), if
$\mathcal{T}\subset(a_1,a_2)$ has positive distance from the boundary of
$[a_1,a_2]$, the bandwidth satisfies
$h \asymp (\frac{\log n}{n})^{1/5}$, and the simulation size $J \asymp n$,
then:
\[
    \sup_{\alpha \in \mathcal{T}}\max_{1\le t\le n} \Big| \widehat{f}_h(I_{t-1}, \bar\theta(\alpha)) - f_{Y_t|I_{t-1}}\big(m(I_{t-1}, \bar\theta(\alpha))\big) \Big| = o_p(1).
\]
\end{theorem}

The estimator smooths over the scalar quantile value
$m(I_{t-1},\bar\theta(\alpha))$ rather than over the full conditioning vector.
The kernel condition in \assKernel\ permits either a smooth compactly supported
kernel or a smooth kernel with exponential tail decay, including the Gaussian
kernel. The Supplementary Appendix establishes the more explicit bound
$O_p\{(\log n/n)^{2/5}\}$ by a Bernstein inequality and a covering argument.
Writing $a_{J_n}$ for the deterministic lower endpoint of the admissible
bandwidth interval in \assKernel(b), the canonical choice
$a_{J_n}\asymp(\log n/n)^{1/5}$ also gives
\[
n^{-1/2}a_{J_n}^{-2}
=n^{-1/10}(\log n)^{-2/5}\longrightarrow0.
\]
Under Assumptions~S4--S5, this rate makes the contribution of nuisance
estimation uniformly \(o_p(1)\) in the weighted sample averages used for
the pre-estimation correction.

\subsubsection{Bootstrap Validity under Local Alternatives}

Using the uniformly consistent density estimator \(\widehat f_h\), we
construct a shifted multiplier bootstrap process. It reproduces the noncentral
limit in \eqref{eq:local-expansion}. The multiplier component approximates the
Gaussian limit of the centered fluctuation \(\mathbb G_{n,B}\), while a deterministic sample
drift estimates \(d_B(\alpha,\gamma)\). Because
\(E^*(\omega_t\mid\mathcal Z_n)=0\), the multiplier term is already conditionally
centered and no explicit subtraction of \(E_{n,B}Z_{nt}\) is needed in the
bootstrap stochastic component.

Let \(\{\omega_t\}_{t=1}^n\) be i.i.d. standard normal multipliers independent of
the data. For a candidate local direction \(B\), define
\begin{align}
    M_{n*}(\alpha,\gamma;B)
    &:=
    \frac1n\sum_{t=1}^n
    \omega_t
    \exp(\gamma^\top I_{t-1})
    \big[\mathbf 1\{Y_t\le m(I_{t-1},\bar\theta(\alpha))\}-\alpha\big]
    \nonumber\\
    &\quad+
    \frac1n\sum_{t=1}^n
    \exp(\gamma^\top I_{t-1})
    \widehat f_h(I_{t-1},\bar\theta(\alpha))
    \frac{B(\alpha)^\top}{\sqrt n}
    \nabla_\theta m(I_{t-1},\bar\theta(\alpha)).
    \label{eq:shifted_M}
\end{align}
The second line is of order \(n^{-1/2}\) at the summand level. After multiplying
by \(\sqrt n\), it converges to the local drift \(d_B(\alpha,\gamma)\).

The bootstrap variance estimator is
\begin{equation}
    s_{n*,B}^2(\alpha,\gamma)
    :=
    \frac{\alpha(1-\alpha)}{n}
    \sum_{t=1}^n \omega_t^2\exp(2\gamma^\top I_{t-1}).
    \label{eq:shifted_s}
\end{equation}

\begin{remark}
If the original statistic is studentized by a sample second moment rather than
by the closed-form variance in \eqref{eq:shifted_s}, the shifted bootstrap
studentizer can be defined as the sample second moment of the shifted bootstrap
summands. The deterministic shift is \(O(n^{-1/2})\) at the summand level, so it
does not change the first order variance limit.
\end{remark}

Define
\begin{equation}
    Q_{n*,B}(\alpha,\gamma)
    :=
    \frac{\sqrt n\,|M_{n*}(\alpha,\gamma;B)|}{s_{n*,B}(\alpha,\gamma)}.
    \label{eq:Qn_shifted}
\end{equation}
The shifted bootstrap statistic is
\begin{equation}
    T_{n*,B,\Pi}(\lambda)
    :=
    \sup_{\gamma\in\Gamma}
    \Bigg[
        \int_{\mathcal T} Q_{n*,B}^2(\alpha,\gamma)d\Pi(\alpha)
        -
        \lambda\|\gamma\|_1
    \Bigg].
    \label{eq:Tn_shifted}
\end{equation}

The shifted bootstrap must estimate the power criterion for every candidate
direction $B$, whether the observed data are generated under the null or under the local alternative. We therefore distinguish the candidate
shift $B$ from the actual data-generating direction $B_0$, with $B_0=0$
denoting the null.

\begin{theorem}[Shifted bootstrap validity]
\label{thm:MB-local_main}
Suppose Assumptions~S1--S3 and S5--S6 hold. Let the bandwidth \(h_n\) satisfy the
conditions in Theorem~\ref{thm:density-consistency_main}. For any admissible
candidate direction $B$, under the null or any actual local sequence
$P_{n,B_0}$ covered by \assLocal, the following results hold conditional on
the data \(\mathcal Z_n\) in probability:
\begin{enumerate}
    \item The shifted bootstrap process replicates the noncentral limit process:
    \[
        \sqrt n\,M_{n*}(\alpha,\gamma;B)
        \rightsquigarrow^*
        \mathcal M(\alpha,\gamma)+d_B(\alpha,\gamma)
        \quad\text{in }\ell^\infty(\mathcal T\times\Gamma).
    \]

    \item The variance estimator remains uniformly consistent:
    \[
        \sup_{\alpha,\gamma}
        \big|s_{n*,B}^2(\alpha,\gamma)-\sigma^2(\alpha,\gamma)\big|
        \xrightarrow{p^*}0.
    \]

    \item The shifted bootstrap statistic converges to the corresponding
    noncentral limit:
    \[
        T_{n*,B,\Pi}(\lambda)
        \rightsquigarrow^*
        \sup_{\gamma\in\Gamma}
        \Bigg[
            \int_{\mathcal T}
            \left(
                \frac{\mathcal M(\alpha,\gamma)+d_B(\alpha,\gamma)}
                {\sigma(\alpha,\gamma)}
            \right)^2
            d\Pi(\alpha)
            -
            \lambda\|\gamma\|_1
        \Bigg].
    \]
\end{enumerate}
\end{theorem}

\subsubsection{Adaptive Penalty Selection Rule}

Using the shifted bootstrap approximation, we select the penalty parameter
$\lambda$ to maximize estimated power against the least favorable local
alternative in a prespecified class $\mathcal B$.

For theoretical purposes, let
\[
    \Lambda=[0,\bar\lambda]
\]
be a compact interval of candidate penalty values. Define
$\widehat W_n(\lambda)=\inf_{B\in\mathcal B}
\widehat{\mathcal R}_n(\lambda,B,\tau)$. When searching over the full interval $\Lambda$, we choose a penalty
$\widehat\lambda\in\Lambda$ whose criterion value is within $\eta_n$
of the supremum:
\begin{equation}
    \widehat W_n(\widehat\lambda)
    \ge \sup_{\lambda\in\Lambda}\widehat W_n(\lambda)-\eta_n,
    \qquad \eta_n\ge0,\quad \eta_n=o_p(1).
    \label{eq:maxmin}
\end{equation}
We assume that the selection rule is measurable. If the selected penalty
attains the supremum, we can set $\eta_n=0$. Here
\(\widehat{\mathcal R}_n(\lambda,B,\tau)\) denotes the shifted bootstrap estimate of the rejection probability against the local direction \(B\), using the null critical value corresponding to the same penalty \(\lambda\).

In computation, \(\Lambda\) may be searched directly or approximated by a finite
grid \(\Lambda_G\). A fixed grid targets the grid maximizer. It approximates the
continuous population maximizer when the grid mesh tends to zero and the
population criterion is continuous. We reuse the same
multiplier draws across penalties and directions to reduce simulation noise.
Algorithm~\ref{alg:adaptive-known} summarizes the procedure.

\begin{algorithm}[htbp]
\caption{Adaptive penalty selection without pre-estimation}
\label{alg:adaptive-known}
\begin{algorithmic}[1]
\Require Data \(\mathcal Z_n\), \(\Pi\), \(\Gamma\), search set
\(\mathcal S=\Lambda\) or \(\Lambda_G\), finite \(\mathcal B\), \(R\), \(\tau\),
and density-smoothing choices \((K,h,J,[a_1,a_2])\)
\State Estimate the conditional density by \eqref{eq:density_est} and, for each
\(B\in\mathcal B\), compute the sample drift
\(\widehat d_{n,B}=n^{-1}\sum_t e^{\gamma^\top I_{t-1}}
\widehat f_h(I_{t-1},\bar\theta(\alpha))
\nabla_\theta m(I_{t-1},\bar\theta(\alpha))^\top B(\alpha)\)
used in \eqref{eq:shifted_M}.
\State Draw \(R\) multiplier sequences
\(\Omega^{(r)}=\{\omega_t^{(r)}\}_{t=1}^n\) and reuse them throughout.
\For{each penalty \(\lambda\) requested by the search routine}
    \For{\(r=1,\ldots,R\)}
        \State Compute the null bootstrap statistic
        \(T_{n*,\Pi}^{(r)}(\lambda)\).
    \EndFor
    \State Set \(c_{n,1-\tau,\Pi}^*(\lambda)\) to their empirical
    \((1-\tau)\)-quantile.
    \For{each \(B\in\mathcal B\)}
        \State Compute \(T_{n*,B,\Pi}^{(r)}(\lambda)\), \(r=1,\ldots,R\),
        using the shifted bootstrap process.
        \State Set
        \(\widehat{\mathcal R}_n(\lambda,B,\tau)
        =R^{-1}\sum_{r=1}^R
        \mathbf 1\{T_{n*,B,\Pi}^{(r)}(\lambda)>
        c_{n,1-\tau,\Pi}^*(\lambda)\}\).
    \EndFor
    \State Set \(\widehat W_n(\lambda)=
    \min_{B\in\mathcal B}\widehat{\mathcal R}_n(\lambda,B,\tau)\).
\EndFor
\State On \(\Lambda_G\), select the smallest maximizer of \(\widehat W_n\).
For continuous search, use a measurable \(\eta_n\)-approximate maximizer,
\(\eta_n=o_p(1)\).
\State Using the original sample, compute \(T_{n,\Pi}(\widehat\lambda)\) and reject when it exceeds
\(c_{n,1-\tau,\Pi}^*(\widehat\lambda)\).
\end{algorithmic}
\end{algorithm}

For a grid search, the selection line in Algorithm~\ref{alg:adaptive-known}
reduces to
\begin{equation}
    \widehat{\lambda}_G
    =
    \arg\max_{\lambda \in \Lambda_G}
    \min_{B \in \mathcal{B}}
    \widehat{\mathcal{R}}_n(\lambda, B, \tau),
    \label{eq:maxmin-grid}
\end{equation}
again choosing the smallest maximizer. For a fixed finite grid, this rule targets
the grid optimum rather than the continuous optimum. Alternatively, the criterion
may be optimized directly over \([0,\bar\lambda]\) using a derivative-free method
such as particle swarm optimization (PSO).

The test rejects when \(T_{n,\Pi}(\widehat\lambda)\) exceeds
\(c_{n,1-\tau,\Pi}^*(\widehat\lambda)\). Here \(\widehat\lambda\) denotes the
continuous selector, while \(\widehat\lambda_G\) denotes the grid
selector. For each candidate penalty, the theory uses the $(1-\tau)$-quantile
of the bootstrap distribution conditional on the observed data.
In practice, we estimate this quantile from a finite number of
bootstrap draws. The validity result below also applies to these
estimates if the largest absolute error across the candidate
penalties converges to zero in probability.

A sufficient condition for post-selection validity is that the selected
penalty converges to a deterministic limit. This condition follows, for example,
when the estimated maximin criterion converges uniformly to a population
criterion with a unique maximizer. The following corollary states the
continuous-\(\Lambda\) result and also applies to the plug-in statistic in
Section~\ref{sec:pre_estimation}.

\begin{corollary}[Validity for adaptive penalty selection]
\label{cor:adaptive-lambda_main}
Let \(\Lambda=[0,\bar\lambda]\) be compact. Suppose the relevant bootstrap approximation in Theorem~\ref{thm:bootstrap_main} holds uniformly over \(\lambda\in\Lambda\), and let \(c_{n,1-\tau,\Pi}^*(\lambda)\) denote the ideal conditional bootstrap critical value, or a Monte Carlo approximation with simulation error \(o_p(1)\) uniformly in \(\lambda\). Let \(c_{1-\tau,\Pi}(\lambda)\) be the corresponding \((1-\tau)\)-quantile of the null limit statistic. If
\[
    \widehat\lambda\xrightarrow{p}\lambda_0
    \qquad\text{for some deterministic }\lambda_0\in\Lambda,
\]
and the limiting distribution is continuous at \(c_{1-\tau,\Pi}(\lambda_0)\) with \(\lambda\mapsto c_{1-\tau,\Pi}(\lambda)\) continuous at \(\lambda_0\), then under \(H_0\),
\[
P\{T_{n,\Pi}(\widehat\lambda)>c_{n,1-\tau,\Pi}^*(\widehat\lambda)\}
\to\tau.
\]
The same conclusion holds for the pre-estimated statistic
\(\widehat T_{n,\Pi}\) under Theorem~\ref{thm:boot-pre_main}, with its
corresponding plug-in bootstrap critical value.
\end{corollary}

The preceding corollary establishes size control once the selector stabilizes.
The next theorem gives sufficient conditions for this convergence and compares
maximin local power with the unpenalized test because the search set contains
\(\lambda=0\). Unlike Theorem~\ref{thm:power-enhance_main}, it compares rejection
probabilities using the null critical value corresponding to each penalty.

For \(\lambda\in\Lambda\), let
\[
    c_{1-\tau,\Pi}(\lambda)
    :=q_{1-\tau}\{T_\Pi(\lambda;0)\}
\]
and define the local asymptotic rejection probability
\[
    \mathcal R(\lambda,B,\tau)
    :=P\{T_\Pi(\lambda;B)>c_{1-\tau,\Pi}(\lambda)\}.
\]

\begin{theorem}[Maximin local power of the adaptive selector]
\label{thm:maximin-no-loss_main}
Let \(\Lambda=[0,\bar\lambda]\) and suppose \(0\in\Lambda\). Let
\(\mathcal B=\{B_1,\ldots,B_J\}\) be a finite, nonempty collection of local
directions. Under Assumptions~S1--S3 and S5--S6, suppose the conclusions of
Theorems~\ref{thm:bootstrap_main} and
\ref{thm:MB-local_main} hold for every fixed \(\lambda\in\Lambda\) and every
\(B\in\mathcal B\), with stochastic equicontinuity of the processes
holding uniformly over this finite collection. Suppose also that
\[
\lim_{\varepsilon\downarrow0}
\max_{B\in\mathcal B\cup\{0\}}
\sup_{\lambda\in\Lambda}
P\!\left(
\left|T_\Pi(\lambda;B)-c_{1-\tau,\Pi}(\lambda)\right|
\le\varepsilon
\right)=0.
\]
Assume that the Monte Carlo errors in the estimated critical values and
rejection probabilities are \(o_p(1)\) uniformly in \(\lambda\), and define
\[
W(\lambda):=\min_{B\in\mathcal B}\mathcal R(\lambda,B,\tau).
\]
If \(W\) is continuous on \(\Lambda\) and has a unique maximizer
\(\lambda^*\), then
\[
\sup_{\lambda\in\Lambda,\,B\in\mathcal B}
\left|
\widehat{\mathcal R}_n(\lambda,B,\tau)
-\mathcal R(\lambda,B,\tau)
\right|\xrightarrow{p}0,
\qquad
\widehat\lambda\xrightarrow{p}\lambda^*.
\]
Moreover,
\[
\min_{B\in\mathcal B}\mathcal R(\lambda^*,B,\tau)
=\max_{\lambda\in\Lambda}
\min_{B\in\mathcal B}\mathcal R(\lambda,B,\tau)
\ge
\min_{B\in\mathcal B}\mathcal R(0,B,\tau).
\]
For every \(B\in\mathcal B\), under the corresponding local
sequence,
\[
P_{n,B}\!\left\{
T_{n,\Pi}(\widehat\lambda)
>c_{n,1-\tau,\Pi}^*(\widehat\lambda)
\right\}
\longrightarrow
\mathcal R(\lambda^*,B,\tau).
\]
Consequently, if \(\mathcal B=\{B^*\}\) is a singleton, then
\(\mathcal R(\lambda^*,B^*,\tau)\ge
\mathcal R(0,B^*,\tau)\). For a collection containing more than one direction,
it concerns worst case local power.
\end{theorem}

The anti-concentration condition requires the probability that
\(T_\Pi(\lambda;B)\) lies near \(c_{1-\tau,\Pi}(\lambda)\) to vanish
uniformly over \(\lambda\) and \(B\) as the neighborhood shrinks. It is
satisfied, for example, if the distributions of \(T_\Pi(\lambda;B)\) have
densities that are uniformly bounded in a neighborhood of the corresponding
critical values. On a fixed finite penalty grid, pointwise continuity at each
critical value is sufficient.

\begin{remark}[Penalty selection on a fixed grid]
\label{rem:grid-no-loss_main}
Let $\Lambda_G$ be a fixed finite grid containing zero, and let
$\lambda_G^*$ maximize $W$ over $\Lambda_G$. Since zero is a candidate,
\[
W(\lambda_G^*)=\max_{\lambda\in\Lambda_G}W(\lambda)\ge W(0).
\]
If $\widehat W_n$ converges uniformly to $W$ on $\Lambda_G$ and
$\lambda_G^*$ is unique, the selected grid value converges to
$\lambda_G^*$. Thus its asymptotic minimum local power over
$\mathcal B$ is at least that of the unpenalized test. This comparison
does not require the grid to approximate the optimum over $\Lambda$.
\end{remark}

\section{Inference with Pre-Estimated Parameters}
\label{sec:pre_estimation}

We now allow the parameter vector to contain unknown nuisance components that
are estimated before the test statistic is constructed. Write
$\theta(\alpha)=(\theta_1(\alpha)^\top,\theta_2(\alpha)^\top)^\top$, where
$\theta_1(\alpha)$ is the interest and $\theta_2(\alpha)$ is a
nuisance parameter. We test the first component uniformly over the quantile
index:
\begin{equation}
    H_0: \theta_{10}(\alpha) = \bar{\theta}_1(\alpha), \quad \forall \alpha \in \mathcal{T},
\end{equation}
while treating the true nuisance parameter $\theta_{20}(\alpha)$ as unknown. Under this null hypothesis, we evaluate the conditional moment restriction by replacing $\theta_{20}(\alpha)$ with a preliminary $\sqrt{n}$-consistent estimator.

\subsection{Test Statistic with Nuisance Parameters}

Let
\[
m_{t\alpha}(\theta_1,\theta_2)
:=m(I_{t-1},\theta_1(\alpha),\theta_2(\alpha))
\]
denote the conditional quantile function, and write its nuisance gradient as
\[
g_{t\alpha}(\theta_1,\theta_2)
:=\nabla_{\theta_2}m(I_{t-1},\theta_1(\alpha),\theta_2(\alpha)).
\]

We estimate $\theta_{20}(\alpha)$ under the null restriction
$\theta_{10}(\alpha)=\bar{\theta}_1(\alpha)$. The theory requires the preliminary
estimator $\hat{\theta}_2(\alpha)$ to be $\sqrt n$-consistent and admit a
uniform linear expansion in terms of its influence function. As an example, consider the
restricted moment estimator that solves
\begin{equation}
    \frac{1}{n} \sum_{t=1}^n g_{t\alpha}(\bar{\theta}_1(\alpha), \hat{\theta}_2(\alpha)) \, 
    \Big[ \mathbf{1}\big\{Y_t \le m_{t\alpha}(\bar{\theta}_1(\alpha), \hat{\theta}_2(\alpha))\big\} - \alpha \Big] = 0,
    \label{eq:theta2-est}
\end{equation}
where $g_{t\alpha}$ is the gradient defined above, and the equation should be
interpreted as the subgradient/first-order condition of the restricted quantile
objective. An alternative estimator may be used if it admits the required
uniform linear expansion and has a consistently influence function.
The corrected score conditions in \assPreEst\ must also hold.

We impose a uniform linear representation for the restricted quantile
estimator. Such representations are standard for quantile regression. See
\citet{gutenbrunner1992}, \citet{mukherjee1999}, and the dynamic quantile
framework of \citet{escanciano2010}. Specifically, we assume that
$\hat\theta_2(\alpha)$ has the following expansion uniformly over $\mathcal T$:
\begin{equation}
    \sup_{\alpha \in \mathcal{T}} \left\| \sqrt{n}\big(\hat{\theta}_2(\alpha) - \theta_{20}(\alpha)\big) - \frac{1}{\sqrt{n}}\sum_{t=1}^n l_{t,\alpha}(\bar{\theta}_1, \theta_{20}) \right\| = o_p(1).
    \label{eq:theta2-IF}
\end{equation}
For the restricted quantile estimator in \eqref{eq:theta2-est}, the familiar expansion gives the influence function
\begin{equation}
    l_{t,\alpha} := -L_\alpha^{-1} \, g_{t\alpha}(\bar{\theta}_1, \theta_{20}) \, \Big[\mathbf{1}\{Y_t \le m_{t\alpha}(\bar{\theta}_1, \theta_{20})\} - \alpha\Big],
    \label{eq:IF-def}
\end{equation}
where $L_\alpha := E[g_{t\alpha}(\bar{\theta}_1, \theta_{20}) g_{t\alpha}^\top(\bar{\theta}_1, \theta_{20}) f_{Y_t|I_{t-1}}(m_{t\alpha}(\bar{\theta}_1, \theta_{20}))]$ is the expected Jacobian. At the true parameters, the influence function is a martingale difference sequence under the maintained conditions. Hence it has mean zero and is orthogonal across time. In particular,
$E[l_{t,\alpha}\{\mathbf 1(Y_s\le m_{s\alpha})-\alpha\}]=0$ for
$t\ne s$.
In the linear quantile regression model, $g_{t\alpha}$ is the vector
of nuisance regressors, and $L_\alpha$ is the expected outer product weighted
by the conditional density. For the plug-in statistic, centering makes the standardized process
undefined at $\gamma=0$. We therefore take $\Gamma$ to be compact
and separated from zero.

Let $\bar{e}_n(\gamma) := \frac{1}{n}\sum_{s=1}^n \exp(\gamma^\top I_{s-1})$ denote the sample mean of the weights, and define the {demeaned weight function} by
\begin{equation}
    w_{t,n}(\gamma) := \exp(\gamma^\top I_{t-1}) - \bar{e}_n(\gamma).
    \label{eq:weight-def}
\end{equation}
The {plug-in empirical process} is then constructed using the estimated nuisance parameters and the demeaned weights:
\begin{equation}
    \widehat M_n(\alpha, \gamma) := \frac{1}{n}\sum_{t=1}^n w_{t,n}(\gamma) \Big[ \mathbf{1}\{Y_t\le m_{t\alpha}(\bar{\theta}_1, \hat{\theta}_2)\} - \alpha \Big].
    \label{eq:Mhat-def}
\end{equation}

By stochastic equicontinuity and a Taylor expansion of the population moment,
the plug-in process admits the following asymptotic expansion under $H_0$.
Let
\[
w_t(\gamma)
:=
\exp(\gamma^\top I_{t-1})
-
E\!\left[\exp(\gamma^\top I_{t-1})\right]
\]
denote the centered weight, and define the corresponding null process
\[
M_n^c(\alpha,\gamma)
:=
\frac1n\sum_{t=1}^n
w_t(\gamma)
\Big[
\mathbf 1\{Y_t\le m_{t\alpha}(\theta_{10},\theta_{20})\}
-
\alpha
\Big].
\]
Furthermore, define
\[
A_2(\alpha,\gamma)
:=
E\!\left[
w_t(\gamma)
f_{Y_t\mid I_{t-1}}
\!\left(
m_{t\alpha}(\theta_{10},\theta_{20})
\right)
\nabla_{\theta_2}
m_{t\alpha}(\theta_{10},\theta_{20})
\right].
\]
Under $H_0$ and Assumptions~S1--S2 and S4--S5, we have, uniformly over
$(\alpha,\gamma)\in\mathcal T\times\Gamma$,
\begin{equation}
\sqrt{n}\,\widehat M_n(\alpha,\gamma)
=
\sqrt{n}\,M_n^c(\alpha,\gamma)
+
A_2(\alpha,\gamma)^\top
\sqrt{n}\{\widehat\theta_2(\alpha)-\theta_{20}(\alpha)\}
+
o_p(1).
\label{eq:Mhat-expansion}
\end{equation}
This expansion separates the empirical process from the additional term induced by nuisance estimation.

Define the population corrected score
\[
\Psi_t(\alpha,\gamma)
:=
w_t(\gamma)
\big[\mathbf 1\{Y_t\le m_{t\alpha}(\theta_{10},\theta_{20})\}-\alpha\big]
+l_{t,\alpha}^{\top}A_2(\alpha,\gamma).
\]
Under \assPreEst(iii), this score is serially uncorrelated, so its long run
variance is
\(\tilde\sigma^2(\alpha,\gamma)=E[\Psi_t(\alpha,\gamma)^2]\).
We estimate this variance using corrected scores. Here and below,
\[
\widehat f_{Y_t\mid I_{t-1}}
\!\left(m_{t\alpha}(\bar\theta_1,\widehat\theta_2)\right)
\]
denotes the fitted quantile density estimator \(\widehat f_h\) from
Section~\ref{sec:bootstrap_power}, evaluated along the plug-in
quantile curve \(u\mapsto m_{tu}(\bar\theta_1,\widehat\theta_2)\). The
corresponding sample Jacobian is
\[
\widehat L_\alpha
:=
\frac1n\sum_{t=1}^n
\widehat f_{Y_t\mid I_{t-1}}
\!\left(m_{t\alpha}(\bar\theta_1,\widehat\theta_2)\right)
 g_{t\alpha}(\bar\theta_1,\widehat\theta_2)
 g_{t\alpha}(\bar\theta_1,\widehat\theta_2)^\top .
\]
Using this density estimate and the Jacobian, define
\[
    \widehat l_{t,\alpha} := - \widehat L_\alpha^{-1} \, g_{t\alpha}(\bar{\theta}_1, \hat{\theta}_2) \, \Big[\mathbf{1}\{Y_t \le m_{t\alpha}(\bar{\theta}_1, \hat{\theta}_2)\} - \alpha\Big],
\]
\[
    \widehat A_{2,n}(\alpha,\gamma) := \frac{1}{n}\sum_{t=1}^n \widehat f_{Y_t\mid I_{t-1}}(m_{t\alpha}(\bar{\theta}_1,\hat{\theta}_2))\, w_{t,n}(\gamma)\, \nabla_{\theta_2} m\!\big(I_{t-1},\bar{\theta}_1(\alpha),\hat\theta_2(\alpha)\big).
\]

Also define the {corrected variance estimator} as the sample second moment of
the estimated corrected scores:
\begin{equation}
\begin{split}
    \widehat s_n^2(\alpha, \gamma) := \frac{1}{n} \sum_{t=1}^n \Bigg( &
        w_{t,n}(\gamma) \Big[ \mathbf{1}\{Y_t\le m_{t\alpha}(\bar{\theta}_1, \hat{\theta}_2)\} - \alpha \Big] \\
        & + \widehat l_{t,\alpha}^\top \widehat A_{2,n}(\alpha,\gamma)
    \Bigg)^2.
\end{split}
\label{eq:shat-def}
\end{equation}

Then the standardized plug-in process is
\[
\widehat Q_n(\alpha, \gamma) := \frac{\sqrt{n}\,|\widehat M_n(\alpha, \gamma)|}{\widehat s_n(\alpha, \gamma)}.
\]
Finally, the {penalized plug-in test statistic} using the
        same CvM--KS aggregation scheme is
\begin{equation}
    \widehat T_{n,\Pi}(\lambda)
    :=
    \sup_{\gamma \in \Gamma}
    \left[
        \int_{\mathcal T}
        \widehat Q_n^2(\alpha,\gamma)
        d\Pi(\alpha)
        -
        \lambda\|\gamma\|_1
    \right].
    \label{eq:That-def}
\end{equation}
The next theorem gives the null limit.

\begin{theorem}[Pre-estimation null limit]
\label{thm:null-pre_main}
Under $H_0$ and Assumptions~S1--S5, the following results hold:
\begin{enumerate}
    \item The plug-in process converges in
    $\ell^\infty(\mathcal T\times\Gamma)$:
    \[
        \sqrt n\,\widehat M_n(\alpha,\gamma)
        \rightsquigarrow
        \widetilde{\mathcal M}(\alpha,\gamma)
        :=\mathcal M^c(\alpha,\gamma)
        +A_2(\alpha,\gamma)^\top\mathcal G_{\theta_2}(\alpha).
    \]
    Here $(\mathcal M^c,\mathcal G_{\theta_2})$ is the joint Gaussian limit of
    $\big(\sqrt n\,M_n^c,\,
    n^{-1/2}\sum_{t=1}^n l_{t,\cdot}\big)$.
    The covariance kernel of $\widetilde{\mathcal M}$ is
    $E[\Psi_t(\alpha_1,\gamma_1)\Psi_t(\alpha_2,\gamma_2)]$.

    \item The corrected variance estimator is uniformly consistent:
    \[
        \sup_{\alpha \in \mathcal{T}, \gamma \in \Gamma}
        \big|\, \widehat s_n^2(\alpha, \gamma)
        - \tilde{\sigma}^2(\alpha, \gamma) \,\big|
        \xrightarrow{p} 0,
    \]
    where
    $\tilde{\sigma}^2(\alpha,\gamma)
    =E[\Psi_t(\alpha,\gamma)^2]
    =\mathrm{Var}\{\widetilde{\mathcal M}(\alpha,\gamma)\}$.

    \item The test statistic therefore converges in distribution to
    \[
        \widehat T_{n,\Pi}(\lambda)
        \rightsquigarrow
        \sup_{\gamma \in \Gamma}
        \Bigg[
            \int_{\mathcal T}
            \left(
                \frac{\widetilde{\mathcal{M}}(\alpha,\gamma)}
                {\tilde{\sigma}(\alpha,\gamma)}
            \right)^2
            d\Pi(\alpha)
            -
            \lambda\|\gamma\|_1
        \Bigg].
    \]
\end{enumerate}
\end{theorem}

\subsection{Bootstrap Critical Values with Pre-Estimation}

To approximate the null distribution of \(\widehat T_{n,\Pi}(\lambda)\), we keep
\(\hat\theta_2(\alpha)\) fixed and bootstrap the corrected score in the expansion
above.

Let $\{\omega_t\}_{t=1}^n$ be i.i.d.\ standard normal multipliers independent of
the data. Define the bootstrap analog of the plug-in empirical process by
\begin{equation}
\begin{split}
\widehat M_{n*}(\alpha,\gamma)
:= \frac{1}{n}\sum_{t=1}^n \omega_t \Bigg\{ &
\big[\mathbf 1\{Y_t\le m_{t\alpha}(\bar{\theta}_1,\hat\theta_2)\}-\alpha\big] w_{t,n}(\gamma) \\
& + \widehat l_{t,\alpha}^\top\, \widehat A_{2,n}(\alpha,\gamma)
\Bigg\}.
\end{split}
\label{eq:Mhat-star-def}
\end{equation}
Here $\widehat l_{t,\alpha}$ and
$\widehat A_{2,n}(\alpha,\gamma)$ are the estimates defined above.

Also define the bootstrap variance estimator by
\begin{equation}
\begin{split}
\widehat s_{n*}^2(\alpha,\gamma)
:= \frac{1}{n}\sum_{t=1}^n \omega_t^2 \Bigg\{ &
\big[\mathbf 1\{Y_t\le m_{t\alpha}(\bar{\theta}_1,\hat\theta_2)\}-\alpha\big] w_{t,n}(\gamma) \\
& + \widehat l_{t,\alpha}^\top\, \widehat A_{2,n}(\alpha,\gamma)
\Bigg\}^2.
\end{split}
\label{eq:shat-star-def}
\end{equation}

The standardized bootstrap process is
\[
\widehat Q_{n*}(\alpha,\gamma)
:=
\frac{\sqrt n|\widehat M_{n*}(\alpha,\gamma)|}
{\widehat s_{n*}(\alpha,\gamma)}.
\]
The penalized bootstrap statistic is then constructed using the same
aggregation scheme as the original test:
\begin{equation}
\widehat T_{n*,\Pi}(\lambda)
:=
\sup_{\gamma\in\Gamma}
\Bigg[
\int_{\mathcal T}
\widehat Q_{n*}^2(\alpha,\gamma)
d\Pi(\alpha)
-
\lambda\|\gamma\|_1
\Bigg].
\label{eq:That-star-def}
\end{equation}

\begin{theorem}[Bootstrap validity with pre-estimation]
\label{thm:boot-pre_main}
\leavevmode\par\noindent
Under $H_0$ and Assumptions~S1--S5, the following conditional bootstrap
results hold, with the process convergence in
$\ell^\infty(\mathcal T\times\Gamma)$:
\begin{enumerate}

\item The bootstrap process converges to the limiting process:
\[
\sqrt n\,\widehat M_{n*}(\alpha,\gamma)
\ \overset{*}{\rightsquigarrow}\
\widetilde{\mathcal M}(\alpha,\gamma),
\]
where $\widetilde{\mathcal M}(\alpha,\gamma)$ is the same limiting Gaussian process as in Theorem~\ref{thm:null-pre_main}.

\item The bootstrap variance estimator consistently estimates the total asymptotic variance:
\[
\sup_{\alpha,\gamma}
\Big|
\widehat s_{n*}^2(\alpha,\gamma)
- \tilde{\sigma}^2(\alpha, \gamma)
\Big|
\ \xrightarrow{p^*}\ 0.
\]

\item The bootstrap statistic converges to the null limit:
\[
\widehat T_{n*,\Pi}(\lambda)
\ \overset{*}{\rightsquigarrow}\
\sup_{\gamma\in\Gamma}
\Bigg[
\int_{\mathcal T}
\left(
    \frac{\widetilde{\mathcal M}(\alpha,\gamma)}
    {\tilde{\sigma}(\alpha,\gamma)}
\right)^2
d\Pi(\alpha)
-
\lambda\|\gamma\|_1
\Bigg].
\]

\end{enumerate}
Here \(\xrightarrow{p^*}\) and \(\overset{*}{\rightsquigarrow}\) denote convergence in probability and weak convergence conditional on the data, respectively.
\end{theorem}

\subsection{Local Power with Pre-Estimation}

We now consider Pitman local alternatives for the parameter of interest:
\begin{equation}
\theta_{1n,B}(\alpha)
=
\theta_{10}(\alpha)
-
\frac{B(\alpha)}{\sqrt n},
\qquad \alpha\in\mathcal T,
\label{eq:pre-local-alt}
\end{equation}
where $B$ is an admissible bounded continuous direction in the tangent set of
\assLocal, so the local sequence remains a valid conditional quantile model for
all sufficiently large $n$. The data are generated under
\eqref{eq:pre-local-alt}, but the implemented test is still constructed under
the null value \(\bar\theta_1(\alpha)=\theta_{10}(\alpha)\). Hence the
nuisance parameter is estimated under the null restriction by solving
\eqref{eq:theta2-est} with \(\theta_1(\alpha)\) fixed at \(\theta_{10}(\alpha)\).
We denote this restricted estimator under the local sequence by
\[
    \hat\theta_{2,0n}(\alpha):=\hat\theta_2(\theta_{10},\alpha).
\]
The null plug-in process under the local sequence is
\[
\widehat M_{n,0}^{B}(\alpha,\gamma)
=
\frac1n\sum_{t=1}^n
w_{t,n}(\gamma)
\Big[
    \mathbf 1\{Y_t\le m_{t\alpha}(\theta_{10},\hat\theta_{2,0n})\}
    -
    \alpha
\Big].
\]

For notational clarity, write
\[
U_{nt}(\alpha,\gamma)
:=
 w_t(\gamma)
 \Big[
 \mathbf 1\{Y_t\le m_{t\alpha}(\theta_{10},\theta_{20})\}-\alpha
 \Big],
\]
and
\[
V_{nt}(\alpha)
:=
 g_{t\alpha}(\theta_{10},\theta_{20})
 \Big[
 \mathbf 1\{Y_t\le m_{t\alpha}(\theta_{10},\theta_{20})\}-\alpha
 \Big],
\]
where replacing the sample centered weight \(w_{t,n}\) by its population version
\(w_t\) only changes the expansion by \(o_p(1)\). Define
\[
A_1(\alpha,\gamma)
:=
E\!\left[
 f_{Y_t|I_{t-1}}\!\left(m_{t\alpha}(\theta_{10},\theta_{20})\right)
 w_t(\gamma)
 \nabla_{\theta_1}m_{t\alpha}(\theta_{10},\theta_{20})
\right],
\]

and
\[
D_\alpha
:=
E\!\left[
 g_{t\alpha}(\theta_{10},\theta_{20})
 f_{Y_t|I_{t-1}}\!\left(m_{t\alpha}(\theta_{10},\theta_{20})\right)
 \nabla_{\theta_1}^{\top}m_{t\alpha}(\theta_{10},\theta_{20})
\right],
\qquad
H_\alpha:=-L_\alpha^{-1}D_\alpha.
\]

The corresponding local mean expansions derived in the Appendix give
\[
\frac1{\sqrt n}\sum_{t=1}^nE_{n,B}U_{nt}(\alpha,\gamma)
=
A_1(\alpha,\gamma)^\top B(\alpha)+o(1),
\]
and
\[
\frac1{\sqrt n}\sum_{t=1}^nE_{n,B}V_{nt}(\alpha)
=
D_\alpha B(\alpha)+o(1),
\]
uniformly in \((\alpha,\gamma)\). The restricted nuisance estimator satisfies
\begin{equation}
\label{eq:delta_2n_full}
\Delta_{2n}(\alpha)
:=
\sqrt n\{\hat\theta_{2,0n}(\alpha)-\theta_{20}(\alpha)\}
=
\mathbb L_{n,B}(\alpha)+H_\alpha B(\alpha)+o_p(1),
\end{equation}
where
\[
\mathbb L_{n,B}(\alpha)
:=
-L_\alpha^{-1}
\frac1{\sqrt n}\sum_{t=1}^n
\Big\{
V_{nt}(\alpha)-E_{n,B}V_{nt}(\alpha)
\Big\}.
\]
The first term is the centered nuisance estimation fluctuation, while
\(H_\alpha B(\alpha)\) is the local bias induced by estimating the nuisance
parameter under the null restriction.

A Taylor expansion of the plug-in process around \((\theta_{10},\theta_{20})\)
then yields
\begin{equation}
\label{eq:pre-local-expansion}
\sqrt n\,\widehat M_{n,0}^{B}(\alpha,\gamma)
=
\mathbb G_{n,B}^{\mathrm{pre}}(\alpha,\gamma)
+
d_{\mathrm{pre},B}(\alpha,\gamma)
+o_p(1),
\end{equation}
where the centered stochastic component is
\[
\mathbb G_{n,B}^{\mathrm{pre}}(\alpha,\gamma)
:=
\frac1{\sqrt n}\sum_{t=1}^n
\Big\{
U_{nt}(\alpha,\gamma)-E_{n,B}U_{nt}(\alpha,\gamma)
\Big\}
+
A_2(\alpha,\gamma)^\top\mathbb L_{n,B}(\alpha),
\]
and the deterministic local drift is
\begin{equation}
\label{eq:pre-local-drift}
d_{\mathrm{pre},B}(\alpha,\gamma)
:=
A_1(\alpha,\gamma)^\top B(\alpha)
+
A_2(\alpha,\gamma)^\top H_\alpha B(\alpha).
\end{equation}
The drift has two components: the direct term
\(A_1(\alpha,\gamma)^\top B(\alpha)\) and the additional term
\(A_2(\alpha,\gamma)^\top H_\alpha B(\alpha)\) induced by restricted nuisance
estimation. Under contiguity and the weak convergence assumptions in
Assumptions~S4--S6,
\begin{equation}
\label{eq:final_decomposition}
\sqrt n\,\widehat M_{n,0}^{B}(\alpha,\gamma)
\rightsquigarrow
\widetilde{\mathcal M}(\alpha,\gamma)
+
d_{\mathrm{pre},B}(\alpha,\gamma),
\end{equation}
where \(\widetilde{\mathcal M}\) is the Gaussian process with mean zero from
Theorem~\ref{thm:null-pre_main}. The deterministic drift
\(d_{\mathrm{pre},B}\) is the component that drives local power after
pre-estimation.

\subsection{Construction of the Shifted Bootstrap Process}
\label{subsec:shifted-bootstrap}

To select \(\lambda\) with pre-estimated nuisance parameters, the bootstrap must
approximate the noncentral limit in \eqref{eq:final_decomposition}. The multiplier
term reproduces the centered Gaussian component \(\widetilde{\mathcal M}\), while
a deterministic shift estimates \(d_{\mathrm{pre},B}(\alpha,\gamma)\).

Let \(\widehat d_{\mathrm{pre},n,B}(\alpha,\gamma)\) denote a uniformly
consistent sample analog of \(d_{\mathrm{pre},B}(\alpha,\gamma)\). Define
\[
\widehat A_{1,n}(\alpha,\gamma)
:=
\frac1n\sum_{t=1}^n
\widehat f_{Y_t|I_{t-1}}
\!\left(m_{t\alpha}(\bar\theta_1,\hat\theta_2)\right)
 w_{t,n}(\gamma)
 \nabla_{\theta_1}m_{t\alpha}(\bar\theta_1,\hat\theta_2),
\]
\[
\widehat D_{\alpha,n}
:=
\frac1n\sum_{t=1}^n
 g_{t\alpha}(\bar\theta_1,\hat\theta_2)
 \widehat f_{Y_t|I_{t-1}}
 \!\left(m_{t\alpha}(\bar\theta_1,\hat\theta_2)\right)
 \nabla_{\theta_1}^{\top}m_{t\alpha}(\bar\theta_1,\hat\theta_2),
\qquad
\widehat H_{\alpha,n}:=-\widehat L_\alpha^{-1}\widehat D_{\alpha,n}.
\]
We set
\begin{equation}
\label{eq:dpre-hat}
\widehat d_{\mathrm{pre},n,B}(\alpha,\gamma)
:=
\widehat A_{1,n}(\alpha,\gamma)^\top B(\alpha)
+
\widehat A_{2,n}(\alpha,\gamma)^\top\widehat H_{\alpha,n}B(\alpha),
\end{equation}
which is obtained from \eqref{eq:pre-local-drift} by replacing population
quantities with their sample analogues.

For compactness, define the estimated corrected score
\[
\widehat\Psi_{t,n}(\alpha,\gamma)
:=
\big[\mathbf 1\{Y_t\le m_{t\alpha}(\bar\theta_1,\hat\theta_2)\}-\alpha\big]
 w_{t,n}(\gamma)
+
\widehat l_{t,\alpha}^\top\widehat A_{2,n}(\alpha,\gamma).
\]
For a local direction \(B\), the shifted bootstrap process is
\begin{equation}
\label{eq:pre-local-bootstrap-M}
\widehat M_{n*,B}(\alpha,\gamma)
:=
\frac1n\sum_{t=1}^n
\left[
\omega_t\widehat\Psi_{t,n}(\alpha,\gamma)
+
\frac{\widehat d_{\mathrm{pre},n,B}(\alpha,\gamma)}{\sqrt n}
\right].
\end{equation}
The multiplier component is conditionally centered because \(E^*\omega_t=0\),
and the deterministic shift generates the local drift after multiplication by
\(\sqrt n\).

The shifted bootstrap variance estimator is
\begin{equation}
\label{eq:pre-local-bootstrap-s}
\widehat s_{n*,B}^2(\alpha,\gamma)
:=
\frac1n\sum_{t=1}^n
\left[
\omega_t\widehat\Psi_{t,n}(\alpha,\gamma)
+
\frac{\widehat d_{\mathrm{pre},n,B}(\alpha,\gamma)}{\sqrt n}
\right]^2.
\end{equation}
Define
\[
\widehat Q_{n*,B}(\alpha,\gamma)
:=
\sqrt n
\frac{|\widehat M_{n*,B}(\alpha,\gamma)|}
{\widehat s_{n*,B}(\alpha,\gamma)}.
\]
Then the shifted bootstrap statistic is
\begin{equation}
\widehat T_{n*,B,\Pi}(\lambda)
:=
\sup_{\gamma\in\Gamma}
\Bigg[
\int_{\mathcal T}\widehat Q_{n*,B}^2(\alpha,\gamma)d\Pi(\alpha)
-
\lambda\|\gamma\|_1
\Bigg].
\label{eq:That-star-local-def}
\end{equation}

As above, $B$ denotes the candidate shift used to estimate power, while
$B_0$ denotes the actual local data generating direction. The case $B_0=0$ is the
null.

\begin{theorem}[Shifted multiplier bootstrap validity with pre-estimation]
\label{thm:boot-local-pre_main}
Suppose Assumptions~S1--S6 hold, and let $B$ be any admissible candidate
direction. Under $H_0$ or any local sequence $P_{n,B_0}$ covered by
\assLocal, the following conditional bootstrap results hold, with process
convergence in $\ell^\infty(\mathcal T\times\Gamma)$:
\begin{enumerate}
    \item The shifted process converges to the noncentral limit:
    \[
        \sqrt n\,\widehat M_{n*,B}(\alpha,\gamma)
        \overset{*}{\rightsquigarrow}
        \widetilde{\mathcal M}(\alpha,\gamma)
        +
        d_{\mathrm{pre},B}(\alpha,\gamma).
    \]

    \item The variance estimator remains uniformly consistent:
    \[
        \sup_{\alpha,\gamma}
        \big|\widehat s_{n*,B}^2(\alpha,\gamma)-\tilde\sigma^2(\alpha,\gamma)\big|
        \xrightarrow{p^*}0.
    \]

    \item The penalized shifted bootstrap statistic captures the noncentral
    distribution:
    \[
        \widehat T_{n*,B,\Pi}(\lambda)
        \overset{*}{\rightsquigarrow}
        \sup_{\gamma\in\Gamma}
        \Bigg[
            \int_{\mathcal T}
            \left(
                \frac{\widetilde{\mathcal M}(\alpha,\gamma)
                +d_{\mathrm{pre},B}(\alpha,\gamma)}
                {\tilde\sigma(\alpha,\gamma)}
            \right)^2
            d\Pi(\alpha)
            -
            \lambda\|\gamma\|_1
        \Bigg].
    \]
\end{enumerate}
\end{theorem}

\subsection{Adaptive Selection of the Penalty Parameter}
\label{subsec:lambda-selection}

We select $\lambda$ by the same estimated minimum local power criterion as in
Section~\ref{subsec:calibration}. The difference is that the null and shifted
bootstrap processes are adjusted for nuisance parameter.
Algorithm~\ref{alg:adaptive-pre} gives the full procedure.

\begin{algorithm}[htbp]
\caption{Adaptive penalty selection with pre-estimated nuisance parameters}
\label{alg:adaptive-pre}
\begin{algorithmic}[1]
\Require Data \(\mathcal Z_n\), \(\Pi\), \(\Gamma\), search set
\(\mathcal S=\Lambda\) or \(\Lambda_G\), finite \(\mathcal B\), \(R\), and \(\tau\)
\State Estimate the restricted nuisance curve \(\widehat\theta_2(\alpha)\)
under \(H_0\).
\State Estimate the fitted-quantile density and construct
\(\widehat L_\alpha\), \(\widehat l_{t,\alpha}\),
\(\widehat A_{2,n}(\alpha,\gamma)\), and
\(\widehat\Psi_{t,n}(\alpha,\gamma)\).
\State Draw \(R\) multiplier sequences and reuse them for all \(\lambda\)
and \(B\).
\For{each penalty \(\lambda\) requested by the search routine}
    \For{\(r=1,\ldots,R\)}
        \State Compute the null bootstrap statistic
        \(\widehat T_{n*,\Pi}^{(r)}(\lambda)\) from
        \eqref{eq:Mhat-star-def}.
    \EndFor
    \State Set \(\widehat c_{n,1-\tau,\Pi}(\lambda)\) to their empirical
    \((1-\tau)\)-quantile.
    \For{each \(B\in\mathcal B\)}
        \State Construct \(\widehat d_{\mathrm{pre},n,B}\) and compute
        \(\widehat T_{n*,B,\Pi}^{(r)}(\lambda)\), \(r=1,\ldots,R\).
        \State Set
        \[
        \widehat{\mathcal R}_n(\lambda,B,\tau)
        =\frac{1}{R}\sum_{r=1}^R
        \mathbf 1\!\left\{
        \widehat T_{n*,B,\Pi}^{(r)}(\lambda)>
        \widehat c_{n,1-\tau,\Pi}(\lambda)
        \right\}.
        \]
    \EndFor
    \State Set
    \(\widehat W_n(\lambda)=
    \min_{B\in\mathcal B}\widehat{\mathcal R}_n(\lambda,B,\tau)\).
\EndFor
\State On \(\Lambda_G\), select the smallest maximizer
\(\widehat\lambda_G\) of \(\widehat W_n\).
For continuous search, use a measurable \(\eta_n\)-approximate maximizer
\(\widehat\lambda\), where \(\eta_n=o_p(1)\).
\State Using the original sample, compute
\(\widehat T_{n,\Pi}(\widehat\lambda)\) and reject when it exceeds
\(\widehat c_{n,1-\tau,\Pi}(\widehat\lambda)\).
For a grid search, use \(\widehat\lambda_G\) in place of
\(\widehat\lambda\).
\end{algorithmic}
\end{algorithm}

\FloatBarrier
\section{Monte Carlo Simulations}
\label{sec:simulation}

We examine the finite sample performance of the proposed tests in a linear
autoregressive model. The experiments compare adaptive penalization with
the unpenalized procedure, related tests of quantile restrictions, and
alternative aggregation schemes. The nonlinear design and additional size
comparisons are reported in Section~S5 of the
Supplementary Appendix.

\subsection{Simulation Design}

The data follow the Gaussian AR(2) model
\begin{equation}
Y_t=\mu_0+\mu_1Y_{t-1}+\mu_2Y_{t-2}+\sigma\varepsilon_t,
\qquad \varepsilon_t\overset{\mathrm{iid}}{\sim}\mathcal N(0,1),
\end{equation}
with $(\mu_0,\mu_1,\mu_2,\sigma)=(0.5,0.4,-0.2,1)$ under the null.
The conditional quantile function is
\begin{equation}
m(I_{t-1},\theta_0(\alpha))
=\beta_0(\alpha)+\mu_1Y_{t-1}+\mu_2Y_{t-2},
\qquad \beta_0(\alpha)=\mu_0+\sigma\Phi^{-1}(\alpha),
\label{eq:quantile_model}
\end{equation}
where $I_{t-1}=(Y_{t-1},Y_{t-2})^\top$ and $\Phi$ is the standard normal
distribution function.

Test 1 fixes the entire null quantile curve without pre-estimation.
Test 2 concerns $H_0:\mu_2=-0.2$, with the remaining quantile coefficients
estimated under the null by regressing $Y_t+0.2Y_{t-2}$ on
$(1,Y_{t-1})^\top$. The fitted intercept absorbs
$\mu_0+\sigma\Phi^{-1}(\alpha)$.
The expansion and process conditions for the underlying restricted quantile
regression estimator in this linear design are verified in
Proposition~\ref{prop:linear-ar2-verification} in the Supplementary Appendix.
Both initial lags are zero, and 300 burn-in observations are discarded.
We use $N$ for the retained series length and $n=N-2$ for the effective
sample size. Each experiment uses 2,000 Monte Carlo replications and the
proposed tests use 1,000 Gaussian multiplier draws per replication.
Power is evaluated at nominal level $0.10$, with $N=500$ for Test 1 and
$N=1000$ for Test 2. Fixed alternatives change one generating parameter
by $\delta\in[-0.2,0.2]$, holding the others at their null values.

Test 1 uses nine equally spaced test quantiles on $[0.10,0.90]$, and
Test 2 uses seven on $[0.20,0.80]$. The restricted curve is fitted at
100 nodes on $[0.15,0.85]$. Both tests use 1,000 random auxiliary quantile indices on
$[0.01,0.99]$ for density estimation. The remaining implementation choices
are given in Section~S5.1.

The power comparisons include CFV--KS, based on the subsampling approach
of \citet{chernozhukov2005}, an EV-style CvM test based on
\citet{escanciano2010}, and tests of the QR coefficient process.
The QR comparison uses Sup-Wald for the full curve in Test 1 and Sup-$t$
for the single target in Test 2. Their implementation and size results are
reported in Sections~S5.1--S5.2. All curves report rejection frequencies
using each method's nominal $0.10$ critical value.

\subsection{Power without Pre-Estimation}

Figure~\ref{fig:power_no_pre} reports the comparisons for deviations in
each of the four generating parameters. Rejection frequencies generally
increase as the parameter moves away from its null value. The improvement
from adaptive penalization is most pronounced for location and scale.
At $\delta=0.1$ in $\mu_0$, the adaptive and unpenalized tests reject with
frequencies $0.6030$ and $0.5425$. For a scale deviation of
$\delta=-0.1333$ (rounded), the corresponding frequencies are $0.7295$
and $0.5890$. The differences are smaller for the two autoregressive
coefficients, with a few deviations in $\mu_2$ favoring the unpenalized test.

\begin{figure}[htbp]
\centering
\includegraphics[width=\textwidth]{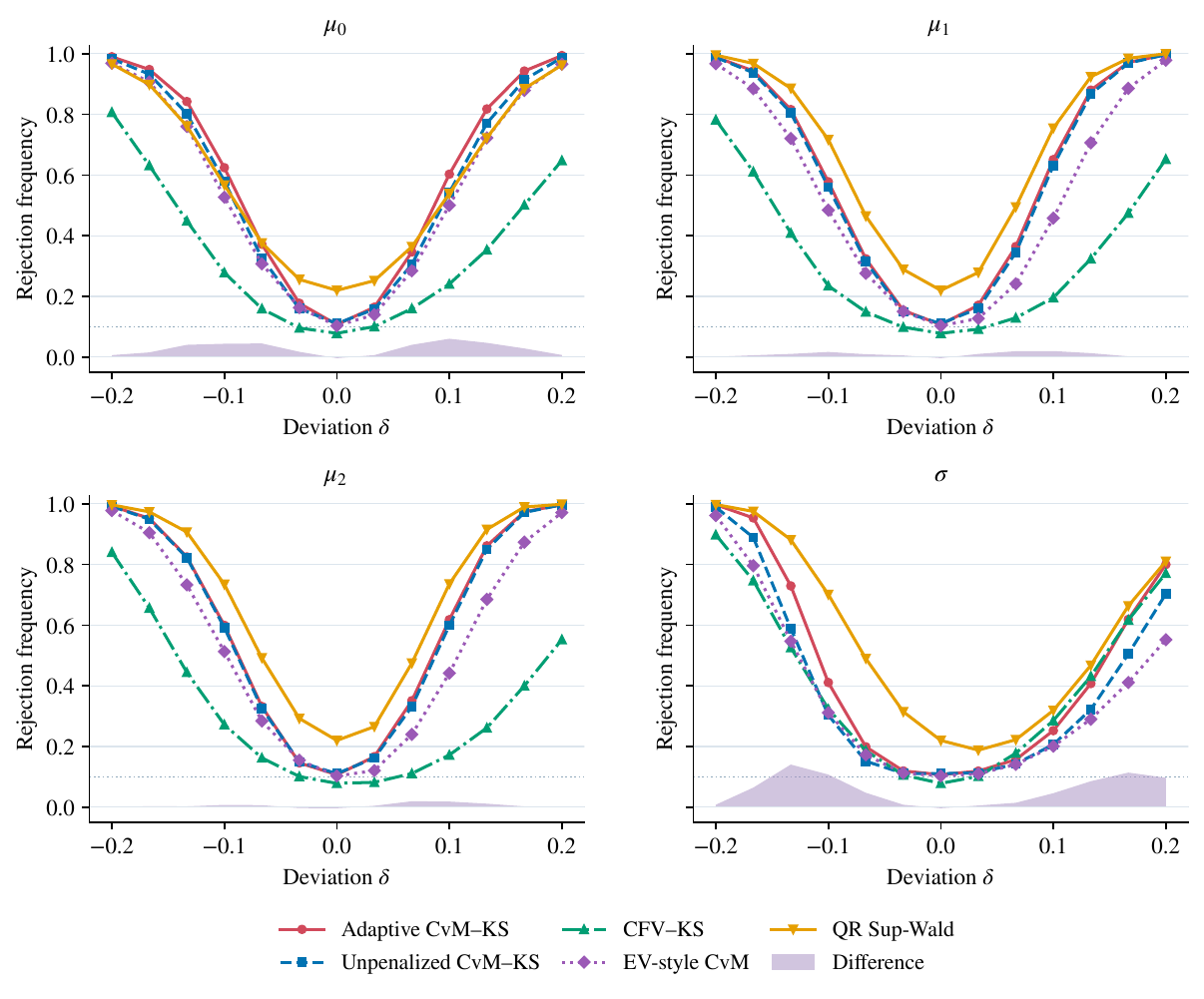}
\caption{Rejection frequencies in the linear AR(2) model without
pre-estimation. Test 1 fixes the full null quantile curve. Each panel varies
one generating parameter, with $N=500$ ($n=498$) and nominal level $0.10$.
The comparison methods are CFV--KS, EV-style CvM, and QR Sup-Wald.
Lines connect the evaluated alternatives. Shading shows the adaptive minus
unpenalized rejection frequency relative to zero on the same vertical scale.}
\label{fig:power_no_pre}
\label{fig:linear-known}
\end{figure}

The relative performance of the comparison methods also depends on the
parameter and the direction of departure. CFV--KS has lower rejection
frequencies than the adaptive test at most of these alternatives, but
performs better for some positive scale deviations. QR Sup-Wald rejects
frequently in several panels, although its null rejection frequency is
$0.2200$ at the nominal $0.10$ level. The differences in calibration,
reported in Table~\ref{tab:comparison-size} of the Supplementary Appendix,
are relevant to interpreting these curves.
\FloatBarrier

\subsection{Power with Pre-Estimation}

Figure~\ref{fig:power_pre} considers departures from $\mu_2=-0.2$ with
the other quantile coefficients estimated under the null. Adaptive
penalization produces modest gains over the unpenalized procedure.
For example, at $\delta=-0.0667$ (rounded), their rejection frequencies
are $0.5330$ and $0.5090$. Both procedures approach one at the largest
departures in either direction.

\begin{figure}[htbp]
\centering
\includegraphics[width=.88\textwidth]{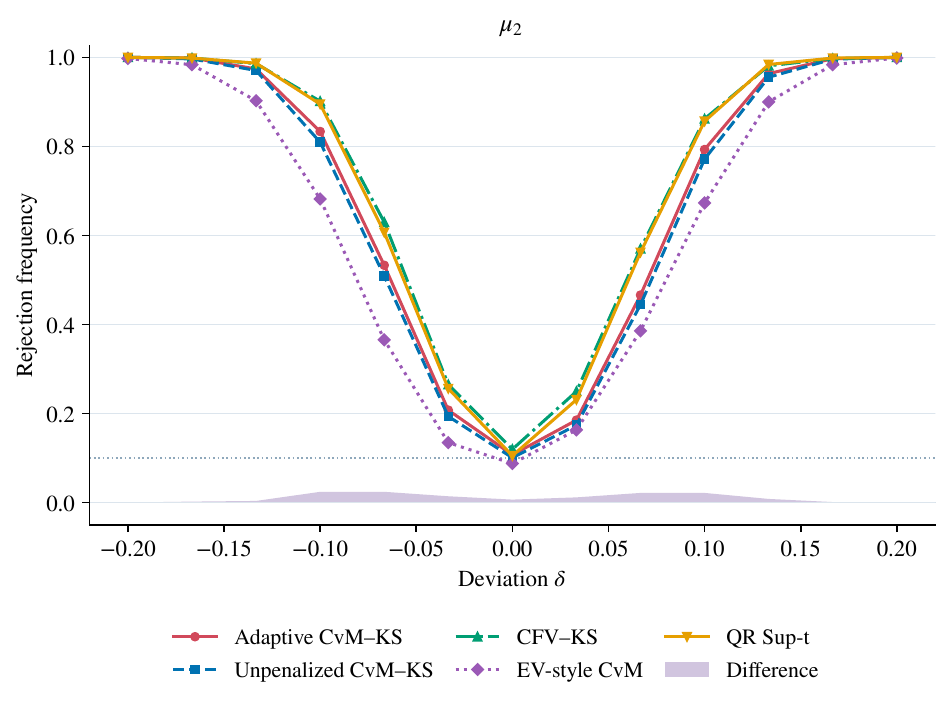}
\caption{Rejection frequencies in the linear AR(2) model with
pre-estimation. Test 2 concerns $\mu_2=-0.2$, with the remaining
quantile coefficients estimated under the null, $N=1000$ ($n=998$), and
nominal level $0.10$. The comparison methods are CFV--KS, EV-style CvM,
and QR Sup-$t$. Shading shows the adaptive minus unpenalized rejection
frequency relative to zero on the same vertical scale.}
\label{fig:power_pre}
\label{fig:linear-qr}
\end{figure}

CFV--KS and QR Sup-$t$ reject more often than the adaptive procedure at
moderate departures, whereas EV-style CvM rejects less often. At the
same deviation of $\delta=-0.0667$, their rejection frequencies are
$0.6300$, $0.6070$, and $0.3655$, respectively. The corresponding size
results are in Section~S5.2. Since Test 1 and Test 2 differ in both the
hypothesis and the sample size, the two figures do not isolate the effect
of pre-estimation alone.
\FloatBarrier

\subsection{Comparison of Aggregation Schemes}
\label{subsec:aggregation-comparison}

Table~\ref{tab:size_comparison} compares the empirical size of CvM--KS,
KS--KS, and KS--CvM, as defined in Remark~\ref{rem:alternative-schemes}.
It reports the adaptive and unpenalized versions at $N=200,500,1000$ and
nominal levels $0.05$ and $0.10$. With 2,000 replications, the Monte Carlo
standard errors at rejection probabilities $0.05$ and $0.10$ are about
$0.0049$ and $0.0067$.

\begin{table}[htbp]
\centering
\setstretch{1}
\footnotesize
\setlength{\tabcolsep}{3pt}
\caption{Empirical size of the three aggregation schemes in the linear AR(2) model}
\label{tab:size_comparison}
\begin{tabular}{rrrrrrrr}
\toprule
& & \multicolumn{3}{c}{Adaptive} & \multicolumn{3}{c}{Unpenalized} \\
\cmidrule(lr){3-5}\cmidrule(lr){6-8}
$N$ & Nominal & CvM--KS & KS--KS & KS--CvM & CvM--KS & KS--KS & KS--CvM \\
\midrule
\multicolumn{8}{l}{\textit{Test 1: full null curve}} \\
200 & 0.05 & 0.0530 & 0.0525 & 0.0545 & 0.0510 & 0.0395 & 0.0470 \\
 & 0.10 & 0.1080 & 0.1020 & 0.0955 & 0.0970 & 0.0940 & 0.0940 \\
\addlinespace
500 & 0.05 & 0.0570 & 0.0555 & 0.0575 & 0.0560 & 0.0495 & 0.0520 \\
 & 0.10 & 0.1070 & 0.1000 & 0.1070 & 0.1105 & 0.1020 & 0.1055 \\
\addlinespace
1000 & 0.05 & 0.0495 & 0.0490 & 0.0515 & 0.0540 & 0.0485 & 0.0500 \\
 & 0.10 & 0.1135 & 0.1095 & 0.1065 & 0.1085 & 0.1035 & 0.0965 \\
\addlinespace
\midrule
\multicolumn{8}{l}{\textit{Test 2: one target with pre-estimation}} \\
200 & 0.05 & 0.0515 & 0.1170 & 0.0900 & 0.0425 & 0.0935 & 0.0615 \\
 & 0.10 & 0.1080 & 0.2135 & 0.1690 & 0.0935 & 0.1705 & 0.1360 \\
\addlinespace
500 & 0.05 & 0.0405 & 0.1005 & 0.0630 & 0.0375 & 0.0865 & 0.0475 \\
 & 0.10 & 0.0945 & 0.1650 & 0.1235 & 0.0860 & 0.1650 & 0.1065 \\
\addlinespace
1000 & 0.05 & 0.0535 & 0.0740 & 0.0670 & 0.0510 & 0.0765 & 0.0605 \\
 & 0.10 & 0.1070 & 0.1340 & 0.1275 & 0.1005 & 0.1410 & 0.1160 \\
\addlinespace
\bottomrule
\end{tabular}
\par\smallskip
\begin{minipage}{\linewidth}\footnotesize
\textit{Notes:} Entries are rejection proportions from 2,000 replications, with $n=N-2$.
\end{minipage}
\end{table}

Without pre-estimation, all three aggregation rules are generally close
to the nominal levels. With pre-estimation, CvM--KS is generally better than the other two rules. At $N=200$ and nominal level $0.10$,
the adaptive rejection frequencies are $0.1080$, $0.2135$, and $0.1690$
for CvM--KS, KS--KS, and KS--CvM. At $N=1000$, they are $0.1070$,
$0.1340$, and $0.1275$. The rules differ in penalty placement as well as
aggregation order, so the comparison reflects both features of their
implementation. Section~S5 reports the additional results and their
analysis.
\FloatBarrier

\section{Empirical Study}\label{sec:empirical}

We apply the proposed penalized Bierens maximum statistic to the Growth-at-Risk
(GaR) framework of \citet{adrian2019vulnerable} to conduct uniform inference on
the association between financial conditions and the conditional distribution
of GDP growth. The original GaR analysis documents that tighter financial
conditions are especially informative about lower conditional quantiles of
future growth. Related work develops formal out-of-sample tests of conditional
quantile coverage and GaR model comparisons \citep{corradi2023outofsample}.
Our application addresses two in-sample restrictions: whether the NFCI
coefficient is uniformly zero and whether it can be represented by a common
constant over the quantile range under study.

\subsection{Data and Model Specification}

We use U.S.\ macroeconomic data from the Federal Reserve Economic Data
(FRED) database, downloaded on March 6, 2026. Let $y_t$ denote annualized
quarter-over-quarter real GDP growth (\texttt{A191RL1Q225SBEA}). The two
growth outcomes are
\[
Y_t^{(1)}=y_{t+1},
\qquad
Y_t^{(4)}=\frac14\sum_{j=1}^{4}y_{t+j}.
\]
The four-quarter outcome averages the next four quarterly growth rates and
does not include the current quarter. Following the data construction of
\citet{adrian2019vulnerable}, its NFCI regressor is the quarterly average of
weekly observations, assigned by the week-ending date. For the one-quarter
specification, we use the last weekly NFCI observation in each quarter.
Higher NFCI values indicate tighter financial conditions.
The prediction origins span 1971:Q1--2025:Q3 for $h=1$ ($n=219$) and
1971:Q1--2024:Q4 for $h=4$ ($n=216$). Both response series end in 2025:Q4.

For $h\in\{1,4\}$, we consider the conditional quantile model
\begin{equation}\label{eq:gar_model}
Q_{Y_t^{(h)}}(\alpha\mid\mathcal I_t)
=\beta_0(\alpha)+\beta_1(\alpha)\,\mathrm{NFCI}_t
+\beta_2(\alpha)\,y_t,
\end{equation}
where $\mathcal I_t$ denotes information at the prediction origin.
The coefficient of interest, $\beta_1(\alpha)$, measures the association
between financial conditions and the $\alpha$-quantile of future growth,
controlling for current GDP growth. We examine
\begin{align}
H_0^{(1)} &: \beta_1(\alpha)=0
\quad\text{for all }\alpha\in\mathcal T,
\label{eq:H0_zero}\\
H_0(c) &: \beta_1(\alpha)=c
\quad\text{for all }\alpha\in\mathcal T,
\label{eq:H0_const}
\end{align}
where $c$ is a fixed candidate common coefficient. Inverting $H_0(c)$
identifies the candidate values retained by the test.

We also compare the quantile coefficients with the OLS slope from the same
regression. The values are $c_{\mathrm{OLS}}=-1.255$ for $h=1$ and
$c_{\mathrm{OLS}}=-0.601$ for $h=4$, with unrounded values used in the
calculations. These are descriptive comparisons with fixed numerical
benchmarks. They do not incorporate uncertainty from estimating the OLS
slope in the same sample.

\subsubsection{Implementation}

The main statistic averages over 100 equally spaced quantile levels in
$[0.10,0.90]$. A separate statistic averages over 25 levels in
$[0.05,0.30]$ to examine the lower tail. Each aggregation measure assigns
equal mass to its nodes. The two grids are not nested.
The nuisance intercept and GDP coefficient are estimated under each null
restriction. We use centered exponential weights, the correction for
pre-estimation in Section~\ref{sec:pre_estimation}, and a $40\times40$
grid on $[-4,4]^2$ for the standardized conditioning directions.
The penalty is selected from $\Lambda_G=\{0,0.1,\ldots,2\}$ using the
constant direction $B(\alpha)=0.50$.

For the one-quarter outcome, we use the ordinary Gaussian multiplier
bootstrap and the variance estimator with the correction for pre-estimation.
For the overlapping four-quarter outcome, we use Bartlett HAC
studentization and a block multiplier bootstrap, both based on the full
corrected scores. The baseline block length is eight quarters.
This four-quarter implementation is an empirical extension beyond the
MDS conditions of the theory developed in this paper.
Penalty selection uses 4,999 multiplier draws, followed by an independent
set of 9,999 draws for calibration. Section~\ref{app:empirical-implementation}
gives the construction and the remaining numerical details.

\subsection{Uniform Test Results}\label{ssec:test_results}

Table~\ref{tab:gar_results} reports the baseline results. At the one-quarter
horizon, both procedures reject the restriction of a zero coefficient over the
main quantile range at the 5\% level, with $p$-values of $0.0201$ and
$0.0268$. Neither procedure rejects the OLS benchmark over the main range
or the lower tail at the 10\% level.

\begin{table}[!htbp]
\centering
\caption{Tests of the NFCI coefficient}
\label{tab:gar_results}
\begingroup
\small
\setlength{\tabcolsep}{4pt}
\renewcommand{\arraystretch}{1.08}
\begin{tabular}{cllrrr}
\toprule
Horizon $h$ & Quantile range $\mathcal T$ & Restriction / benchmark & $p_{\mathrm A}$ & $p_{\mathrm U}$ & $\widehat\lambda$ \\
\midrule
1 & $[0.10,0.90]$ & $c=0$ & 0.0201 & 0.0268 & 0.3 \\
1 & $[0.10,0.90]$ & $c=c_{\mathrm{OLS}}$ & 0.1298 & 0.1238 & 1.3 \\
1 & $[0.05,0.30]$ & $c=c_{\mathrm{OLS}}$ & 0.2870 & 0.3581 & 0.3 \\
\addlinespace
4 & $[0.10,0.90]$ & $c=0$ & 0.0657 & 0.0540 & 0.2 \\
4 & $[0.10,0.90]$ & $c=c_{\mathrm{OLS}}$ & 0.1890 & 0.2895 & 1.0 \\
4 & $[0.05,0.30]$ & $c=c_{\mathrm{OLS}}$ & 0.0055 & 0.0070 & 0.2 \\
\bottomrule
\end{tabular}
\endgroup
\par\vspace{3pt}
\begin{minipage}{\linewidth}
\footnotesize
Notes: A and U denote adaptive and unpenalized procedures. The one-quarter results use ordinary multipliers. The overlapping four-quarter results use Bartlett HAC and block multipliers with $\ell=8$. 
\end{minipage}
\end{table}

At the four-quarter horizon, the restriction of a zero coefficient is rejected
at 10\% but not at 5\%, with $p$-values of $0.0657$ and $0.0540$.
The OLS benchmark is retained over the main range but rejected over the
lower tail, where the $p$-values are $0.0055$ and $0.0070$.
This tail comparison concerns the particular OLS value, not the hypothesis
that some other constant coefficient could describe the lower tail.

Figure~\ref{fig:gar_paths} shows the unrestricted quantile regression
estimates. The NFCI coefficient is more negative at lower quantiles and
approaches zero, or becomes positive, at upper quantiles. This pattern
is consistent with the downside-risk association documented by
\citet{adrian2019vulnerable}. Differences between the two panels reflect
both the growth horizon and the quarterly aggregation of NFCI, so they
cannot be attributed to the horizon alone.

\begin{figure}[htbp]
\centering
\includegraphics[width=\textwidth]{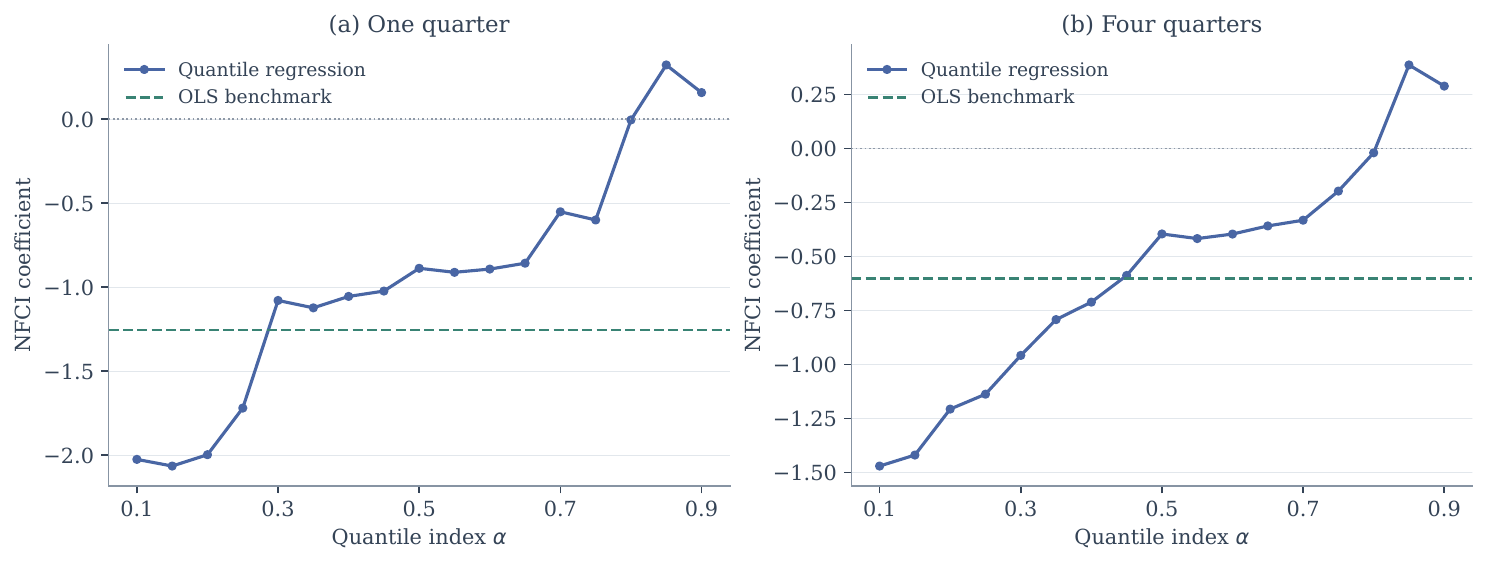}
\caption{NFCI coefficient estimates in the GaR model.
Solid curves show unrestricted quantile regression estimates, and dashed
lines show the OLS benchmark.}
\label{fig:gar_paths}
\end{figure}

The inversion results over the main range in Section~\ref{app:empirical-implementation}
retain common negative coefficients at both horizons. For four-quarter
growth, the retained grid values form several separate components, rather
than a single interval. Thus, the varying estimated paths do not by
themselves establish that every common coefficient is rejected.

\subsection{Sensitivity and Interpretation}

The four-quarter evidence against a zero coefficient depends on the
treatment of serial dependence and the sample period. With a block length
of four, the adaptive $p$-value is $0.1070$, whereas the unpenalized value is
$0.0840$. With block lengths of twelve and sixteen, both values are below
$0.05$. Restricting all response quarters to precede 2020 raises them to
$0.1812$ and $0.1118$. By contrast, the four-quarter lower-tail comparison
with the OLS benchmark is rejected by both procedures at 5\% for all
reported block lengths and in the sample ending before 2020.

Changing the density bandwidth or increasing the number of fitted
quantiles leaves the baseline decisions at the 10\% level unchanged at
both horizons. Independent sets of multiplier draws also leave the decisions
in all six baseline comparisons unchanged. The selected penalties and some values near
the boundaries of the inversion sets are less stable. The appendix reports
these results separately, since stability of a rejection decision does not
imply stability of the selected penalty or of every retained value.

The empirical distinction is therefore between evidence against a zero
coefficient and evidence against the particular OLS benchmark. The latter
is concentrated in the four-quarter lower-tail comparison, while the
former is more sensitive to dependence and sample choices at that horizon.
The adaptive and unpenalized procedures yield the same baseline decisions,
but neither has uniformly smaller $p$-values across the reported settings.

\section{Conclusion}
\label{sec:conclusion}

This paper develops tests of quantile coefficients identified by
conditional quantile restrictions. Exponential weights turn the conditional
restriction into unconditional moments indexed by conditioning directions,
while CvM aggregation combines evidence across quantiles. An $\ell_1$-penalty
regularizes the directional supremum, and a shifted bootstrap selects the
penalty by an estimated criterion for maximin local power. For known parameters,
the selector has no lower maximin local power than the unpenalized test under
the stated conditions. When nuisance parameters are estimated under the null,
the limiting process contains the contribution of their influence functions. The
corrected variance estimator and multiplier bootstrap retain that term
throughout the feasible procedure.

The simulations show that finite sample performance depends on the aggregation
order and the pre-estimation problem. In the linear AR(2) design, the proposed
CvM--KS statistic is relatively close to nominal size in most reported
comparisons. The adaptive test outperforms the unpenalized test in most reported cases, consistent with our theoretical results that adaptive penalization does not reduce asymptotic maximin local power.

The Growth-at-Risk application distinguishes tests of a zero coefficient,
comparisons with a fixed OLS benchmark, and inversion over candidate common
coefficients. The one-quarter results reject a zero coefficient but retain
the OLS benchmark. For future four-quarter growth, HAC studentization and
a block multiplier bootstrap yield a rejection of the zero restriction
at 10\%, although this decision is sensitive to block length and sample
period. The lower tail comparison with the OLS benchmark is more stable
across the reported checks. Nonempty sets of common coefficients remain
over the main quantile range.

An important direction for future work is to extend the asymptotic theory
beyond martingale differences to more general serial dependence, including
the use of HAC covariance estimators and a block multiplier bootstrap.

Extending the framework to an infinite dimensional conditioning vector would
allow inference based on the full history of a process, as in the
ARCH($\infty$) representation of GARCH models. A sharper analysis of the
maximin selector, particularly in Gaussian location models where the
studentized drift has a closed form, could also provide primitive conditions
for strict improvements in maximin local power beyond the criterion
comparison in Theorem~\ref{thm:maximin-no-loss_main}.

More broadly, the idea of choosing a penalty to improve local power may
extend to other testing problems. Establishing conditions to figure out when this
approach preserves size and improves power would be a useful direction for
future theoretical research.

\clearpage
\phantomsection
\addcontentsline{toc}{section}{Appendix: Supplementary Material}
\section*{Appendix: Supplementary Material}

\renewcommand{\theHsection}{appendix.\arabic{section}}
\renewcommand{\theHequation}{appendix.\arabic{equation}}
\renewcommand{\theHtheorem}{appendix.\arabic{theorem}}
\renewcommand{\theHassumption}{appendix.\arabic{assumption}}
\renewcommand{\theHlemma}{appendix.\arabic{lemma}}
\renewcommand{\theHremark}{appendix.\arabic{remark}}
\renewcommand{\theHcorollary}{appendix.\arabic{corollary}}
\setcounter{figure}{0}
\setcounter{table}{0}
\renewcommand{\thefigure}{S\arabic{figure}}
\renewcommand{\thetable}{S\arabic{table}}
\renewcommand{\theHfigure}{appendix.\arabic{figure}}
\renewcommand{\theHtable}{appendix.\arabic{table}}

\setcounter{section}{0}
\setcounter{equation}{0}
\setcounter{theorem}{0}
\setcounter{assumption}{0}
\setcounter{lemma}{0}
\setcounter{remark}{0}
\setcounter{corollary}{0}

\renewcommand{\thesection}{S\arabic{section}}
\renewcommand{\theequation}{S\arabic{equation}}
\renewcommand{\thetheorem}{S\arabic{theorem}}
\renewcommand{\theassumption}{S\arabic{assumption}}
\renewcommand{\thelemma}{S\arabic{lemma}}
\renewcommand{\theremark}{S\arabic{remark}}
\renewcommand{\thecorollary}{S\arabic{corollary}}

\section*{S1. Additional Assumptions}

\begin{assumption}[Data-generating process]
\label{ass:mixing}
\begin{enumerate}
\item[(i)] $Z_t=(Y_t,I_{t-1})$ is strictly stationary and $\beta$-mixing, and
for some $q>2$,
\[
k^{\frac{q}{q-2}}(\log k)^{\frac{2(q-1)}{q-2}}\beta_Z(k)\to0.
\]
With $\mathcal F_{t-1}:=\sigma(Z_s:s\le t-1)$, $I_{t-1}$ is
$\mathcal F_{t-1}$-measurable.
\item[(ii)] $m(\cdot,\theta_0(\alpha))$ is nondecreasing in $\alpha$ a.s., and
for every $\alpha\in\mathcal T$,
\[
P\{Y_t\le m(I_{t-1},\theta(\alpha))\mid I_{t-1}\}=\alpha
\quad\Longleftrightarrow\quad
\theta(\alpha)=\theta_0(\alpha).
\]
\item[(iii)] For every $\alpha\in\mathcal T$,
\[
E\!\left[\mathbf 1\{Y_t\le m(I_{t-1},\theta_0(\alpha))\}-\alpha
\mid\mathcal F_{t-1}\right]=0.
\]
\end{enumerate}
\end{assumption}

Assumption~\ref{ass:mixing}(iii) eliminates score covariances across time. It is
stronger than conditioning on $I_{t-1}$ unless the information sets coincide.

\begin{assumption}[Regularity conditions]
\label{ass:regularity}
(i) $\Gamma\subset\mathbb R^{d_I}$ is compact with nonempty interior and
$0\in\Gamma$, and
$w(I_{t-1},\gamma)=\exp\{\gamma^\top\Phi(I_{t-1})\}$. Let
$\bar\gamma:=\sup_{\gamma\in\Gamma}\|\gamma\|_2$. If
\[
E\!\left[\exp\{q\bar\gamma\|I_{t-1}\|_2\}\right]<\infty,
\]
take $\Phi(I)=I$. Otherwise, use a bounded Borel-measurable injection $\Phi$.
Both choices give an $L_q$ envelope. Inside exponential weights, write
$I_{t-1}$ for $\Phi(I_{t-1})$.

(ii) The parameter space $\Theta$ is compact. The class
$\{x\mapsto m(x,\theta):\theta\in\Theta\}$ is a permissible VC-subgraph class.

(iii) For a vector function $M(I_{t-1},\theta(\alpha))$ and every fixed $k<\infty$,
\[
\sup_{\substack{\|\theta_1-\theta_2\|_2\leq\frac{k}{\sqrt{n}},\\ 1\leq t \leq n, \,\alpha \in \mathcal{T}}}
\left\{
\sqrt{n}\left|
\begin{aligned}
&m(I_{t-1},\theta_2(\alpha))-m(I_{t-1},\theta_1(\alpha))\\
&\quad-(\theta_2(\alpha)-\theta_1(\alpha))^\top
M(I_{t-1},\theta_1(\alpha))
\end{aligned}
\right|
\right\} = o_p(1).
\]

(iv) The conditional density satisfies:
\begin{enumerate}
    \item $\sup_{y,i}f(y\mid i)<\infty$.
    \item $\mathcal T\Subset(0,1)$ and, for some $c_0>0$,
    \[
    \inf_{i,\,\alpha\in\mathcal T}
    f(m(i,\theta_0(\alpha))\mid i)\ge c_0.
    \]
    \item $\sup_{y,i}|\partial_y^k f(y\mid i)|<\infty$ for $k=1,2$.
    \item For density-based results, let
    $\mathcal N_{\varepsilon_0}(\mathcal T)
    :=\{u\in(0,1):\operatorname{dist}(u,\mathcal T)<\varepsilon_0\}$
    for some $\varepsilon_0>0$. On this neighborhood the null-imposed curve
    $q_i^0(u):=m(i,\bar\theta(u))$ is strictly increasing and $C^3$, with
    \[
    \begin{aligned}
    P\{Y_t\le q_{I_{t-1}}^0(u)\mid I_{t-1}\}&=u\quad\text{a.s.},\\
    0<c\le\partial_u q_i^0(u)
      &=1/f_{Y_t\mid I_{t-1}}(q_i^0(u)\mid i)\le C<\infty .
    \end{aligned}
    \]
    The constants are uniform in $(i,u)$; the first two derivatives of
    $z\mapsto\partial_z(q_i^0)^{-1}(z)$ are uniformly bounded on the images
    of the neighborhood. Independently of the data,
    $U_1,\ldots,U_J\stackrel{\mathrm{i.i.d.}}{\sim}U(\mathcal A)$, where
    $\mathcal A=[a_1,a_2]$ is fixed,
    $\mathcal T\subset\operatorname{int}(\mathcal A)$,
    $\mathcal A\subset\mathcal N_{\varepsilon_0}(\mathcal T)$, and
    $L_{\mathrm{sim}}:=a_2-a_1$.
\end{enumerate}
\end{assumption}

\begin{remark}[Role of the neighborhood condition]
\label{rem:density-neighborhood}
Assumption~\ref{ass:regularity}(iv)(4) extends the quantile restriction beyond
$\mathcal T$ for density estimation only; the test is still indexed by $\mathcal T$.
\end{remark}

\begin{remark}[Bounded transformation of covariates]
\label{rem:bounded_transform}
For a bounded Borel injection $\Phi$ (e.g., componentwise arctangent),
$\sigma(\Phi(I_{t-1}))=\sigma(I_{t-1})$
\citep{bierens1990,stinchcombe1998consistent}. The identity map is allowed under
the exponential moment in Assumption~\ref{ass:regularity}(i).
\end{remark}

\begin{assumption}[Kernel and bandwidth]
\label{ass:kernel}
(a) $K$ is nonnegative, symmetric, bounded, and $C^2$ with bounded
derivatives,
\[
\int K(u)\,du=1,
\qquad
\int u^2K(u)\,du=\mu_{2K}\in(0,\infty).
\]
It is either compactly supported or, for some $C,c>0$,
\[
|K^{(j)}(u)|\le C\exp(-c|u|),
\qquad j=0,1,2.
\]

(b) For deterministic $a_J\le b_J$,
$P(a_J\le h_J\le b_J)\to1$ and
\[
b_J \to 0, \qquad \frac{J a_J}{\log J} \to \infty.
\]

(c) Along $J=J_n$,
\[
\frac{J a_J}{\log n} \to \infty.
\]

\begin{samepage}
(d) For the pre-estimation results,
\begin{equation}
    n^{-1/2}a_{J_n}^{-2}\longrightarrow0.
    \label{eq:generated-parameter-bandwidth-rate}
\end{equation}
\end{samepage}
\end{assumption}

In density-based bootstrap arguments, conditioning on the data and all
pre-multiplier randomness ($\mathcal H_n$) fixes the pilot, density draws,
and bandwidth. The multipliers are independent of $\mathcal H_n$;
averaging over the auxiliary randomness gives convergence conditional on the data.

Throughout, write
$m_{t\alpha}(\theta_1,\theta_2):=m(I_{t-1},\theta_1(\alpha),\theta_2(\alpha))$.

\begin{assumption}[Nuisance estimator]
\label{ass:pre-est-regularity}
Under $H_0$, a unique nuisance curve $\theta_{20}$ makes
$\bar q_i(u):=m(i,\bar\theta_1(u),\theta_{20}(u))$ satisfy
Assumption~\ref{ass:regularity}(iv)(4) uniformly in $i$ on
$\mathcal N_{\varepsilon_0}(\mathcal T)$. In particular,
\[
0<c_q\le\partial_u\bar q_i(u)
=1/f_{Y_t\mid I_{t-1}}(\bar q_i(u)\mid i)\le C_q<\infty.
\]

The restricted estimator satisfies
\[
\sup_{u\in\mathcal A}
\left\|
\sqrt{n}\big(\hat{\theta}_2(u) - \theta_{20}(u)\big)
- \frac{1}{\sqrt{n}}\sum_{t=1}^n l_{t,u}
\right\|=o_p(1),
\]
where
\[
l_{t,u}
=-L_u^{-1}g_{tu}(\bar\theta_1,\theta_{20})
\left[
\mathbf 1\{Y_t\le \bar q_{I_{t-1}}(u)\}-u
\right].
\]
\begin{itemize}
    \item[(i)] For
    $J_u:=E[l_{t,u}l_{t,u}^\top]$,
    $\inf_{u\in\mathcal A}\lambda_{\min}(J_u)>0$.

    \item[(ii)] In $\ell^\infty(\mathcal A)$,
    $n^{-1/2}\sum_{t=1}^n l_{t,u}\rightsquigarrow
    \mathcal G_{\theta_2}(u)$, a zero-mean Gaussian process with kernel
    \[
    K_{\theta_2}(u_1,u_2)
    =\sum_{k\in\mathbb Z}
    \operatorname{Cov}\!\left(l_{0,u_1},l_{k,u_2}\right),
    \]
    with absolute uniform convergence on $\mathcal A^2$. The coordinate classes
    of $\{l_{t,u}:u\in\mathcal A\}$ and the centered raw-score class
    \[
    \left\{
    w_t(\gamma)
    [\mathbf 1\{Y_t\le m_{t\alpha}(\bar\theta_1,\theta_{20})\}-\alpha]:
    (\alpha,\gamma)\in\mathcal T\times\Gamma
    \right\}
    \]
    are jointly permissible and VC-subgraph with a common $L_q$ envelope and
    joint empirical process convergence. Their polynomial covering bound and
    envelope also hold for the localized score, gradient, and corrected score
    classes in a fixed neighborhood of $\theta_{20}$ on $\mathcal A$.

    \item[(iii)] Define
    \[
    \Psi_t(\alpha,\gamma)
    :=w_t(\gamma)\big[\mathbf 1\{Y_t\le
    m_{t\alpha}(\bar\theta_1,\theta_{20})\}-\alpha\big]
    +l_{t,\alpha}^\top A_2(\alpha,\gamma),
    \]
    where $w_t$ and $A_2$ are as in the main text. For all index pairs in
    $\mathcal T\times\Gamma$, predictability and Assumption~\ref{ass:mixing}(iii)
    give $E[\Psi_t(\alpha,\gamma)\Psi_s(\alpha',\gamma')]=0$ for $t\ne s$
    (Remark~\ref{rem:orthogonality}).
    \item[(iv)] Let
    \[
    L_u
    :=E\!\left[
    g_{tu}(\bar\theta_1,\theta_{20})
    g_{tu}(\bar\theta_1,\theta_{20})^\top
    f_{Y_t\mid I_{t-1}}\!\left(\bar q_{I_{t-1}}(u)\mid I_{t-1}\right)
    \right].
    \]
    Require $\inf_{u\in\mathcal A}s_{\min}(L_u)>0$. For pre-estimation, replace
    the direction set in Assumption~\ref{ass:regularity}(i) by a compact
    $\Gamma$ on which the centered weights and corrected scores satisfy
    \[
    \inf_{\gamma\in\Gamma}\operatorname{Var}\!\left(e^{\gamma^\top I_{t-1}}\right)>0,
    \qquad
    \inf_{(\alpha,\gamma)\in\mathcal T\times\Gamma}\tilde\sigma^2(\alpha,\gamma)>0,
    \]
    where $\tilde\sigma^2(\alpha,\gamma):=E[\Psi_t(\alpha,\gamma)^2]$.
\end{itemize}
\end{assumption}

\begin{remark}[Orthogonality conditions]
\label{rem:orthogonality}
Write $r_{t,u}=\mathbf 1\{Y_t\le\bar q_{I_{t-1}}(u)\}-u$.
The conditional quantile restriction gives $E[r_{t,u}\mid I_{t-1}]=0$,
so $E[l_{t,u}]=0$ for every $u\in\mathcal A$. For $\alpha\in\mathcal T$,
\[
l_{t,\alpha}=-L_\alpha^{-1}g_{t\alpha}r_{t,\alpha},
\qquad
\Psi_t(\alpha,\gamma)
=\big[w_t(\gamma)-A_2(\alpha,\gamma)^\top
L_\alpha^{-1}g_{t\alpha}\big]r_{t,\alpha}.
\]
For $\alpha,\alpha'\in\mathcal T$, the coefficients are predictable and the
residuals are martingale differences by Assumption~\ref{ass:mixing}(iii).
Conditioning at the later date gives, for $t\ne s$,
\[
E[l_{t,\alpha}r_{s,\alpha'}]=0,\qquad
E[\Psi_t(\alpha,\gamma)\Psi_s(\alpha',\gamma')]=0.
\]
For other influence functions, the i.i.d.\ multiplier argument requires
separate checks of mean zero on $\mathcal A$ and cross-time orthogonality on $\mathcal T$.
\end{remark}

\begin{assumption}[Smoothness for the plug-in expansion]
\label{ass:nuisance-smooth}
\begin{itemize}
    \item[(i)] For each $i$, $m(i,\cdot)$ has a continuous gradient on
    a fixed open neighborhood $\Theta^+$ of $\Theta$. There is a measurable
    finite $\mathcal L(i)\ge0$ such that
    \[
    \|\nabla_\theta m(i,\theta)-\nabla_\theta m(i,\vartheta)\|
    \le \mathcal L(i)\|\theta-\vartheta\|
    \]
    whenever the segment joining $\theta$ and $\vartheta$ lies in $\Theta^+$.

    \item[(ii)] For some $q>2$,
    \[
    \sup_{\alpha\in\mathcal T}E\!\left[\sup_{\theta\in\Theta}\|\nabla_{\theta_1}m(I_{t-1},\theta)\|^q\right]<\infty,
    \qquad
    \sup_{\alpha\in\mathcal T}E\!\left[\sup_{\theta\in\Theta}\|g_{t\alpha}(\theta)\|^q\right]<\infty,
    \]
    where $g_{t\alpha}(\theta):=\nabla_{\theta_2}m(I_{t-1},\theta)$.

    \item[(iii)] Define
    \[
    \begin{aligned}
    \mathcal E_t
    :={}&1+\mathcal L(I_{t-1})
    +\sup_{\gamma\in\Gamma}\exp\{|\gamma^\top I_{t-1}|\}\\
    &+\sup_{\theta\in\Theta}\|\nabla_{\theta_1}m(I_{t-1},\theta)\|
    +\sup_{\alpha\in\mathcal T,\,\theta\in\Theta}
    \|g_{t\alpha}(\theta)\|.
    \end{aligned}
    \]
    For some $r>4$, $E[\mathcal E_t^r]<\infty$. For density estimation, the same
    bound holds with $\mathcal T$ replaced by $\mathcal A$.

\end{itemize}
\end{assumption}

\begin{assumption}[Local alternatives and local stability]
\label{ass:local}
Let $\mathcal B$ be a set of continuous tangent directions with
$\sup_{B\in\mathcal B}\|B\|_\infty<\infty$. For $\alpha\in\mathcal T$, set
\[
\begin{aligned}
\theta_{n,B}(\alpha)&=\theta_0(\alpha)-n^{-1/2}B(\alpha)
&&\text{(known curve)},\\
\theta_{1n,B}(\alpha)&=\theta_{10}(\alpha)-n^{-1/2}B(\alpha)
&&\text{(pre-estimation)}.
\end{aligned}
\]
For all large $n$, the implied quantiles are nondecreasing in $\alpha$;
the segments from the null to the local parameters lie in $\Theta^+$ of
Assumption~\ref{ass:nuisance-smooth}(i), with gradients bounded by
$C\mathcal E_t$. Uniformly in $B\in\mathcal B$ and the indicated indices:
\begin{itemize}
    \item[(i)] The corresponding triangular-array probability measures
    \(P_{n,B}\) are contiguous to the null sequence.
    \item[(ii)] Write the null quantile as
    \[
    q_i^0(\alpha)=
    \begin{cases}
    m(i,\theta_0(\alpha)),&\text{without pre-estimation},\\
    m(i,\theta_{10}(\alpha),\theta_{20}(\alpha)),&\text{with pre-estimation}.
    \end{cases}
    \]
    The local and null conditional densities $f_{n,B,t}$ and $f_0$ are uniformly bounded on
    $|y-q_i^0(\alpha)|\le\delta$ for some fixed $\delta>0$. With
    $\mathcal E(i)$ as in Assumption~\ref{ass:nuisance-smooth}(iii), for each
    fixed $M<\infty$,
    \[
    \sup_{\substack{B\in\mathcal B,\,1\le t\le n,\,\mathcal E(i)\le M,\,\alpha\in\mathcal T\\
                    |y-q_i^0(\alpha)|\le\delta}}
    |f_{n,B,t}(y\mid i)-f_0(y\mid i)|\longrightarrow0.
    \]
    For the same $r>4$ and envelope $\mathcal E_t$ as in
    Assumption~\ref{ass:nuisance-smooth}(iii), require, for some $n_0$,
    \[
    \sup_{n\ge n_0}\sup_{B\in\mathcal B}
    \frac1n\sum_{t=1}^n E_{n,B}[\mathcal E_t^r]<\infty.
    \]
    \item[(iii)] After centering at their local means, the ordinary and stacked
    pre-estimation score processes satisfy the null functional central limit
    theorems and stochastic equicontinuity conditions. Their covariance kernels
    and studentizer probability limits differ from the null limits by $o(1)$.
\end{itemize}
\end{assumption}

\begin{remark}[Scope of the local-stability condition]
Contiguity alone does not imply (iii). \citet[Assumptions A.2--A.3]{chernozhukov2005}
likewise maintain joint Gaussian limits under contiguous alternatives; the
processes and drift corrections here pertain to the present statistic.
\end{remark}
\section*{S2. Preliminary Lemmas}

\begin{lemma}[Local drift expansions]
\label{lem:local-drift-expansions}
Suppose Assumptions~\ref{ass:mixing}--\ref{ass:regularity} and
\ref{ass:nuisance-smooth}--\ref{ass:local} hold. Under the ordinary local
sequence in Assumption~\ref{ass:local}, write
$q_t^0(\alpha)=m(I_{t-1},\theta_0(\alpha))$,
$r_{t,\alpha}=\mathbf 1\{Y_t\le q_t^0(\alpha)\}-\alpha$, and
$f_{t\alpha}=f_0(q_t^0(\alpha)\mid I_{t-1})$. For
$Z_{nt}(\alpha,\gamma)=e^{\gamma^\top I_{t-1}}r_{t,\alpha}$,
uniformly on $\mathcal T\times\Gamma$,
\[
\begin{aligned}
n^{-1/2}\sum_{t=1}^n E_{n,B} Z_{nt}(\alpha,\gamma)
 &=d_B(\alpha,\gamma)+o(1),\\
d_B(\alpha,\gamma)
 &:=E\!\left[e^{\gamma^\top I_{t-1}} f_{t\alpha}
       \nabla_\theta m(I_{t-1},\theta_0(\alpha))^\top B(\alpha)\right].
\end{aligned}
\]

With pre-estimation, suppose also that Assumption~\ref{ass:pre-est-regularity}
holds. Evaluate $m_{t\alpha}$ and its gradients at
$(\theta_{10},\theta_{20})$, so $q_t^0(\alpha)=m_{t\alpha}(\theta_{10},\theta_{20})$
and $g_{t\alpha}=\nabla_{\theta_2}m_{t\alpha}$. Define
\[
\begin{aligned}
U_{nt}(\alpha,\gamma)&:=w_t(\gamma)r_{t,\alpha},
& V_{nt}(\alpha)&:=g_{t\alpha}r_{t,\alpha},\\
A_j(\alpha,\gamma)&:=E[f_{t\alpha}w_t(\gamma)\nabla_{\theta_j}m_{t\alpha}],
& j&=1,2,\\
D_\alpha&:=E[g_{t\alpha}f_{t\alpha}\nabla_{\theta_1}^{\top}m_{t\alpha}],
& H_\alpha&:=-L_\alpha^{-1}D_\alpha.
\end{aligned}
\]
Then, uniformly over the indicated indices,
\[
\begin{aligned}
n^{-1/2}\sum_{t=1}^n E_{n,B}U_{nt}(\alpha,\gamma)
 &=A_1(\alpha,\gamma)^\top B(\alpha)+o(1),\\
n^{-1/2}\sum_{t=1}^n E_{n,B}V_{nt}(\alpha)
 &=D_\alpha B(\alpha)+o(1).
\end{aligned}
\]
The nuisance estimator under the null restriction satisfies
\[
\Delta_{2n}(\alpha)
:=
\sqrt n\{\hat\theta_{2,0n}(\alpha)-\theta_{20}(\alpha)\}
=
\mathbb L_{n,B}(\alpha)+H_\alpha B(\alpha)+o_p(1),
\]
where
\[
\mathbb L_{n,B}(\alpha)
:=
-L_\alpha^{-1}
\frac1{\sqrt n}\sum_{t=1}^n
\big\{V_{nt}(\alpha)-E_{n,B}V_{nt}(\alpha)\big\}.
\]
The pre-estimation drift \(\widetilde d_B\equiv d_{\mathrm{pre},B}\) in
\eqref{eq:pre-local-drift} is therefore
\[
\widetilde d_B(\alpha,\gamma)
=
A_1(\alpha,\gamma)^\top B(\alpha)
+
A_2(\alpha,\gamma)^\top H_\alpha B(\alpha).
\]
\end{lemma}

\begin{proof}
Consider each coordinate $a$ of the drift coefficients
\begin{gather*}
e^{\gamma^\top i}f_0(m(i,\vartheta)\mid i)\nabla_\vartheta m(i,\vartheta),\\
w(i,\gamma)f_0(m(i,\vartheta)\mid i)\nabla_{\vartheta_1}m(i,\vartheta),\\
g(i,\vartheta)f_0(m(i,\vartheta)\mid i)
\nabla_{\vartheta_1}^{\top}m(i,\vartheta)
\end{gather*}
as a function of its compact parameter index. It is continuous and
$|a|\le C\mathcal E_t^2$; the last two classes use the null-centered weight
or the nuisance gradient. For $a^M=a\mathbf 1\{\mathcal E_t\le M\}$ and
$P_na=n^{-1}\sum_t a(I_{t-1})$, Lemma~\ref{lem:uwlln} and contiguity give
$\sup_a|P_na^M-E_0a^M|=o_{P_{n,B}}(1)$. This supremum is bounded by $2CM^2$.
The moment bounds in Assumptions~\ref{ass:nuisance-smooth}(iii) and
\ref{ass:local}(ii) therefore yield
\[
\sup_a|E_{n,B}P_na-E_0a|
\le E_{n,B}\sup_a|P_na^M-E_0a^M|+CM^{2-r}=o(1)+CM^{2-r}.
\]
Let $n\to\infty$, then $M\to\infty$. The bound remains uniform along the
null curves after multiplication by the bounded directions $B(\alpha)$.

Next write $q_{nt}(\alpha)$ for the true local conditional quantile and
$q_t^0(\alpha)$ for its null counterpart. Put
$v_t(\alpha)=\nabla_\theta m(I_{t-1},\theta_0(\alpha))^\top B(\alpha)$,
using $\nabla_{\theta_1}m$ in the pre-estimation case.
Writing $\mathcal L_t=\mathcal L(I_{t-1})$, the integral Taylor formula and
Assumption~\ref{ass:nuisance-smooth}(i) give
\[
q_t^0(\alpha)-q_{nt}(\alpha)
=n^{-1/2}v_t(\alpha)+R_{nt}(\alpha),
\qquad |R_{nt}(\alpha)|\le Cn^{-1}\mathcal L_t,
\]
and $|q_t^0(\alpha)-q_{nt}(\alpha)|\le Cn^{-1/2}\mathcal E_t$,
uniformly in the quantile and tangent indices. On $\mathcal E_t\le M$,
$\sqrt n|R_{nt}|\le CM/\sqrt n$.
Since $F_{n,B,t}(q_{nt}(\alpha)\mid I_{t-1})=\alpha$, define
\[
\rho_{nt}(\alpha)
:=F_{n,B,t}(q_t^0(\alpha)\mid I_{t-1})-\alpha
  -n^{-1/2}f_0(q_t^0(\alpha)\mid I_{t-1})v_t(\alpha).
\]
For $b_t=e^{\gamma^\top I_{t-1}}$, $w_t(\gamma)$, or a coordinate of
$g_{t\alpha}$, $|b_t|\le C\mathcal E_t$. Choose $c>0$ so that the segment
from $q_{nt}$ to $q_t^0$ lies in the density neighborhood whenever
$\mathcal E_t\le c\sqrt n$. Uniformly in the indices, for fixed $M$,
\[
\sqrt n\,|b_t\rho_{nt}|\le
\begin{cases}
o(1), & \mathcal E_t\le M,\\
C\mathcal E_t^2, & M<\mathcal E_t\le c\sqrt n,\\
C(\sqrt n\,\mathcal E_t+\mathcal E_t^2), & \mathcal E_t>c\sqrt n .
\end{cases}
\]
The three bounds use, respectively, local density convergence and the bounded
null density derivative, the uniform density bound, and $0\le F_{n,B,t}\le1$.
Thus, by the local $r$th-moment bound,
\[
\sup_{\text{indices}}n^{-1/2}\sum_{t=1}^n E_{n,B}|b_t\rho_{nt}|
\le o(1)+CM^{2-r}+Cn^{(2-r)/2}.
\]
Letting $n\to\infty$ and then $M\to\infty$, and using the coefficient limit
above, gives the three local mean expansions.

The null expansion in Assumption~\ref{ass:pre-est-regularity} is a function
of the same observations. Contiguity transfers its uniform remainder, so
\begin{align*}
\Delta_{2n}(\alpha)
&=-L_\alpha^{-1}\frac1{\sqrt n}\sum_{t=1}^n V_{nt}(\alpha)+o_{P_{n,B}}(1)\\
&=\mathbb L_{n,B}(\alpha)-L_\alpha^{-1}D_\alpha B(\alpha)
 +o_{P_{n,B}}(1).
\end{align*}
Here $L_\alpha^{-1}$ is uniformly bounded. Contiguity likewise transfers the
uniform plug-in expansion: its nuisance term $A_2^\top\Delta_{2n}$ contributes
$A_2^\top H_\alpha B$ in addition to the direct drift $A_1^\top B$.
\end{proof}

\begin{lemma}
\label{lem:VC}
Under Assumption~\ref{ass:regularity}(i)--(ii),
\[
\mathcal{G} = \Big\{ (i,y) \mapsto \exp(\gamma^{\top}i) \big[ \mathbf{1}\{y \le m(i,\theta(\alpha))\} - \alpha \big] \;\Big|\; \gamma \in \Gamma, \alpha \in \mathcal{T}, \theta \in \Theta \Big\}
\]
is a VC-subgraph class.
\end{lemma}

\begin{proof}
Consider the unrestricted class
\[
\widetilde{\mathcal G}
=
\left\{
(i,y)\mapsto
\exp(\gamma^\top i)
\left[\mathbf 1\{y\le m(i,\theta)\}-\alpha\right]:
\gamma\in\Gamma,\ \alpha\in\mathcal T,\ \theta\in\Theta
\right\}.
\]
It suffices to prove that $\widetilde{\mathcal G}$ is VC-subgraph, since
$\mathcal G\subset\widetilde{\mathcal G}$ restricts $\theta$ to $\theta(\alpha)$.

Fix $\psi\in\widetilde{\mathcal G}$ and let
$S_\psi=\{(i,y,t):t<\psi(i,y)\}$. With
$A_\theta=\{(i,y):y\le m(i,\theta)\}$, the two values of the indicator give
\[
S_\psi
=
\big[(A_\theta\times\mathbb R)\cap C^+_{\gamma,\alpha}\big]
\cup
\big[(A_\theta^c\times\mathbb R)\cap C^-_{\gamma,\alpha}\big],
\]
where
\begin{align*}
C^+_{\gamma,\alpha}
&=\{(i,y,t):t<(1-\alpha)\exp(\gamma^\top i)\}\\
&=\{t\le0\}\cup
\{t>0:\log t-\gamma^\top i-\log(1-\alpha)<0\},\\
C^-_{\gamma,\alpha}
&=\{(i,y,t):t<-\alpha\exp(\gamma^\top i)\}\\
&=\{t<0:\log\alpha+\gamma^\top i-\log(-t)<0\}.
\end{align*}

By Assumption~\ref{ass:regularity}(ii), the cylinders over $A_\theta$ and
their complements are VC. On each half-axis the logarithmic threshold is
a halfspace in $(i,\log|t|,1)$; $\alpha\in(0,1)$ ensures both logarithms
of quantile levels exist. The two threshold classes, and hence the displayed
finite Boolean combination defining $S_\psi$, are VC.
\end{proof}

\begin{lemma}[Adapted from \citet{arconesyu1994}, Theorem 2.1]
\label{lem:mixing-fclt}
Let $X_1,X_2,\ldots$ be stationary in a Polish space with marginal law $P$,
and let $\mathcal F\subset L_2(P)$. Suppose, for some $q>2$,
\begin{itemize}
    \item[(i)] The $\beta$-mixing coefficients satisfy
    \[
    \lim_{k \to \infty} k^{\frac{q}{q-2}} (\log k)^{\frac{2(q-1)}{q-2}} \, \beta_X(k) = 0.
    \]
    \item[(ii)] $\mathcal{F}$ is permissible, VC-subgraph, and has an envelope $F$ such that $P^* F^q < \infty$.
\end{itemize}
Then
\[
\mathbb{G}_n f := \frac{1}{\sqrt{n}} \sum_{t=1}^n (f(X_t) - P f) \rightsquigarrow \mathbb{G} f \quad \text{in } \ell^\infty(\mathcal{F}),
\]
where $\mathbb G$ is tight, centered Gaussian with covariance
\[
Cov(\mathbb{G} f, \mathbb{G} g) = \sum_{k=-\infty}^{\infty} Cov(f(X_0), g(X_k)) \quad \text{for all } f,g \in \mathcal{F}.
\]
\end{lemma}

\begin{lemma}[Adapted from \citet{neweymcfadden1994}, Lemma 2.4]
\label{lem:uwlln}
Let $\{Z_t\}$ be strictly stationary and ergodic and $\Theta$ compact.
Suppose $a(\cdot,\theta)$ is measurable for each $\theta$,
$a(Z_t,\cdot)$ is almost surely continuous, and
\[
\|a(Z_t,\theta)\|\le d(Z_t)\quad\text{for all }\theta\in\Theta,
\qquad E[d(Z_t)]<\infty.
\]
Then
\[
\sup_{\theta\in\Theta}
\left|
\frac{1}{n}\sum_{t=1}^{n} a(Z_t,\theta)
-
\mathbb{E}[a(Z_t,\theta)]
\right|
\overset{p}{\longrightarrow} 0.
\]
\end{lemma}

\begin{proof}
Write $P_nu=n^{-1}\sum_{t=1}^n u(Z_t)$ and $Pu=E[u(Z_t)]$.
For $\theta\in\Theta$ and $\delta>0$, let
\[
b_{\theta,\delta}(z)
:=
\sup_{\vartheta\in\Theta:\|\vartheta-\theta\|<\delta}
\|a(z,\vartheta)-a(z,\theta)\|.
\]
Separability gives measurability, and continuity and $b_{\theta,\delta}\le2d$
give $P b_{\theta,\delta}\downarrow0$ as $\delta\downarrow0$. By compactness,
choose a finite cover with centers $\theta_j$ and radii $\delta_j$ such that
$P b_j<\varepsilon$, where $b_j=b_{\theta_j,\delta_j}$. Then
\[
\begin{aligned}
\sup_{\vartheta\in\Theta}\|(P_n-P)a(\cdot,\vartheta)\|
&\le\max_j\{\|(P_n-P)a(\cdot,\theta_j)\|+(P_n+P)b_j\},\\
\limsup_{n\to\infty}\sup_{\vartheta\in\Theta}\|(P_n-P)a(\cdot,\vartheta)\|
&\le 2\max_j P b_j\le2\varepsilon\quad\text{a.s.}
\end{aligned}
\]
The second line uses the ergodic theorem on the finite cover; let
$\varepsilon\downarrow0$.
\end{proof}

\begin{lemma}[Conditional Gaussian multiplier convergence]
\label{lem:conditional-gaussian-multiplier}
Let
\[
g_{\alpha,\gamma}(Y_t,I_{t-1})
:=
\exp(\gamma^\top I_{t-1})
\{\mathbf 1[Y_t\le m(I_{t-1},\theta_0(\alpha))]-\alpha\},
\qquad
(\alpha,\gamma)\in\mathcal T\times\Gamma,
\]
and let $\mathcal G_0$ be the resulting class. Under
Assumptions~\ref{ass:mixing}--\ref{ass:regularity} and $H_0$, define
\[
\rho_n^2(g,h):=\frac1n\sum_{t=1}^n\{g(Z_t)-h(Z_t)\}^2,
\qquad
\rho^2(g,h):=E\{g(Z_t)-h(Z_t)\}^2.
\]
Then
\begin{equation}
\sup_{g,h\in\mathcal G_0}
|\rho_n^2(g,h)-\rho^2(g,h)|\xrightarrow{p}0.
\label{eq:sample-semimetric}
\end{equation}
Moreover, for i.i.d. $N(0,1)$ multipliers independent of the data,
\[
\mathbb G_n^*g:=\frac1{\sqrt n}\sum_{t=1}^n\omega_t g(Z_t)
\rightsquigarrow^*\mathcal M(g)
\quad\text{in }\ell^\infty(\mathcal G_0)
\]
conditionally on the data in probability, where $\mathcal M$ is the same
Gaussian limit as the centered original empirical process.
\end{lemma}

\begin{proof}
Write $P_nu=n^{-1}\sum_{t=1}^n u(Z_t)$ and $Pu=E[u(Z_t)]$, and put
$x=\Phi(I)$, $w_\gamma=\exp(\gamma^\top x)$, and
$B_\alpha=\mathbf 1\{Y\le m(I,\theta_0(\alpha))\}$.
By Lemma~\ref{lem:VC}, $\mathcal G_0$ is a permissible VC-subgraph class with
envelope $F=\exp(\bar\gamma\|x\|_2)$ satisfying $PF^q<\infty$ for the
$q>2$ in Assumption~\ref{ass:mixing}. The nested indicators give the identity
\[
g_{\alpha,\gamma}g_{\beta,\eta}
=e^{(\gamma+\eta)^\top x}
\{B_{\min(\alpha,\beta)}-\beta B_\alpha-\alpha B_\beta+\alpha\beta\}.
\]
On $D_R=\{\|x\|_2\le R\}$ the weights $e^{\nu^\top x}$,
$\nu\in\Gamma+\Gamma$, admit a finite uniform net by boundedness and
Lipschitz continuity. Cover $\mathcal T$ by finitely many intervals $[a,b]$
with endpoints in $\mathcal T$ and $b-a\le\delta$. For $a\le\alpha\le b$,
$B_a\le B_\alpha\le B_b$ and
\[
E\!\left[e^{\nu^\top x}\mathbf 1_{D_R}(B_b-B_a)\right]
=(b-a)E\!\left[e^{\nu^\top x}\mathbf 1_{D_R}\right]
\le\delta e^{2\bar\gamma R}.
\]
The net and these brackets give arbitrarily small $L_1(P)$ widths. Apply
the ergodic theorem at the finite endpoints; the omitted tail is bounded by
$(P_n+P)[F^2\mathbf 1_{D_R^c}]$, with $P[F^2\mathbf 1_{D_R^c}]\to0$.
The same bound applies to the weights alone. The product identity then gives
\[
\sup_{g,h\in\mathcal G_0}|\rho_n^2(g,h)-\rho^2(g,h)|
\le 4\sup_{g,h\in\mathcal G_0}|(P_n-P)gh|=o_p(1),
\]
as required in \eqref{eq:sample-semimetric}. This bracketing argument uses
the nested indicators in the raw scores.

Conditional on the data, $\mathbb G_n^*$ is a centered Gaussian process with
canonical semimetric $\rho_n$. Predictability and
Assumption~\ref{ass:mixing}(iii) eliminate all cross-time score covariances,
so the preceding covariance law proves conditional finite-dimensional
convergence to $\mathcal M$. For equicontinuity, use the VC bound, valid for
every probability measure $Q$ with $QF^2<\infty$:
\[
N\!\left(\varepsilon\|F\|_{Q,2},\mathcal G_0,L_2(Q)\right)
\le (A/\varepsilon)^v,
\qquad 0<\varepsilon\le1,
\]
for constants $A,v$ independent of $Q$. Put
\[
\eta_n:=\sup_{g,h\in\mathcal G_0}|\rho_n^2(g,h)-\rho^2(g,h)|,
\qquad L_n:=(P_nF^2)^{1/2},
\qquad r_{n,\delta}:=(\delta^2+\eta_n)^{1/2}.
\]
Every member of
$\mathcal D_\delta:=\{g-h:\rho(g,h)\le\delta\}$ has empirical $L_2$ norm
at most $r_{n,\delta}$. Taking products of covers of $\mathcal G_0$ gives
\[
N(\varepsilon,\mathcal D_\delta,L_2(P_n))
\le N(\varepsilon/4,\mathcal G_0,L_2(P_n))^2.
\]
Conditional on the data, the process is continuous on the evaluation vectors
in $\mathbb R^n$, hence separable on their closure. Since $0\in\mathcal D_\delta$,
$\mathbb G_n^*0=0$ and the empirical diameter is at most $2r_{n,\delta}$,
the Gaussian entropy inequality \citep[Theorem~2.2.4 and Corollary~2.2.5]{vdvw1996} gives
\[
E^*\!\left[\sup_{\rho(g,h)\le\delta}
|\mathbb G_n^*g-\mathbb G_n^*h|\right]
\le C\int_0^{2r_{n,\delta}}
\sqrt{\log(2+A L_n/\varepsilon)}\,d\varepsilon.
\]
On $\{L_n\le L\}$ with $L$ fixed the bound vanishes in probability as first
$n\to\infty$, then $\delta\downarrow0$, since $\eta_n=o_p(1)$.
Now $L_n=O_p(1)$, so conditional Markov gives equicontinuity in probability.
The entropy bound at $Q=P$ gives total boundedness, completing the conditional
convergence argument.

For measurability of the integrated functionals, note that
\[
\rho(g_{\alpha,\gamma},g_{\beta,\eta})
\le \|w_\gamma-w_\eta\|_{P,2}
+(PF^2)^{1/2}|\alpha-\beta|^{1/2},
\]
because
$E[\{(B_\alpha-\alpha)-(B_\beta-\beta)\}^2\mid I]
=|\alpha-\beta|-(\alpha-\beta)^2$.
Dominated convergence and compactness give uniform $L_2(P)$ continuity of
$\gamma\mapsto w_\gamma$. The $\rho$-uniformly continuous Gaussian version
from the entropy bound thus has Euclidean-continuous paths on $\mathcal T\times\Gamma$.
The sample and multiplier $\alpha$-sections are measurable by nesting, so all
three processes take values in $\mathbb B_\Pi$ of Lemma~\ref{lem:gaussian-quantile-crossing}.
\end{proof}

The smoother first estimates the reciprocal derivative of a smooth reference
curve. When this curve is the true conditional quantile on \(\mathcal T\),
the target is
\[
f_{t\alpha}
:=
f_{Y_t\mid I_{t-1}}\!\big(m(I_{t-1},\theta_0(\alpha))\mid I_{t-1}\big).
\]
Following \citet{escanciano2019}, use $\mathcal A$, $L_{\mathrm{sim}}$, and
$\{U_j\}$ from Assumption~\ref{ass:regularity}(iv)(4) and, for $h>0$, define
\[
\widehat f_h(i,\vartheta;\alpha)
:=
\frac{L_{\mathrm{sim}}}{J h}
\sum_{j=1}^{J}
K\!\left(
\frac{m(i,\vartheta(\alpha))-m(i,\vartheta(U_j))}{h}
\right).
\]
The estimator uses $h=h_J$.
For a deterministic parameter curve $\vartheta^r$ on $\mathcal A$, let
\[
q_t^r(u):=m(I_{t-1},\vartheta^r(u)),
\qquad
s_t^r(\alpha):=\frac{1}{\partial_uq_t^r(\alpha)}.
\]
\begin{lemma}[Uniform rate along a smooth reference curve]
\label{lem:unif_rate}
Suppose Assumptions~\ref{ass:mixing}--\ref{ass:regularity} and
\ref{ass:kernel}(a)--(c) hold. Almost surely, $q_t^r\in C^3(\mathcal A)$
is strictly increasing and, uniformly in $t$,
\[
0<c_r\le \partial_u q_t^r(u)\le C_r<\infty,\qquad
\sup_{z\in q_t^r(\mathcal A)}
  |(\widetilde s_t^r)^{(k)}(z)|\le C,\quad k=1,2,
\]
where $\widetilde s_t^r(z):=\partial_z(q_t^r)^{-1}(z)$. For
$r_J\to\infty$ (possibly depending on $J,n$), set
\begin{equation}
u_J := \sqrt{\frac{r_J}{J a_J}}.
\label{eq:U}
\end{equation}

If, for a sufficiently large constant $C_0>0$,
\begin{equation}
r_J \ge C_0\{\log J+\log n+\log(1/a_J)\},
\qquad
\frac{r_J}{J a_J}\to0,
\label{eq:D}
\end{equation}
then
\begin{equation}
\sup_{h \in [a_J, b_J]} \sup_{\alpha \in \mathcal{T}} \max_{1 \le t \le n}
\left| \widehat{f}_h(I_{t-1}, \vartheta^r;\alpha) - s_t^r(\alpha) \right|
= O_p\big(b_J^2 + u_J\big).
\label{eq:R}
\end{equation}
For the two oracle curves $\bar\theta$ and $(\bar\theta_1,\theta_{20})$,
$s_t^r(\alpha)=f_{t\alpha}$ by Assumptions~\ref{ass:regularity}(iv)(4)
and~\ref{ass:pre-est-regularity}, respectively.
\end{lemma}

\begin{proof}[Proof of Lemma \ref{lem:unif_rate}]
For either oracle curve, differentiation of the maintained quantile
restriction on its specified neighborhood gives
\[
F_{Y_t\mid I_{t-1}=i}\{q_i^r(u)\}=u
\quad\Longrightarrow\quad
f_{Y_t\mid I_{t-1}=i}\{q_i^r(u)\}\partial_u q_i^r(u)=1.
\]
For a general reference curve, condition on the data and put
$L=L_{\mathrm{sim}}$, $y=q_t^r(\alpha)$, and
\[
Z_{j,t}(\alpha,h)
:=\frac{L}{h}K\!\left(\frac{q_t^r(\alpha)-q_t^r(U_j)}{h}\right),
\qquad
\Delta_t(\alpha,h)
:=J^{-1}\sum_{j=1}^J
\{Z_{j,t}(\alpha,h)-E_{\mathcal U}Z_{j,t}(\alpha,h)\}.
\]
Thus $|\widehat f_h-s_t^r(\alpha)|\le B_{t,\alpha,h}+|\Delta_t(\alpha,h)|$,
where $B_{t,\alpha,h}:=|E_{\mathcal U}\widehat f_h-s_t^r(\alpha)|$.

The change of variable $z=q_t^r(u)$ gives
\[
E_{\mathcal U}\widehat f_h
=\frac1h\int_{q_t^r(a_1)}^{q_t^r(a_2)}
K\!\left(\frac{y-z}{h}\right)\widetilde s_t^r(z)\,dz.
\]
The derivative lower bound and $\mathcal T\subset\operatorname{int}(\mathcal A)$
give a uniform distance from $y$ to the image endpoints. The boundary term is
zero for a compactly supported kernel when $J$ is large, and
$O(e^{-c/h})=o(h^2)$ for an exponentially decaying kernel. Taylor expansion,
kernel symmetry, and the uniformly bounded second derivative give
\[
E_{\mathcal U}\widehat f_h
=\widetilde s_t^r(y)+O(h^2)=s_t^r(\alpha)+O(h^2).
\]
Thus
\begin{equation}
\sup_{t,\alpha,h}B_{t,\alpha,h}=O(b_J^2).
\label{eq:bias_bound}
\end{equation}

Uniformly in $(t,\alpha,h)$,
\[
|Z_{j,t}(\alpha,h)|\le C/a_J,
\qquad
E_{\mathcal U}Z_{j,t}(\alpha,h)^2
\le\frac{L}{h}\sup_z\widetilde s_t^r(z)\int K^2(v)\,dv
\le C/a_J.
\]
Since $u_J\to0$, Bernstein's inequality yields
\begin{equation}
P_{\mathcal U}\!\left\{|\Delta_t(\alpha,h)|>u_J/2\right\}
\le 2e^{-cr_J}.
\label{eq:pointwise_prob}
\end{equation}
The constants are uniform in the conditioned data.

Boundedness of $K'$, $vK'(v)$, and $\partial_uq_t^r$ gives
\[
\sup_{j,t,\alpha,h}
\|\nabla_{(\alpha,h)}Z_{j,t}(\alpha,h)\|
\le C a_J^{-2}.
\]
The same bound, after enlarging $C$, applies to $\Delta_t$. Cover
$\mathcal T\times[a_J,b_J]$ by a grid with mesh
$\delta_J=u_Ja_J^2/(2C)$ and cardinality
\[
N_J\le C(1+\delta_J^{-1})^2
\le \frac{CJ}{r_Ja_J^3}
\]
for all large $J$. The interpolation error is at most $u_J/2$, so
\eqref{eq:pointwise_prob} and a union bound give
\[
P_{\mathcal U}\!\left\{
\max_{t\le n}\sup_{\alpha,h}|\Delta_t(\alpha,h)|>u_J
\right\}
\le 2nN_Je^{-cr_J}.
\]
Since $r_J\ge1$ eventually,
\[
\log N_J\le C\{\log J+\log(1/a_J)\},
\qquad 2nN_Je^{-cr_J}\to0
\]
by \eqref{eq:D} for sufficiently large $C_0$. Uniformity in the data gives
the same bound unconditionally. Combine it with \eqref{eq:bias_bound} to
obtain \eqref{eq:R}.
\end{proof}

\begin{lemma}[Completeness of exponential weights]\label{lem:bierens}
Under Assumptions~\ref{ass:mixing}--\ref{ass:regularity}, define
\[
\eta(\alpha,\gamma;\theta)
= \mathbb{E}\!\left[\exp(\gamma^\top \Phi(I_{t-1}))\big( \mathbf{1}\{Y_t \le m(I_{t-1},\theta(\alpha))\} - \alpha \big)\right],
\]
where $\Phi$ is specified in Assumption~\ref{ass:regularity}(i). If
$\theta(\cdot)=\theta_0(\cdot)$, then $\eta(\alpha,\gamma;\theta)=0$ on
$\mathcal T\times\Gamma$. If the null restriction fails at some
$\alpha\in\mathcal T$, then $\eta(\alpha,\gamma;\theta)\neq0$ for
Lebesgue-a.e. $\gamma\in\operatorname{int}(\Gamma)$.
For bounded $\Phi$, the conclusion follows from \citet{bierens1990},
as restated in \citet{stinchcombe1998consistent}. The proof below also covers
the identity map under the exponential moment condition in
Assumption~\ref{ass:regularity}(i).
\end{lemma}

\begin{proof}
Fix $\alpha\in\mathcal T$. Write $I=I_{t-1}$ and
$r_\alpha(I)=P\{Y_t\le m(I,\theta(\alpha))\mid I\}-\alpha$.
Under the null, $r_\alpha=0$ a.s., hence $\eta=0$. For the converse,
let $\mu_\alpha$ be the finite signed measure on $\mathbb R^{d_I}$ and
$L_\alpha$ its bilateral Laplace transform given by
\[
\mu_\alpha(A):=E[r_\alpha(I)\mathbf 1\{\Phi(I)\in A\}],\qquad
L_\alpha(\gamma):=\int e^{\gamma^\top x}\,d\mu_\alpha(x)
=\eta(\alpha,\gamma;\theta).
\]
The transform is real analytic on $\mathbb R^{d_I}$ for bounded $\Phi$ and
on a connected open neighborhood of $\Gamma$ under the exponential moment
in Assumption~\ref{ass:regularity}(i) for the identity map. By analytic
continuation and uniqueness of the Laplace transform,
\[
L_\alpha|_{\operatorname{int}(\Gamma)}\equiv0
\ \Longrightarrow\ \mu_\alpha=0
\ \Longrightarrow\ E[r_\alpha(I)\mid\Phi(I)]=0\quad\text{a.s.}
\]
A Borel injection between standard Borel spaces has a measurable inverse
on its image, so $\sigma(\Phi(I))=\sigma(I)$. The last equality therefore
forces $r_\alpha=0$ a.s., hence $\theta(\alpha)=\theta_0(\alpha)$ by
Assumption~\ref{ass:mixing}(ii). Otherwise $L_\alpha$ is nontrivial and
its zero set has Lebesgue measure zero.
\end{proof}

\begin{lemma}[Strict crossing of Gaussian null quantiles]
\label{lem:gaussian-quantile-crossing}
Let $\Pi$ be a probability measure on $\mathcal T$, and let
$\mathbb B_\Pi$ be the closed subspace of $\ell^\infty(\mathcal T\times\Gamma)$
consisting of functions whose $\alpha$-sections are $\Pi$-measurable for every
$\gamma$. The integrated processes are understood on this space.
Measurability is automatic for $\Pi_A$ on a finite grid.
Let $Q$ be a tight centered Gaussian random element in $\mathbb B_\Pi$.
For fixed $\lambda\ge0$, put
\[
\phi_\lambda(f):=\sup_{\gamma\in\Gamma}
\left\{\int_{\mathcal T}f(\alpha,\gamma)^2d\Pi(\alpha)
-\lambda\|\gamma\|_1\right\},
\qquad T_\lambda:=\phi_\lambda(Q).
\]
If $F_\lambda$, the distribution function of $T_\lambda$, is continuous at
$c_\lambda:=q_p(T_\lambda)$ for $0<p<1$, then
\begin{equation}
F_\lambda(c_\lambda-\varepsilon)<p
<F_\lambda(c_\lambda+\varepsilon),\qquad \varepsilon>0.
\label{eq:null-quantile-crossing}
\end{equation}
This also holds with a maximum over a fixed finite direction grid.
\end{lemma}

\begin{proof}
The tight Borel law of $Q$ is carried by a countable union of compact sets,
whose closed linear span is separable. In this subspace let $S$ be the support.
For independent copies $Q_1,Q_2$,
\[
(Q_1+Q_2)/\sqrt2\stackrel d=Q,
\qquad -Q\stackrel d=Q.
\]
Every neighborhood of $x,y\in S$ has positive probability. Independence
therefore gives
\[
(x+y)/\sqrt2,-x\in S
\quad\Longrightarrow\quad
0,x/\sqrt2,\sqrt2x,x+y\in S.
\]
Dyadic scaling and closedness give all real multiples of $x$. Thus $S$ is a
closed linear subspace, and hence connected.

The bound
\[
|\phi_\lambda(f)-\phi_\lambda(g)|
\le(\|f\|_\infty+\|g\|_\infty)\|f-g\|_\infty
\]
gives continuity. Since $Q\in S$ a.s.\ and each neighborhood of a support
point has positive probability, $\operatorname{supp}(T_\lambda)=
\overline{\phi_\lambda(S)}$ is an interval. At $c_\lambda$, continuity gives
$F_\lambda(c_\lambda)=p$ and $P(T_\lambda=c_\lambda)=0$. Thus
$0<p<1$ places $c_\lambda$ in the interior of the interval, where every
open subinterval has positive mass. This proves \eqref{eq:null-quantile-crossing};
the same continuity bound applies to a fixed finite grid.
\end{proof}
\section*{S3. Main results}

Throughout Sections S3--S4, $\Pi$ is a fixed probability measure on
$\mathcal T$. Write $\|\cdot\|_\infty$ for the supremum norm on
$\mathcal T\times\Gamma$ and use the functional from
Lemma~\ref{lem:gaussian-quantile-crossing}:
\[
\phi_\lambda(f)=\sup_{\gamma\in\Gamma}
\left\{\int_{\mathcal T}f(\alpha,\gamma)^2\,d\Pi(\alpha)-\lambda\|\gamma\|_1\right\},
\qquad f\in\mathbb B_\Pi.
\]
A fixed grid $\mathcal T_A=\{\alpha_1,\ldots,\alpha_A\}$ corresponds to
$\Pi_A=A^{-1}\sum_{j=1}^A\delta_{\alpha_j}$.
For pre-estimation, $\Gamma$ is the compact set of centered-weight
directions in Assumption~\ref{ass:pre-est-regularity}(iv).
In the statements below, moment processes converge in
$\ell^\infty(\mathcal T\times\Gamma)$; starred limits are conditional on
the data, in probability.

\begin{theorem}[Asymptotic null distribution]
\label{thm:null}
Under $H_0$ and Assumptions~\ref{ass:mixing}--\ref{ass:regularity}:
\begin{enumerate}
\item $\sqrt n M_n\rightsquigarrow\mathcal M$, where $\mathcal M$ is centered Gaussian with covariance
\[
\operatorname{Cov}\{\mathcal M(\alpha_1,\gamma_1),\mathcal M(\alpha_2,\gamma_2)\}
=(\min\{\alpha_1,\alpha_2\}-\alpha_1\alpha_2)
E e^{(\gamma_1+\gamma_2)^\top I_{t-1}}.
\]
\item $\|s_n^2-\sigma^2\|_\infty\xrightarrow{p}0$, where
$\sigma^2(\alpha,\gamma)=\alpha(1-\alpha)E e^{2\gamma^\top I_{t-1}}$.
\item $T_{n,\Pi}(\lambda)\rightsquigarrow
T_\Pi(\lambda):=\phi_\lambda(\mathcal M/\sigma)$.
\end{enumerate}
\end{theorem}

\begin{theorem}[Multiplier bootstrap validity]
\label{thm:bootstrap}
Under $H_0$ and Assumptions~\ref{ass:mixing}--\ref{ass:regularity}, use
i.i.d.\ $N(0,1)$ multipliers independent of the data in
\eqref{eq:Mn_star}--\eqref{eq:Tn_star}. With the limits from Theorem~\ref{thm:null}:
\begin{enumerate}
\item $\sqrt n M_{n*}\overset{*}{\rightsquigarrow}\mathcal M$;
\item $\|s_{n*}^2-\sigma^2\|_\infty\xrightarrow{p^*}0$;
\item $T_{n*,\Pi}(\lambda)\overset{*}{\rightsquigarrow}T_\Pi(\lambda)$.
\end{enumerate}
\end{theorem}

\begin{corollary}[Validity on a fixed direction grid]
\label{cor:fixed-direction-grid}
For fixed $\Gamma_G=\{\gamma_1,\ldots,\gamma_G\}\subset\Gamma$, write
$\phi_{\lambda,G}$, $T_{n,\Pi,G}$ and $T_{n*,\Pi,G}$ for the functional and
statistics with $\sup_{\gamma\in\Gamma}$ replaced by $\max_{\gamma\in\Gamma_G}$.
Under Theorems~\ref{thm:null} and~\ref{thm:bootstrap},
\[
T_{n,\Pi,G}(\lambda)\rightsquigarrow T_{\Pi,G}(\lambda),\qquad
T_{n*,\Pi,G}(\lambda)\overset{*}{\rightsquigarrow}T_{\Pi,G}(\lambda),
\]
where $T_{\Pi,G}(\lambda)=\phi_{\lambda,G}(\mathcal M/\sigma)$.
Continuity at its $(1-\tau)$-quantile gives
asymptotic size $\tau$ for the ideal bootstrap test. Under
Theorems~\ref{thm:null-pre} and~\ref{thm:boot-pre}, the same conclusions hold
with $\widetilde{\mathcal M}/\tilde\sigma$ if the grid excludes the degenerate
centered direction and satisfies Assumption~\ref{ass:pre-est-regularity}(iv).
\end{corollary}

With $\Pi=\Pi_A$, the statistic, its bootstrap law and its power criterion
all use $\mathcal T_A\times\Gamma_G$. Neither grid need become dense;
consistency requires a nonzero signal on the grid.

\begin{corollary}[Size control for a fixed penalty]
\label{cor:fixed-lambda-size}
Fix a deterministic $\lambda\ge0$, and let
$c_{n,1-\tau,\Pi}^*(\lambda)$ be the ideal conditional bootstrap
$(1-\tau)$-quantile. Under Theorem~\ref{thm:bootstrap}, if the distribution of
$T_\Pi(\lambda)$ is continuous at its $(1-\tau)$-quantile, then
\[
P\{T_{n,\Pi}(\lambda)>c_{n,1-\tau,\Pi}^*(\lambda)\}\to\tau.
\]
The same holds for the plug-in test under Theorem~\ref{thm:boot-pre}.
\end{corollary}

\begin{corollary}[Validity for adaptive penalty selection]
\label{cor:adaptive-lambda}
Suppose the bootstrap approximation in Theorem~\ref{thm:bootstrap} is
uniform on compact $\Lambda=[0,\bar\lambda]$ and $\widehat\lambda\to_p\lambda_0$ for
deterministic $\lambda_0\in\Lambda$. Write $c_{1-\tau,\Pi}(\lambda)$ for the
null-limit quantile. Assume it is continuous in $\lambda$ at $\lambda_0$
and $T_\Pi(\lambda_0)$ has no atom at $c_{1-\tau,\Pi}(\lambda_0)$.
For ideal bootstrap critical values or approximations with uniform
$o_p(1)$ simulation error,
\[
P\{T_{n,\Pi}(\widehat\lambda)>c_{n,1-\tau,\Pi}^*(\widehat\lambda)\}
\to\tau.
\]
The same holds for $\widehat T_{n,\Pi}$ under Theorem~\ref{thm:boot-pre}.
\end{corollary}

\begin{theorem}[Consistency under fixed alternatives]
\label{thm:consistency}
Suppose Assumptions~\ref{ass:mixing}--\ref{ass:regularity} hold. If
\[
\sup_{\gamma\in\Gamma}
\int_{\mathcal T}
\left(\frac{\eta(\alpha,\gamma;\bar\theta)}{\sigma(\alpha,\gamma)}\right)^2
\,d\Pi(\alpha)>0,
\]
where
\[
\eta(\alpha,\gamma;\bar\theta)
:=E\!\left[\exp(\gamma^\top I_{t-1})
\big(\mathbf 1\{Y_t\le m(I_{t-1},\bar\theta(\alpha))\}-\alpha\big)\right],
\]
then, for any fixed \(\lambda\ge0\),
\[
T_{n,\Pi}(\lambda) \xrightarrow{P} +\infty.
\]
\end{theorem}

\begin{theorem}[Oracle uniform consistency and rate of $\widehat f_h$]
\label{thm:consistency_density}
Under $H_0$, Assumptions~\ref{ass:mixing}--\ref{ass:regularity}, and
\ref{ass:kernel}(a)--(c), construct $\widehat f_h(I_{t-1},\bar\theta;\alpha)$
using Assumption~\ref{ass:regularity}(iv)(4). For $J\asymp n$, set
\[
(a_{J_n},b_{J_n})=(a_n,b_n)
=\left(\frac{\log n}{n}\right)^{1/5}(c_a,c_b),
\qquad 0<c_a\le c_b<\infty.
\]
Suppose $P(a_n\le h_J\le b_n)\to1$. With
$f_{t\alpha}=f_{Y_t\mid I_{t-1}}(m(I_{t-1},\theta_0(\alpha))\mid I_{t-1})$,
\[
\Delta_{f,n}:=
\sup_{\substack{a_n\le h\le b_n\\ \alpha\in\mathcal T}}\max_{1\le t\le n}
\left|\widehat f_h(I_{t-1},\bar\theta;\alpha)-f_{t\alpha}\right|
=O_p\!\left((\log n/n)^{2/5}\right)=o_p(1).
\]
\end{theorem}

Theorem~\ref{thm:consistency_density} applies to the known curve. Using the
restricted estimated curve additionally requires the replacement bound in the
proof of Theorem~\ref{thm:null-pre}.

\begin{theorem}[Multiplier bootstrap validity under local alternatives]
\label{thm:MB-local}
Suppose Assumptions~\ref{ass:mixing}--\ref{ass:kernel} and
\ref{ass:nuisance-smooth}--\ref{ass:local} hold, with the bandwidth of
Theorem~\ref{thm:consistency_density}. Fix an admissible bootstrap direction
$B$ and evaluate the statistic at $\theta_0$ under $H_0$ or any actual local
law $P_{n,B_0}$ in Assumption~\ref{ass:local}. With $\mathcal M$ from
Theorem~\ref{thm:null} and $d_B$ from Lemma~\ref{lem:local-drift-expansions}:
\begin{enumerate}
\item $\sqrt n M_{n*}(\cdot,\cdot;B)\rightsquigarrow^*\mathcal M+d_B$;
\item $\|s_{n*,B}^2-\sigma^2\|_\infty\xrightarrow{p^*}0$;
\item $T_{n*,B,\Pi}(\lambda)\rightsquigarrow^*
\phi_\lambda\big((\mathcal M+d_B)/\sigma\big)$.
\end{enumerate}
\end{theorem}

\noindent Define the standardized process, drift, and local limit by
\[
Q:=\mathcal M/\sigma,\qquad R_B:=d_B/\sigma,\qquad
T_\Pi(\lambda;B):=\phi_\lambda(Q+R_B).
\]

\begin{theorem}[Effect of penalization on local separation]
\label{thm:power-enhance}
Let $B(\alpha)\equiv B_0$ and define
\[
J_{\Pi}(\gamma;B)
:=
\int_{\mathcal T}\big(Q(\alpha,\gamma)+R_B(\alpha,\gamma)\big)^2d\Pi(\alpha).
\]
Fix a realization for which \(J_\Pi(\cdot;B)\) is continuous on the compact
set \(\Gamma\) and has a unique maximizer \(\tilde\gamma_\Pi(B)\) for both
\(B=0\) and some \(B_0\neq0\). If
\[
\|\tilde\gamma_{\Pi}(0)\|_1-
\|\tilde\gamma_{\Pi}(B_0)\|_1>0
\]
holds on this realization, there exists
$\delta=\delta(Q,B_0)>0$ such that, for every $\lambda\in(0,\delta)$,
\[
T_{\Pi}(\lambda;B_0)-T_{\Pi}(\lambda;0)
>
T_{\Pi}(0;B_0)-T_{\Pi}(0;0).
\]
\end{theorem}

For the next result, let $\Lambda=[0,\bar\lambda]$ with $0\in\Lambda$ and let
$\mathcal B=\{B_1,\ldots,B_J\}$ be a finite, nonempty collection of local
directions. Define
\[
c_{1-\tau,\Pi}(\lambda)
:=q_{1-\tau}\{T_\Pi(\lambda;0)\},
\qquad
\mathcal R(\lambda,B,\tau)
:=P\{T_\Pi(\lambda;B)>c_{1-\tau,\Pi}(\lambda)\},
\]
and $W(\lambda):=\min_{B\in\mathcal B}\mathcal R(\lambda,B,\tau)$.
Let $c_{n,1-\tau,\Pi}^*(\lambda)$ be the null multiplier bootstrap critical
value and $\widehat{\mathcal R}_n(\lambda,B,\tau)$ the corresponding estimate
of rejection probability from the shifted bootstrap. Set
\[
\widehat W_n(\lambda)
:=\min_{B\in\mathcal B}\widehat{\mathcal R}_n(\lambda,B,\tau).
\]
Choose a measurable $\widehat\lambda\in\Lambda$ satisfying
\begin{equation}
\widehat W_n(\widehat\lambda)
\ge\sup_{\ell\in\Lambda}\widehat W_n(\ell)-\eta_n,
\qquad 0\le\eta_n=o_p(1).
\label{eq:approximate-penalty-selector}
\end{equation}
An exact measurable maximizer has $\eta_n=0$.

\begin{theorem}[Maximin local power of the adaptive selector]
\label{thm:maximin-no-loss}
Under Assumptions~\ref{ass:mixing}--\ref{ass:kernel} and
\ref{ass:nuisance-smooth}--\ref{ass:local}, suppose Theorems~\ref{thm:bootstrap}
and \ref{thm:MB-local} hold for each fixed
$\lambda\in\Lambda$ and $B\in\mathcal B$, with stochastic equicontinuity
uniform over $\mathcal B$. Assume
\begin{equation}
\lim_{\varepsilon\downarrow0}
\max_{B\in\mathcal B\cup\{0\}}
\sup_{\lambda\in\Lambda}
P\!\left(
\left|T_\Pi(\lambda;B)-c_{1-\tau,\Pi}(\lambda)\right|
\le\varepsilon
\right)=0,
\label{eq:uniform-anti-concentration}
\end{equation}
and uniform $o_p(1)$ Monte Carlo errors in the bootstrap critical values and
rejection frequencies over $\lambda\in\Lambda$. If $W$ is continuous with
unique maximizer $\lambda^*$, the selector in
\eqref{eq:approximate-penalty-selector} satisfies
\[
\sup_{\lambda\in\Lambda,\,B\in\mathcal B}
\left|
\widehat{\mathcal R}_n(\lambda,B,\tau)
-\mathcal R(\lambda,B,\tau)
\right|\xrightarrow{p}0,
\qquad
\widehat\lambda\xrightarrow{p}\lambda^*.
\]
Moreover,
\[
\min_{B\in\mathcal B}\mathcal R(\lambda^*,B,\tau)
=\max_{\lambda\in\Lambda}
\min_{B\in\mathcal B}\mathcal R(\lambda,B,\tau)
\ge
\min_{B\in\mathcal B}\mathcal R(0,B,\tau).
\]
For every $B\in\mathcal B$, under the corresponding contiguous local sequence,
\[
P_{n,B}\!\left\{
T_{n,\Pi}(\widehat\lambda)>
c_{n,1-\tau,\Pi}^*(\widehat\lambda)
\right\}
\longrightarrow
\mathcal R(\lambda^*,B,\tau).
\]
If $\mathcal B=\{B^*\}$, this gives the pointwise conclusion
$\mathcal R(\lambda^*,B^*,\tau)\ge
\mathcal R(0,B^*,\tau)$. For larger $\mathcal B$, the guarantee is maximin and
need not hold separately for every direction.
\end{theorem}

\begin{remark}[Penalty grids and bootstrap simulation]
\label{rem:grid-no-loss}
Let $\Lambda_G$ be any fixed finite grid containing zero, and select the
smallest maximizer of $\widehat W_n$ over this grid. Then
\[
\max_{\lambda\in\Lambda_G}W(\lambda)\ge W(0).
\]
Nodewise consistency of bootstrap quantiles and rejection probabilities,
with vanishing simulation errors, is uniform on a finite grid.
A unique grid maximizer then gives selector consistency by
Theorem~\ref{thm:maximin-no-loss}. For grids $\Lambda_{G_n}$ with
$\sup_{\lambda\in\Lambda}\inf_{\ell\in\Lambda_{G_n}}|\lambda-\ell|\to0$,
uniform criterion convergence on the grids and continuity of $W$ give
consistency for its unique continuous-search maximizer.

On fixed finite $\Lambda_G$ and $\mathcal B$, the empirical bootstrap
distribution functions converge uniformly conditional on $\mathcal H_n$ as
the number of independent draws $R_n\to\infty$. A finite union bound
allows draws to be shared across penalties and directions.
Lemma~\ref{lem:gaussian-quantile-crossing} and
\eqref{eq:uniform-anti-concentration} then give vanishing quantile and
rejection-frequency errors. Fixed $R$ need not have this property;
continuous search requires the stated uniform Monte Carlo bound.
\end{remark}

For pre-estimation, write
\[
M_n^c(\alpha,\gamma)
:=
\frac1n\sum_{t=1}^n
w_t(\gamma)
\big[\mathbf 1\{Y_t\le m_{t\alpha}(\theta_{10},\theta_{20})\}-\alpha\big],
\]
where $\sqrt n M_n^c\rightsquigarrow\mathcal M^c$. With
$\mathcal G_{\theta_2}$ from Assumption~\ref{ass:pre-est-regularity}(ii),
the nuisance contribution is
\[
\mathcal M_{\mathrm{pre}}(\alpha,\gamma)
=A_2(\alpha,\gamma)^\top\mathcal G_{\theta_2}(\alpha).
\]

\begin{theorem}[Pre-estimation null limit]
\label{thm:null-pre}
Under $H_0$ and Assumptions~\ref{ass:mixing}--\ref{ass:nuisance-smooth}:
\begin{enumerate}
\item $\sqrt n\,\widehat M_n\rightsquigarrow
\widetilde{\mathcal M}:=\mathcal M^c+\mathcal M_{\mathrm{pre}}$,
a tight centered Gaussian process.
\item $\|\widehat s_n^2-\tilde\sigma^2\|_\infty\xrightarrow{p}0$, where,
by Assumption~\ref{ass:pre-est-regularity}(iii),
\[
\tilde\sigma^2(\alpha,\gamma)
=E[\Psi_t(\alpha,\gamma)^2]
=\operatorname{Var}\{\widetilde{\mathcal M}(\alpha,\gamma)\}.
\]
\item $\widehat T_{n,\Pi}(\lambda)\rightsquigarrow
\phi_\lambda(\widetilde{\mathcal M}/\tilde\sigma)$.
\end{enumerate}
\end{theorem}

\begin{theorem}[Bootstrap validity with pre-estimation]
\label{thm:boot-pre}
Under $H_0$ and Assumptions~\ref{ass:mixing}--\ref{ass:nuisance-smooth},
the bootstrap has the limits in Theorem~\ref{thm:null-pre}:
\begin{enumerate}
\item $\sqrt n\,\widehat M_{n*}\overset{*}{\rightsquigarrow}\widetilde{\mathcal M}$;
\item $\|\widehat s_{n*}^2-\tilde\sigma^2\|_\infty\xrightarrow{p^*}0$;
\item $\widehat T_{n*,\Pi}(\lambda)\overset{*}{\rightsquigarrow}
\phi_\lambda(\widetilde{\mathcal M}/\tilde\sigma)$.
\end{enumerate}
\end{theorem}

\begin{theorem}[Multiplier bootstrap validity under local alternatives with pre-estimation]
\label{thm:boot-local-pre}
Suppose Assumptions~\ref{ass:mixing}--\ref{ass:local} hold. Fix an admissible
bootstrap direction $B$. Under $H_0$ or any actual local law $P_{n,B_0}$ in
Assumption~\ref{ass:local}, evaluate the statistic with the null-restricted
nuisance fit of Assumptions~\ref{ass:pre-est-regularity} and~\ref{ass:local}.
With $\widetilde{\mathcal M}$ from Theorem~\ref{thm:null-pre} and
$\widetilde d_B$ from Lemma~\ref{lem:local-drift-expansions}:
\begin{enumerate}
\item $\sqrt n\,\widehat M_{n*,B}\overset{*}{\rightsquigarrow}
\widetilde{\mathcal M}+\widetilde d_B$;
\item $\|\widehat s_{n*,B}^2-\tilde\sigma^2\|_\infty\xrightarrow{p^*}0$;
\item $\widehat T_{n*,B,\Pi}(\lambda)\overset{*}{\rightsquigarrow}
\phi_\lambda\big((\widetilde{\mathcal M}+\widetilde d_B)/\tilde\sigma\big)$.
\end{enumerate}
\end{theorem}

\section*{S4. Proof of main theorems}

\begin{proof}[Proof of Theorem \ref{thm:null}]
By Lemma~\ref{lem:VC} and Assumption~\ref{ass:regularity}(i), the null
score class $\mathcal G_0$ is a permissible VC-subgraph class with an
$L_q$ envelope. Assumption~\ref{ass:mixing}(i) and
Lemma~\ref{lem:mixing-fclt} give
\[
\sqrt n M_n\rightsquigarrow\mathcal M
\quad\text{in }\ell^\infty(\mathcal T\times\Gamma).
\]
Predictability and Assumption~\ref{ass:mixing}(iii) make covariances at
nonzero lags vanish. The nested quantile indicators give
\[
\operatorname{Cov}\{\mathcal M(\alpha_1,\gamma_1),
\mathcal M(\alpha_2,\gamma_2)\}
=(\min\{\alpha_1,\alpha_2\}-\alpha_1\alpha_2)
E e^{(\gamma_1+\gamma_2)^\top I_{t-1}}.
\]
Lemma~\ref{lem:uwlln}, applied to the continuous class of exponential
weights with its integrable envelope, gives
$\sup_{\alpha,\gamma}|s_n^2(\alpha,\gamma)-\sigma^2(\alpha,\gamma)|=o_p(1)$.
The limit variance is uniformly positive on the compact index set.

The sample sections are measurable, and the proof of
Lemma~\ref{lem:conditional-gaussian-multiplier} gives a continuous Gaussian
version. Thus the processes belong to $\mathbb B_\Pi$.
Since $\sigma$ is continuous and $\inf\sigma>0$, studentization and the
bound from Lemma~\ref{lem:gaussian-quantile-crossing},
\[
|\phi_\lambda(f)-\phi_\lambda(g)|
\le(\|f\|_\infty+\|g\|_\infty)\|f-g\|_\infty,
\]
give $T_{n,\Pi}(\lambda)\rightsquigarrow\phi_\lambda(\mathcal M/\sigma)$.
\end{proof}

\begin{proof}[Proof of Theorem \ref{thm:bootstrap}]
In the notation of \citet{hansen1996},
Lemma~\ref{lem:conditional-gaussian-multiplier} gives
\[
\sqrt n M_{n*}\rightsquigarrow^*\mathcal M
\quad\text{in }\ell^\infty(\mathcal T\times\Gamma).
\]
The lemma gives conditional equicontinuity and uniform semimetric
convergence. The conditional covariance is the sample second moment of
the scores; its limit is the full long-run covariance since nonzero-lag
covariances vanish under Assumption~\ref{ass:mixing}(iii).

Write $P_na=n^{-1}\sum_t a(I_{t-1})$ and $Pa=E[a(I_{t-1})]$,
with the transformed instrument when applicable. For
$v_\gamma(I)=e^{2\gamma^\top I}$ and $\xi_t=\omega_t^2-1$, write
$P_n\xi v_\gamma=n^{-1}\sum_t\xi_t v_\gamma(I_{t-1})$. Then
\[
s_{n*}^2-\sigma^2
=\alpha(1-\alpha)\{(P_n-P)v_\gamma+P_n\xi v_\gamma\},
\qquad \sup_\gamma|(P_n-P)v_\gamma|=o_p(1)
\]
by Lemma~\ref{lem:uwlln}. Set $D_R=\{\|I\|\le R\}$,
$\bar\gamma=\sup_\Gamma\|\gamma\|$ and $F_2=e^{2\bar\gamma\|I\|}$.
On $D_R$, $v_\gamma$ has Lipschitz constant $C_R=2Re^{2\bar\gamma R}$.
For a fixed finite $\delta$-net $\Gamma_\delta$ of $\Gamma$,
\[
\begin{aligned}
\operatorname{Var}^*(P_n\xi v_\gamma\mathbf1_{D_R})
   &\le 2e^{4\bar\gamma R}/n,\\
\sup_\Gamma|P_n\xi v_\gamma\mathbf1_{D_R}|
   &\le \max_{\Gamma_\delta}|P_n\xi v_\gamma\mathbf1_{D_R}|
        +C_R\delta P_n|\xi|,\qquad E^*P_n|\xi|=E|\xi_1|.
\end{aligned}
\]
Conditional Chebyshev and Markov inequalities, first as $n\to\infty$ then as
$\delta\downarrow0$, give $o_p^*(1)$ on $D_R$. For the tail,
\[
E^*\sup_\Gamma|P_n\xi v_\gamma\mathbf1_{D_R^c}|
\le E|\xi_1|P_n(F_2\mathbf1_{D_R^c})
\xrightarrow{\mathrm{a.s.}}E|\xi_1|P(F_2\mathbf1_{D_R^c})
\xrightarrow[R\to\infty]{}0
\]
by ergodicity and integrability of $F_2$ (Assumption~\ref{ass:regularity}(i)).
Thus $\|s_{n*}^2-\sigma^2\|_\infty=o_p^*(1)$; $\inf\sigma>0$ and
continuity of $\phi_\lambda$ give the conditional statistic limit.
\end{proof}

\begin{proof}[Proof of Corollary \ref{cor:fixed-direction-grid}]
Restriction to $\Gamma_G$ is continuous and
\[
|\phi_{\lambda,G}(f)-\phi_{\lambda,G}(g)|
\le(\|f\|_\infty+\|g\|_\infty)\|f-g\|_\infty.
\]
The ordinary and conditional mapping theorems therefore apply to
Theorems~\ref{thm:null} and~\ref{thm:bootstrap}, or
Theorems~\ref{thm:null-pre} and~\ref{thm:boot-pre} for the restricted fit.
Lemma~\ref{lem:gaussian-quantile-crossing} gives strict crossing on the
grid, so the size argument of Corollary~\ref{cor:fixed-lambda-size} applies.
\end{proof}

\begin{proof}[Proof of Corollary \ref{cor:fixed-lambda-size}]
Write $c=q_{1-\tau}\{T_\Pi(\lambda)\}$. Since $\inf\sigma>0$ and $\sigma$
is deterministic, $\mathcal M/\sigma$ is tight centered Gaussian.
Lemma~\ref{lem:gaussian-quantile-crossing} gives strict crossing at $c$;
Theorem~\ref{thm:bootstrap} at continuity points on either side gives
$c_{n,1-\tau,\Pi}^*(\lambda)\to_p c$. Hence, by continuity at $c$,
\[
P\{T_{n,\Pi}(\lambda)>c_{n,1-\tau,\Pi}^*(\lambda)\}
\to P\{T_\Pi(\lambda)>c\}=\tau.
\]
The same argument applies to the tight centered Gaussian process
$\widetilde{\mathcal M}/\tilde\sigma$.
\end{proof}

\begin{proof}[Proof of Corollary \ref{cor:adaptive-lambda}]
For $C_\Gamma=\sup_{\gamma\in\Gamma}\|\gamma\|_1<\infty$, the common penalty gives
\[
|T_{n,\Pi}(\lambda)-T_{n,\Pi}(\lambda')|\le C_\Gamma|\lambda-\lambda'|.
\]
The same bound holds pathwise for the bootstrap and limit statistics, their
plug-in versions, and hence their ideal quantiles. Since
$\widehat\lambda\to_p\lambda_0$ and any simulation error is uniformly $o_p(1)$,
\[
\begin{aligned}
|T_{n,\Pi}(\widehat\lambda)-T_{n,\Pi}(\lambda_0)|
&\le C_\Gamma|\widehat\lambda-\lambda_0|=o_p(1),\\
|c_{n,1-\tau,\Pi}^*(\widehat\lambda)-c_{n,1-\tau,\Pi}^*(\lambda_0)|
&\le C_\Gamma|\widehat\lambda-\lambda_0|+o_p(1)=o_p(1).
\end{aligned}
\]
At $\lambda_0$, Lemma~\ref{lem:gaussian-quantile-crossing} and
Theorem~\ref{thm:bootstrap} give quantile consistency as in
Corollary~\ref{cor:fixed-lambda-size}. Continuity at the limiting critical
value therefore gives the same rejection limit $\tau$ at $\widehat\lambda$.
For the plug-in statistic use Theorem~\ref{thm:boot-pre}.
\end{proof}

\begin{proof}[Proof of Theorem \ref{thm:consistency}]
For the fixed curve $\bar\theta$, the scores
$e^{\gamma^\top i}[\mathbf1\{y\le m(i,\bar\theta(\alpha))\}-\alpha]$
form a permissible VC-subgraph class with an $L_q$ envelope
(Lemma~\ref{lem:VC}). By Assumption~\ref{ass:mixing}(i) and
Lemmas~\ref{lem:mixing-fclt} and~\ref{lem:uwlln},
\[
\|M_n-\eta(\cdot,\cdot;\bar\theta)\|_\infty=O_p(n^{-1/2}),
\qquad \|s_n^2-\sigma^2\|_\infty=o_p(1).
\]
Since $\inf\sigma>0$ and $\sup_\Gamma \lambda\|\gamma\|_1/n\to0$,
\[
\begin{aligned}
\frac{T_{n,\Pi}(\lambda)}{n}
&=\sup_{\gamma\in\Gamma}\left\{
\int_{\mathcal T}(M_n/s_n)^2\,d\Pi-\frac{\lambda}{n}\|\gamma\|_1\right\}\\
&\xrightarrow{p}
\sup_{\gamma\in\Gamma}\int_{\mathcal T}
(\eta(\alpha,\gamma;\bar\theta)/\sigma(\alpha,\gamma))^2\,d\Pi(\alpha)>0.
\end{aligned}
\]
Thus $T_{n,\Pi}(\lambda)\to_p+\infty$. If $\Pi$ has full support and
$\eta(\cdot,\gamma;\bar\theta)/\sigma(\cdot,\gamma)$ is continuous,
a nonzero moment at any $\alpha'$ persists on a set of positive $\Pi$-measure.
\end{proof}

\begin{proof}[Proof of Theorem \ref{thm:consistency_density}]
Apply Lemma~\ref{lem:unif_rate} with $\vartheta^r=\bar\theta$.
Assumption~\ref{ass:regularity}(iv)(4) supplies the smooth reference
curve, the buffer between $\mathcal T$ and the boundary of the pilot support,
and the identity $s_t^r(\alpha)=f_{t\alpha}$.

With $J\asymp n$ and $a_n\asymp b_n\asymp(\log n/n)^{1/5}$,
\[
\frac{Ja_n}{\log J}\asymp
\frac{Ja_n}{\log n}\asymp
\left(\frac{n}{\log n}\right)^{4/5}\longrightarrow\infty.
\]
Choose
$r_J=C_0\{\log J+\log n+\log(1/a_n)\}\asymp\log n$.
Then $r_J/(Ja_n)\to0$, so condition~\eqref{eq:D} holds, and
\[
b_n^2\asymp
\sqrt{\frac{r_J}{Ja_n}}
\asymp\left(\frac{\log n}{n}\right)^{2/5}.
\]
The lemma gives the rate and uniform consistency. This bandwidth also gives
\[
n^{-1/2}a_{J_n}^{-2}
\asymp n^{-1/10}(\log n)^{-2/5}\longrightarrow0.
\]
Thus Assumption~\ref{ass:kernel}(d), used by the plug-in results but not by
the oracle density rate, also holds.
\end{proof}

\begin{proof}[Proof of Theorem \ref{thm:MB-local}]
Distinguish the data law $P_{n,B_0}$ ($B_0=0$ under the null) from the
bootstrap direction $B$. With
$Z_{nt}=e^{\gamma^\top I_{t-1}}
[\mathbf1\{Y_t\le m(I_{t-1},\theta_0(\alpha))\}-\alpha]$,
Assumption~\ref{ass:local}(iii) and Lemma~\ref{lem:local-drift-expansions} give
\[
\sqrt n M_{n,0}(\alpha,\gamma;B_0)
=\frac1{\sqrt n}\sum_t(Z_{nt}-E_{n,B_0}Z_{nt})
 +d_{B_0}(\alpha,\gamma)+o(1)
\rightsquigarrow\mathcal M+d_{B_0},
\]
with uniform remainder. In contrast, \eqref{eq:shifted_M} gives
\[
\sqrt n M_{n*}(\alpha,\gamma;B)
=Z_{n*}(\alpha,\gamma)+\widehat d_{n,B}(\alpha,\gamma),
\qquad Z_{n*}=n^{-1/2}\sum_t\omega_t Z_{nt}.
\]
$Z_{n*}$ is conditionally centered. By Theorem~\ref{thm:bootstrap} its
conditional law converges to that of $\mathcal M$ under the null.
Contiguity transfers this convergence to $P_{n,B_0}$ because the
conditional bounded-Lipschitz distance is a function of the sample.

Evaluate $m_{t\alpha}$ and $\nabla_\theta m_{t\alpha}$ at
$(I_{t-1},\theta_0(\alpha))$. Set
\[
f_{t\alpha}=f_{Y_t\mid I_{t-1}}(m_{t\alpha}\mid I_{t-1}),\qquad
b_t=e^{\gamma^\top I_{t-1}}\,\nabla_\theta m_{t\alpha}^\top B(\alpha).
\]
Then $\widehat d_{n,B}=n^{-1}\sum_t b_t\widehat f_{t\alpha}$, with
$\widehat f_{t\alpha}=\widehat f_h(I_{t-1},\bar\theta;\alpha)$.
Putting $\widetilde D_n=n^{-1}\sum_t b_t f_{t\alpha}$ gives
\[
\|\widehat d_{n,B}-\widetilde D_n\|_\infty
\le\Delta_{f,n}\,n^{-1}\sum_t\sup_{\alpha,\gamma}|b_t|=o_p(1).
\]
Here Theorem~\ref{thm:consistency_density} controls the density error on
the bandwidth event, and Assumptions~\ref{ass:regularity}(iv) and
\ref{ass:nuisance-smooth}(ii)--(iii) give an integrable envelope.
For the oracle average, let
\[
a(i;\vartheta,\gamma)
=e^{\gamma^\top i}f_{Y_t\mid I_{t-1}}(m(i,\vartheta)\mid i)
\nabla_\vartheta m(i,\vartheta).
\]
Continuity on compact $\Theta\times\Gamma$ and the integrable envelope
$C\mathcal E_t^2$ give, by Lemma~\ref{lem:uwlln},
\[
\|\widetilde D_n-d_B\|_\infty
\le \sup_\alpha\|B(\alpha)\|\,
\sup_{\vartheta,\gamma}\|(P_n-P)a(\cdot;\vartheta,\gamma)\|=o_p(1)
\]
under the null. Contiguity transfers both bounds to $P_{n,B_0}$.
Thus $\sqrt n M_{n*}(\cdot,\cdot;B)\rightsquigarrow^*\mathcal M+d_B$.
By \eqref{eq:sn_star} and \eqref{eq:shifted_s}, $s_{n*,B}^2=s_{n*}^2$;
Theorem~\ref{thm:bootstrap}, contiguity, and $\inf\sigma>0$ give
$\|s_{n*,B}^2-\sigma^2\|_\infty=o_p^*(1)$. Applying $\phi_\lambda$ to the
studentized process completes the proof.
\end{proof}

\begin{proof}[Proof of Theorem \ref{thm:power-enhance}]
On the stated realization, set
$\Delta_\Pi(\lambda)=T_\Pi(\lambda;B_0)-T_\Pi(\lambda;0)$.
The objective is continuous in $\gamma$ on compact $\Gamma$, affine in
$\lambda$, and has a unique maximizer at zero penalty for $B=0,B_0$.
Danskin's theorem gives
\[
\partial_{\lambda+}T_\Pi(0;B)=-\|\tilde\gamma_\Pi(B)\|_1,\qquad
\Delta'_{\Pi,+}(0)
=\|\tilde\gamma_\Pi(0)\|_1-\|\tilde\gamma_\Pi(B_0)\|_1=:c>0.
\]
Consequently $\Delta_\Pi(\lambda)-\Delta_\Pi(0)=c\lambda+o(\lambda)>0$
for $0<\lambda<\delta(Q,B_0)$ with some $\delta(Q,B_0)>0$.
This is a pathwise comparison of separation, not of rejection probabilities.
\end{proof}

\begin{proof}[Proof of Theorem \ref{thm:maximin-no-loss}]
Write
\[
C_\Gamma:=\sup_{\gamma\in\Gamma}\|\gamma\|_1<\infty.
\]
For any bounded real-valued objective $J(\gamma)$ on $\Gamma$ and any
$\lambda,\lambda'\in\Lambda$,
\begin{equation}
\left|
\sup_{\gamma\in\Gamma}\{J(\gamma)-\lambda\|\gamma\|_1\}
-
\sup_{\gamma\in\Gamma}\{J(\gamma)-\lambda'\|\gamma\|_1\}
\right|
\le C_\Gamma|\lambda-\lambda'|.
\label{eq:lambda-lipschitz}
\end{equation}
The sample statistic, every null and shifted bootstrap statistic, and
$T_\Pi(\lambda;B)$ therefore share this Lipschitz bound in $\lambda$.

Let $D_n(\lambda,B)$ be the conditional bounded-Lipschitz distance
between the bootstrap and limit laws at $(\lambda,B)$, using $B=0$ for
the null. On a fixed finite $\delta$-net $\Lambda_\delta$, pointwise
convergence and finiteness of $\mathcal B$ give, by \eqref{eq:lambda-lipschitz},
\[
\max_{B\in\mathcal B\cup\{0\}}\sup_{\lambda\in\Lambda}D_n(\lambda,B)
\le \max_{\substack{B\in\mathcal B\cup\{0\}\\\ell\in\Lambda_\delta}}
D_n(\ell,B)+2C_\Gamma\delta=o_p(1)+2C_\Gamma\delta.
\]
Let $\delta\downarrow0$.

Quantiles inherit the same pathwise bound:
\[
|c_{1-\tau,\Pi}(\lambda)-c_{1-\tau,\Pi}(\lambda')|
\le C_\Gamma|\lambda-\lambda'|.
\]
The ideal conditional bootstrap quantiles
$c^{*,\mathrm{id}}_{n,1-\tau,\Pi}(\lambda)$ obey the same bound.
The $B=0$ case of
\eqref{eq:uniform-anti-concentration} gives continuity at every null quantile.
Lemma~\ref{lem:gaussian-quantile-crossing} supplies strict crossing.
Conditional weak convergence gives quantile consistency on each fixed
finite $\delta$-net $\Lambda_\delta$ of $\Lambda$, and
\[
\begin{aligned}
&\sup_{\lambda\in\Lambda}
|c^{*,\mathrm{id}}_{n,1-\tau,\Pi}(\lambda)-c_{1-\tau,\Pi}(\lambda)|\\
&\quad\le
\max_{\lambda_j\in\Lambda_\delta}
|c^{*,\mathrm{id}}_{n,1-\tau,\Pi}(\lambda_j)-c_{1-\tau,\Pi}(\lambda_j)|
+2C_\Gamma\delta.
\end{aligned}
\]
Letting first $n\to\infty$ and then $\delta\downarrow0$, and adding the
assumed uniform Monte Carlo quantile error, gives
\begin{equation}
\sup_{\lambda\in\Lambda}
\left|c_{n,1-\tau,\Pi}^*(\lambda)
-c_{1-\tau,\Pi}(\lambda)\right|=o_p(1).
\label{eq:uniform-critical-value}
\end{equation}
Let $\mathcal R_n^*(\lambda,B,\tau)$ be the ideal conditional rejection
probability of the shifted bootstrap using the null bootstrap critical value.
For any $\eta>0$, uniform convergence of the shifted statistic and
\eqref{eq:uniform-critical-value} bound its error by a term vanishing with
$n$, plus
\[
\max_{B\in\mathcal B}
\sup_{\lambda\in\Lambda}
P\!\left(
|T_\Pi(\lambda;B)-c_{1-\tau,\Pi}(\lambda)|\le\eta
\right).
\]
This vanishes as $\eta\downarrow0$ by
\eqref{eq:uniform-anti-concentration}. Add the assumed uniform simulation
error in the empirical rejection frequencies to obtain
\begin{equation}
e_n:=\sup_{\lambda\in\Lambda,\,B\in\mathcal B}
|\widehat{\mathcal R}_n(\lambda,B,\tau)
-\mathcal R(\lambda,B,\tau)|=o_p(1).
\label{eq:uniform-power-error}
\end{equation}
Taking minima gives
$\sup_\lambda|\widehat W_n(\lambda)-W(\lambda)|\le e_n$.

The definition of $\widehat\lambda$ and $0\in\Lambda$ give
\begin{equation}
W(\widehat\lambda)
\ge \widehat W_n(\widehat\lambda)-e_n
\ge \widehat W_n(0)-e_n-\eta_n
\ge W(0)-2e_n-\eta_n.
\label{eq:no-loss-estimated}
\end{equation}
Likewise $0\le W(\lambda^*)-W(\widehat\lambda)\le 2e_n+\eta_n=o_p(1)$.
For each $\varepsilon>0$ with a nonempty comparison set, compactness,
continuity, and uniqueness imply
\[
W(\lambda^*)-\sup_{\substack{\lambda\in\Lambda\\|\lambda-\lambda^*|\ge\varepsilon}}
W(\lambda)>0.
\]
Hence $\widehat\lambda\to_p\lambda^*$ and
$W(\lambda^*)=\max_\Lambda W\ge W(0)$. For $\mathcal B=\{B^*\}$ the
minimum is simply the rejection probability in direction $B^*$.

Fix $B\in\mathcal B$. By contiguity, $\widehat\lambda\to_{P_{n,B}}\lambda^*$
and \eqref{eq:uniform-critical-value} still holds. By \eqref{eq:lambda-lipschitz},
\[
|T_{n,\Pi}(\widehat\lambda)-T_{n,\Pi}(\lambda^*)|
\le C_\Gamma|\widehat\lambda-\lambda^*|=o_{P_{n,B}}(1).
\]
The uniform critical-value bound and continuity of $c_{1-\tau,\Pi}$
likewise give $c_{n,1-\tau,\Pi}^*(\widehat\lambda)
=c_{1-\tau,\Pi}(\lambda^*)+o_{P_{n,B}}(1)$.
The fixed-penalty local limit under Assumption~\ref{ass:local} now yields
\[
T_{n,\Pi}(\widehat\lambda)
-c_{n,1-\tau,\Pi}^*(\widehat\lambda)
\rightsquigarrow
T_\Pi(\lambda^*;B)-c_{1-\tau,\Pi}(\lambda^*).
\]
The anti-concentration condition makes the limit atomless at zero, so
rejection probabilities converge to $\mathcal R(\lambda^*,B,\tau)$.
\end{proof}
\begin{proof}[Proof of Theorem \ref{thm:null-pre}]
Write $P_na=n^{-1}\sum_t a(Z_t)$, $e_t(\gamma)=e^{\gamma^\top I_{t-1}}$,
$\bar e_n=P_ne_t$ and $\mu_e=Ee_t$. Evaluate $r^0_{t\alpha},g^0_{t\alpha}$
at $(\bar\theta_1,\theta_{20})$ and $\widehat r_{t\alpha},\widehat g_{t\alpha}$
at $(\bar\theta_1,\widehat\theta_2)$. Denote the corresponding kernel density
estimates by $\widehat f^0_{t\alpha}$ and $\widehat f_{t\alpha}$, respectively.
Sample centering gives
\begin{equation}
\sqrt n P_n\{(w_{t,n}-w_t)\widehat r_{t\alpha}\}
=-\sqrt n(\bar e_n-\mu_e)P_n\widehat r_{t\alpha}.
\label{eq:sample-centering-product}
\end{equation}
The exponential-weight process gives $\sup_\gamma\sqrt n|\bar e_n-\mu_e|=O_p(1)$.
The influence expansion in Assumption~\ref{ass:pre-est-regularity}, bounded
density and localized equicontinuity give
$\sup_\alpha|P_n\widehat r_{t\alpha}|=O_p(n^{-1/2})$.
Thus \eqref{eq:sample-centering-product} is uniformly $o_p(1)$.

For $s_t(\theta_2)=w_t(\gamma)
[\mathbf 1\{Y_t\le m_{t\alpha}(\bar\theta_1,\theta_2)\}-\alpha]$,
put $\mu(\theta_2)=Es_t(\theta_2)$, suppressing $(\alpha,\gamma)$.
Under Assumptions~\ref{ass:regularity}(iv) and \ref{ass:nuisance-smooth},
the Taylor remainder for $m$ is bounded by $\mathcal L_t\|d\|^2/2$.
Bounded conditional density and density derivative then give, uniformly in
$(\alpha,\gamma)$ for small $d$,
\[
\begin{aligned}
|\mu(\theta_{20}+d)-\mu(\theta_{20})-A_2^\top d|
   &\le C E[\mathcal E_t^3]\|d\|^2,\\
E|s_t(\theta_{20}+d)-s_t(\theta_{20})|^2 &\le C\|d\|.
\end{aligned}
\]
Put $D_n=\|\widehat\theta_2-\theta_{20}\|_{\infty,\mathcal A}=O_p(n^{-1/2})$.
Since $\sqrt nD_n^2=O_p(n^{-1/2})$, localized equicontinuity and joint
convergence in Assumption~\ref{ass:pre-est-regularity}(ii) yield
\[
\sqrt n\,\widehat M_n
=\sqrt n M_n^c+A_2^\top n^{-1/2}\sum_t l_{t,\alpha}+o_p(1)
\ \rightsquigarrow\ \mathcal M^c+A_2^\top\mathcal G_{\theta_2}
=\widetilde{\mathcal M}
\]
in $\ell^\infty(\mathcal T\times\Gamma)$. By
Remark~\ref{rem:orthogonality}, the centered Gaussian limit has kernel
$E[\Psi_t(\alpha,\gamma)\Psi_t(\alpha',\gamma')]$.

For density estimation, write $J=J_n$ and
$r_J=C_0\{\log J+\log n+\log(1/a_J)\}$. Since $Ja_J\to\infty$,
$\log(1/a_J)\le\log J$ eventually. Lemma~\ref{lem:unif_rate} and
Assumption~\ref{ass:kernel}(b)--(c) give
\[
\Delta^0_{f,n}:=\sup_\alpha\max_{t\le n}|\widehat f^0_{t\alpha}-f_{t\alpha}|
=O_p\!\left(b_J^2+\sqrt{r_J/(Ja_J)}\right)=o_p(1).
\]
On $\{a_J\le h_J\le b_J\}$, the mean-value bound is
$\sup_\alpha|\widehat f_{t\alpha}-\widehat f^0_{t\alpha}|
\le C a_J^{-2}D_n G_t$, where $G_t\le C\mathcal E_t$ bounds the gradient.
For an envelope $H_t\le C\mathcal E_t^2$ of the factors multiplying the
density in $\widehat L_\alpha$ and $\widehat A_{2,n}$,
\[
\sup_\alpha P_n\{|\widehat f_{t\alpha}-f_{t\alpha}|H_t\}
\le C a_J^{-2}D_n P_n(G_tH_t)+\Delta^0_{f,n}P_nH_t=o_p(1):
\]
the bandwidth event has probability tending to one,
$n^{-1/2}a_J^{-2}\to0$ by Assumption~\ref{ass:kernel}(d), and
$P_n(G_tH_t)+P_nH_t=O_p(1)$ by Assumption~\ref{ass:nuisance-smooth}(iii).
Bounded $\partial_y f$ and $\partial_u\bar q_i$ imply
$|f_{t\alpha}-f_{t\alpha'}|\le C|\alpha-\alpha'|$ uniformly in $t$.
Thus $f_{t\alpha}g_t(\vartheta)g_t(\vartheta)^\top$ and
$f_{t\alpha}w_t(\gamma)g_t(\vartheta)$ are continuous on
$\mathcal T\times\Theta\times\Gamma$, with integrable envelope $C\mathcal E_t^2$.
Lemma~\ref{lem:uwlln}, gradient continuity and uniform nonsingularity give
\[
\sup_\alpha\big(\|\widehat L_\alpha-L_\alpha\|
+\|\widehat L_\alpha^{-1}-L_\alpha^{-1}\|\big)
+\sup_{\alpha,\gamma}\|\widehat A_{2,n}-A_2\|=o_p(1).
\]

For score replacement, put $q_{t\alpha}=\bar q_{I_{t-1}}(\alpha)$ and
choose $\delta_n\downarrow0$ with $\sqrt n\delta_n\to\infty$.
On $\{D_n\le\delta_n\}$, which has probability tending to one,
\[
|\widehat r_{t\alpha}-r^0_{t\alpha}|^2
\le\mathbf 1\{|Y_t-q_{t\alpha}|\le G_t\delta_n\}.
\]
Take $H_t^0=C\mathcal E_t^2$ dominating both $w_t(\gamma)^2$ and
$\|g^0_{t\alpha}\|^2$. For a finite $\eta$-net $\{\alpha_k\}$ of $\mathcal T$
and fixed $\delta>0$, the Lipschitz bound on $\bar q_i$, ergodicity and
bounded conditional density give, eventually for $\delta_n\le\delta$,
\[
\begin{aligned}
&\sup_\alpha P_n[H_t^0\mathbf 1\{|Y_t-q_{t\alpha}|\le G_t\delta_n\}]\\
&\quad\le\max_k P_n[H_t^0\mathbf 1\{|Y_t-q_{t\alpha_k}|
                    \le G_t\delta+C_q\eta\}]
\le C\{\delta E(H_t^0G_t)+\eta EH_t^0\}+o_p(1).
\end{aligned}
\]
Letting $\delta,\eta\downarrow0$ gives
$\sup_\alpha P_n[H_t^0|\widehat r_{t\alpha}-r^0_{t\alpha}|^2]=o_p(1)$.
By Assumption~\ref{ass:nuisance-smooth}(i),
$\sup_\alpha P_n\|\widehat g_{t\alpha}-g^0_{t\alpha}\|^2
\le D_n^2P_n\mathcal L_t^2=O_p(n^{-1})$. Suppressing indices, expand
\[
\begin{aligned}
\widehat l-l
&=-(\widehat L^{-1}-L^{-1})\widehat g\widehat r
  -L^{-1}(\widehat g-g^0)\widehat r-L^{-1}g^0(\widehat r-r^0),\\
\widehat\Psi-\Psi
&=(w_{t,n}-w_t)\widehat r+w_t(\widehat r-r^0)
  +(\widehat l-l)^\top\widehat A_{2,n}+l^\top(\widehat A_{2,n}-A_2).
\end{aligned}
\]
The preceding bounds and $|\widehat r|\le1$ imply
$\sup_\alpha P_n\|\widehat l-l\|^2=o_p(1)$ and
$\sup_{\alpha,\gamma}P_n(\widehat\Psi-\Psi)^2=o_p(1)$.

For the oracle covariance, embed $\Psi$ in the class
\[
f_\eta(y,i)=(\mathbf 1\{y\le m(i,\theta)\}-\alpha)
\{e^{\gamma^\top i}-c-b^\top\nabla_{\theta_2}m(i,\theta)\},
\qquad \eta=(\theta,\alpha,\gamma,b,c).
\]
The moment bounds and uniform nonsingularity place
$b=L_\alpha^{-\top}A_2(\alpha,\gamma)$ and $c=Ee_t(\gamma)$ in compact sets.
On the resulting compact index set, let $\nu_\delta(i),\omega_\delta(i)$ be
the continuity moduli of $m$ and the second factor, respectively. Both vanish
as $\delta\downarrow0$; $F=C\mathcal E$ bounds the second factor,
$\omega_\delta\le2F$, and $EF^2<\infty$. For a finite $\delta$-net,
$\|\eta-\eta_k\|\le\delta$ gives
\[
|f_\eta-f_{\eta_k}|\le e_{k,\delta}
:=\omega_\delta+\delta F+
F\mathbf 1\{|Y-m(I,\theta_k)|\le\nu_\delta(I)\}.
\]
Conditioning on $I$ and applying dominated convergence, with
$C_f=\sup_{y,i}f(y\mid i)$, gives
\[
\sup_k E(F e_{k,\delta})
\le E(F\omega_\delta)+\delta EF^2+
E[F^2\min\{1,2C_f\nu_\delta(I)\}]\longrightarrow0.
\]
The finite product brackets
$f_{\eta_k}f_{\eta_l}\pm F(e_{k,\delta}+e_{l,\delta})$ thus have integrable
endpoints and vanishing expected widths. Apply the ergodic theorem at the
endpoints, then let $\delta\downarrow0$:
\begin{equation}
\sup_{\alpha,\gamma,\alpha',\gamma'}
\left|(P_n-P)
\{\Psi(\alpha,\gamma)\Psi(\alpha',\gamma')\}\right|=o_p(1).
\label{eq:pre-oracle-product-ulln}
\end{equation}
On the diagonal, Cauchy--Schwarz and the score replacement give
\[
\sup_{\alpha,\gamma}|P_n\widehat\Psi_t^2-E\Psi_t^2|=o_p(1).
\]
Uniform positivity of $\tilde\sigma^2$ and the continuous map of
Theorem~\ref{thm:null} give the statistic limit.
\end{proof}

\begin{lemma}[Uniform feasible--oracle multiplier replacement]
\label{lem:pre-multiplier-replacement}
Under Assumptions~\ref{ass:mixing}--\ref{ass:nuisance-smooth} and $H_0$ on
$\mathcal T$, let $\Delta_t(\alpha,\gamma)=\widehat\Psi_t(\alpha,\gamma)-
\Psi_t(\alpha,\gamma)$ for the corrected scores of Theorem~\ref{thm:boot-pre}.
Then
\[
r_n^2:=\sup_{\alpha,\gamma}n^{-1}\sum_t\Delta_t(\alpha,\gamma)^2=o_p(1).
\]
For i.i.d.\ standard Gaussian multipliers independent of the data,
\[
\sup_{\alpha,\gamma}\left|n^{-1/2}\sum_t\omega_t
\Delta_t(\alpha,\gamma)\right|=o_p^*(1).
\]
\end{lemma}

\begin{proof}[Proof of Lemma \ref{lem:pre-multiplier-replacement}]
The first bound follows from the proof of Theorem~\ref{thm:null-pre}.
With probability tending to one, $\widehat L^{-1}$ and $\widehat A_{2,n}$
lie in fixed bounded sets. The localized classes of
Assumption~\ref{ass:pre-est-regularity}(ii), closed under finite sums and
products with bounded coefficients, then give the difference class
$\widehat{\mathcal D}_n=\{\Delta(\alpha,\gamma)\}\cup\{0\}$ an envelope
$\widehat F_n$ satisfying $\|\widehat F_n\|_{P_n,2}=O_p(1)$ and
\[
N(\varepsilon\|\widehat F_n\|_{Q,2},\widehat{\mathcal D}_n,L_2(Q))
\le(A/\varepsilon)^v,\qquad 0<\varepsilon\le1
\]
for every finitely discrete $Q$. Given $\mathcal H_n$, the multiplier
process is centered Gaussian, vanishes at $0$, and has $L_2(P_n)$ diameter
at most $2r_n$. The entropy inequality of
\citet[Theorem~2.2.4 and Corollary~2.2.5]{vdvw1996} with $Q=P_n$ gives
\[
E\!\left[\sup_{\alpha,\gamma}
\left|n^{-1/2}\sum_t\omega_t\Delta_t(\alpha,\gamma)\right|
\,\middle|\,\mathcal H_n\right]
\le C\int_0^{2r_n}
\sqrt{\log\!\left(1+\frac{A\|\widehat F_n\|_{P_n,2}}{\varepsilon}\right)}
\,d\varepsilon=o_p(1),
\]
since $r\sqrt{\log(1+C/r)}\to0$ as $r\downarrow0$.
Conditional Markov inequality and averaging the bounded tail probabilities
over pilot randomness yield the second bound.
\end{proof}

\begin{proof}[Proof of Theorem \ref{thm:boot-pre}]
By Lemma~\ref{lem:pre-multiplier-replacement} it suffices for the process
limit to use the oracle $Z_n^*=n^{-1/2}\sum_t\omega_t\Psi_t$.
Given the data, $Z_n^*$ is centered Gaussian. For any two oracle scores
$\psi,\psi'$, \eqref{eq:pre-oracle-product-ulln} gives
\[
\sup_{\psi,\psi'}|(P_n-P)(\psi\psi')|=o_p(1),\qquad
\sup_{\psi,\psi'}|(P_n-P)(\psi-\psi')^2|=o_p(1).
\]
The first bound gives finite-dimensional convergence, since
$P(\psi\psi')$ is the long-run covariance by Remark~\ref{rem:orthogonality}.
The second bounds the empirical distance on population $L_2(P)$-balls of
radius $\delta$ by $\delta+o_p(1)$. The permissible class in
Assumption~\ref{ass:pre-est-regularity}(ii) has polynomial covering numbers
and an $L_q$ envelope as in Assumption~\ref{ass:mixing}. The same conditional
entropy integral as in the preceding lemma, now with upper limit
$C\delta+o_p(1)$ and $O_p(1)$ empirical envelope norm, tends to zero as
$n\to\infty$ and then $\delta\downarrow0$. This gives tightness and
$\sqrt n\widehat M_{n*}\rightsquigarrow^*\widetilde{\mathcal M}$.

For the studentizer put $\zeta_t=\omega_t^2-1$. Then
\[
\widehat s_{n*}^2=P_n\widehat\Psi_t^2+P_n(\zeta_t\widehat\Psi_t^2),
\qquad
\sup_{\alpha,\gamma}|P_n\widehat\Psi_t^2-\tilde\sigma^2|=o_p(1)
\]
by Theorem~\ref{thm:null-pre}. Restrict $\widehat L^{-1}$, $\widehat A_{2,n}$
and the sample centering constants to a ball of radius $C$; all are uniformly
$O_p(1)$. The fitted scores then belong to the class $f_\eta$ of
\eqref{eq:pre-oracle-product-ulln}, with envelope $F\le C'\mathcal E$.
Given $\varepsilon>0$, its finite net gives
\[
|f_\eta^2-f_{\eta_k}^2|\le d_k:=2F e_{k,\delta},\qquad
\max_k Pd_k\le\varepsilon,\qquad d_k\le C''\mathcal E^2.
\]
Write $\operatorname{Var}_n^\dagger$ for variance given $\mathcal H_n$.
Since $E\zeta_t=0$, $E\zeta_t^2=2$ and $E\mathcal E^4<\infty$
by Assumption~\ref{ass:nuisance-smooth}(iii),
\[
\max_k\Big\{\operatorname{Var}_n^\dagger[P_n(\zeta_t f_{\eta_k}^2)]
 +\operatorname{Var}_n^\dagger[P_n(|\zeta_t|d_k)]\Big\}
 \le C' P_n\mathcal E^4/n=O_p(n^{-1}).
\]
Conditional Chebyshev inequality at the finitely many centers gives
\[
\begin{aligned}
\sup_\eta|P_n(\zeta_t f_\eta^2)|
&\le\max_k|P_n(\zeta_t f_{\eta_k}^2)|+\max_k P_n(|\zeta_t|d_k)\\
&\le E|\zeta_t|\max_k P_nd_k+o_p^*(1).
\end{aligned}
\]
Ergodicity gives $\max_k P_nd_k\to\max_k Pd_k\le\varepsilon$.
Let $n\to\infty$, then $\varepsilon\downarrow0$ and $C\uparrow\infty$.
Averaging the conditional tail probabilities over the pilot draws yields
$\sup_{\alpha,\gamma}|\widehat s_{n*}^2-\tilde\sigma^2|=o_p^*(1)$
given the data. Uniform positivity and the continuous map in
Theorem~\ref{thm:null} give the statistic limit.
\end{proof}

\begin{proof}[Proof of Theorem \ref{thm:boot-local-pre}]
Fix the actual direction $B_0$ ($B_0=0$ under the null) and candidate
direction $B$. Separate the centered multiplier process from the drift:
\[
\sqrt n\,\widehat M_{n*,B}=Z_n^*+\widehat d_{n,B},\qquad
Z_n^*=n^{-1/2}\sum_t\omega_t\widehat\Psi_{t,n}.
\]
The influence expansion in Assumption~\ref{ass:pre-est-regularity} gives
$\|\widehat\theta_2-\theta_{20}\|_{\infty,\mathcal A}=O_p(n^{-1/2})$.
Contiguity transfers this rate, the conditional score replacement of
Lemma~\ref{lem:pre-multiplier-replacement} and the oracle limit in
Theorem~\ref{thm:boot-pre} to each $P_{n,B_0}$. Hence
$Z_n^*\rightsquigarrow^*\widetilde{\mathcal M}$ in
$\ell^\infty(\mathcal T\times\Gamma)$.

For the drift, use supremum norms over $\alpha\in\mathcal T$ and, where
applicable, $\gamma\in\Gamma$. The oracle density and gradient replacements
in the proof of Theorem~\ref{thm:null-pre} apply also to $A_1$ and $D$:
their gradient factors obey the same moment and Lipschitz bounds in
Assumption~\ref{ass:nuisance-smooth}. With $J=J_n$ and $r_J$ as in that
proof, the weighted density error is
$O_p(n^{-1/2}a_J^{-2}+b_J^2+\sqrt{r_J/(Ja_J)})=o_p(1)$.
The indicator-band bound for residuals and the compact-index uniform law
for sample averages (including centering) give, also under $P_{n,B_0}$
by contiguity,
\[
\|\widehat A_{1,n}-A_1\|_\infty+\|\widehat A_{2,n}-A_2\|_\infty
+\|\widehat D_{\cdot,n}-D\|_\infty+\|\widehat L-L\|_\infty=o_p(1).
\]
Uniform nonsingularity and bounded $B$ now give
\[
\begin{aligned}
\|\widehat H-H\|_\infty
&\le\|\widehat L^{-1}-L^{-1}\|_\infty\|\widehat D_{\cdot,n}\|_\infty
   +\|L^{-1}\|_\infty\|\widehat D_{\cdot,n}-D\|_\infty=o_p(1),\\
\|\widehat d_{n,B}-\widetilde d_B\|_\infty
&\le\|B\|_\infty\big\{\|\widehat A_{1,n}-A_1\|_\infty+
   \|\widehat A_{2,n}-A_2\|_\infty\|\widehat H\|_\infty\\
&\hspace{6em}+\|A_2\|_\infty\|\widehat H-H\|_\infty\big\}=o_p(1),
\end{aligned}
\]
where $\widehat H(\alpha)=\widehat H_{\alpha,n}$.
Thus $\sqrt n\,\widehat M_{n*,B}\rightsquigarrow^*
\widetilde{\mathcal M}+\widetilde d_B$.

Finally, expanding the shifted square gives
\[
\widehat s_{n*,B}^2
=\frac1n\sum_t\omega_t^2\widehat\Psi_{t,n}^2
 +\frac{2\widehat d_{n,B} Z_n^*+\widehat d_{n,B}^2}{n}.
\]
By Theorem~\ref{thm:boot-pre} and contiguity the first term converges
uniformly to $\tilde\sigma^2$ conditionally on the data. The second is
uniformly $O_p^*(n^{-1})$ since $\|Z_n^*\|_\infty=O_p^*(1)$ and
$\|\widehat d_{n,B}\|_\infty=O_p(1)$. Uniform positivity and the continuous
integral-supremum map give the limit of $\widehat T_{n*,B,\Pi}$.
\end{proof}
\renewcommand{\theproposition}{S\arabic{proposition}}
\renewcommand{\theHproposition}{appendix.\arabic{proposition}}
\setcounter{proposition}{0}
\subsection*{Verification for the linear Gaussian design}

For the linear design in Section~\ref{sec:simulation}, consider restricted
quantile regression at each $u\in\mathcal A$.

\begin{proposition}[Linear Gaussian autoregression]
\label{prop:linear-ar2-verification}
Consider the stationary Gaussian AR(2) model in
\eqref{eq:quantile_model}, with $\sigma>0$ and the roots of
$1-\mu_1z-\mu_2z^2$ outside the unit circle. Let
$\mathcal A\Subset(0,1)$ be a compact interval with
$\mathcal T\subset\operatorname{int}(\mathcal A)$, and let
$\Gamma$ be compact and bounded away from zero. Write
\[
X_t=(1,Y_{t-1})^\top,\qquad I_{t-1}=(Y_{t-1},Y_{t-2})^\top,
\qquad Q=E[X_tX_t^\top],
\]
and take a measurable selection
\[
\hat b(u)\in\underset{b}{\arg\min}\,
\sum_{t=1}^n\rho_u(Y_t-\mu_2Y_{t-2}-X_t^\top b),
\qquad u\in\mathcal A,
\]
where $\rho_u(v)=v\{u-\mathbf 1(v<0)\}$.
\begin{enumerate}
\item With $z_u=\Phi^{-1}(u)$,
$b(u)=(\mu_0+\sigma z_u,\mu_1)^\top$,
$f_u=\phi(z_u)/\sigma$, $L_u=f_uQ$, and
$r_{t,u}=\mathbf 1\{\varepsilon_t\le z_u\}-u$,
\begin{equation}
\sup_{u\in\mathcal A}
\left\|\sqrt n\{\hat b(u)-b(u)\}
+L_u^{-1}\frac1{\sqrt n}\sum_{t=1}^nX_tr_{t,u}\right\|
=o_p(1).
\label{eq:ar2-uniform-bahadur}
\end{equation}
The normalized sums of $l_{t,u}=-L_u^{-1}X_tr_{t,u}$ and of the ordinary
and centered weighted quantile scores converge jointly to Gaussian processes
with only same-date covariance terms.
\item For $w_t(\gamma)=e^{\gamma^\top I_{t-1}}-Ee^{\gamma^\top I_{t-1}}$,
the corrected score is
\[
\Psi_t(u,\gamma)
=\left\{w_t(\gamma)-E[w_t(\gamma)X_t^\top]Q^{-1}X_t\right\}r_{t,u}.
\]
Both $L_u$ and $E[l_{t,u}l_{t,u}^\top]$ are uniformly positive
definite on $\mathcal A$, and
\[
\inf_{u\in\mathcal T,\,\gamma\in\Gamma}E[\Psi_t(u,\gamma)^2]>0.
\]
\item Let the Gaussian AR(2) parameter
$\eta=(\mu_0,\mu_1,\mu_2,\sigma)^\top$ vary as
$\eta_{n,d}=\eta-d/\sqrt n$, where $d$ ranges over a fixed finite
set. These paths satisfy the contiguity, density, moment, and centered
score process requirements in Assumption~\ref{ass:local}. Any studentizer
that is uniformly consistent under the null retains the same probability
limit along these paths.
For the full curve their directions are
$B_d(u)=(d_0+d_\sigma z_u,d_1,d_2)^\top$.
For Test 2, taking $d_0=d_1=d_\sigma=0$ gives the constant target
direction $B(u)=d_2$ and
\[
\sqrt n\{\hat b(u)-b(u)\}
=-L_u^{-1}\frac1{\sqrt n}\sum_{t=1}^n
\{X_tr_{nt,u}-E_{n,d}[X_tr_{nt,u}]\}
-Q^{-1}E[X_tY_{t-2}]d_2+o_{P_{n,d}}(1),
\]
uniformly on $\mathcal A$, where $r_{nt,u}$ evaluates the quantile
indicator at the null curve.
\end{enumerate}
The expansions and weak limits also hold with fixed initial values and a
fixed burn-in, using stationary population quantities. They concern the
underlying quantile regression process, not its numerical interpolation.
\end{proposition}

\begin{proof}
Let $V_t=(Y_t,Y_{t-1})^\top$ have stationary mean $\nu$ and covariance
$\Sigma$, and let $C=\bigl(\begin{smallmatrix}\mu_1&\mu_2\\1&0\end{smallmatrix}\bigr)$.
Stability gives $\|C^k\|\le K\varrho^k$ for some $0<\varrho<1$, and
\[
V_k\mid V_0=v\ \sim\
N\!\left(\nu+C^k(v-\nu),\,\Sigma-C^k\Sigma(C^k)^\top\right).
\]
For $k\ge2$ the covariance is nonsingular. By Gaussian relative entropy and
Pinsker's inequality,
\[
\|\mathcal L(V_k\mid V_0=v)-\mathcal L(V_0)\|_{\rm TV}
\le K\varrho^k(1+\|v\|),\qquad \beta_V(k)=O(\varrho^k),
\]
where the second bound follows by integrating over $v$. Since $Z_t$ is a
function of $(V_t,V_{t-1})$, $\beta_Z(k)\le\beta_V(k-1)$ for $k\ge2$,
which gives Assumption~\ref{ass:mixing}(i). Under the null,
\[
E[r_{t,u}\mid\mathcal F_{t-1}]=0,\qquad
\partial_u q_i^0(u)=\sigma/\phi(z_u),\qquad
\nabla_\theta m(i,\theta)=(1,i^\top)^\top.
\]
Full Gaussian support identifies the coefficients. On a compact neighborhood
of $\mathcal A$ in $(0,1)$, $\partial_u q_i^0$ is bounded above and away
from zero. The normal density has bounded first two derivatives, the
linear model has no Taylor remainder, and $E e^{c\|I_{t-1}\|}<\infty$ for every $c>0$.

For $b$ in a fixed compact neighborhood of $b(\mathcal A)$, write
$r_t(u,b)=\mathbf 1\{Y_t-\mu_2Y_{t-2}\le X_t^\top b\}-u$.
The thresholds are halfspaces; the positive subgraphs of
$e^{\gamma^\top i}$ are halfspaces in $(i,\log s)$.
Let $\mathcal F$ collect the ordinary and centered weighted scores, the
coordinate scores $X_{tj}r_t(u,b)$, and their localized subclasses and
linear combinations of fixed finite length with uniformly bounded coefficients. Product covers
under envelope-weighted measures and finite sums give
\[
N(\epsilon\|F\|_{P',p},\mathcal F,L_p(P'))\le C_p\epsilon^{-v_p},
\quad 0<\epsilon<1,\quad 1\le p<\infty,\quad \|F\|_{P',p}<\infty
\]
for every such probability measure $P'$. Here
$F=C(1+|Y_t|^k+\|I_{t-1}\|^k)e^{C\|I_{t-1}\|}$ has all finite moments.
Finite-dimensional indexing and one-sided approximation at thresholds give
permissibility. The bound supplies both the entropy integral and the random
covering condition of \citet[Lemma~2.1]{arconesyu1994}; hence the joint
empirical process converges weakly and is stochastically equicontinuous.
This covering argument also applies to the gradient and corrected scores.

Let $S_n(u)=n^{-1/2}\sum_t X_t r_{t,u}$ and $e_{t,u}=\sigma(\varepsilon_t-z_u)$.
Integrating the quantile score gives
\[
\rho_u(e-a)-\rho_u(e)
=a\{\mathbf 1(e\le0)-u\}+\int_0^a[\mathbf 1(e\le s)-\mathbf 1(e\le0)]\,ds.
\]
Hence, uniformly for $u\in\mathcal A$ and $\|v\|\le M<\infty$,
\begin{align*}
\mathcal Q_{n,u}(v)
&:=\sum_t\{\rho_u(e_{t,u}-X_t^\top v/\sqrt n)-\rho_u(e_{t,u})\}\\
&=v^\top S_n(u)+\tfrac12 v^\top L_uv+o_p(1).
\end{align*}
Indeed, the normal cdf expansion bounds the mean remainder by
$CM^3 E\|X_t\|^3/\sqrt n$. For the centered integral, bounded conditional
density gives, uniformly in $u$,
\[
E\!\left[\|X_t\|^p|r_t(u,b)-r_t(u,b(u))|^p\right]
\le C\|b-b(u)\|\,E\|X_t\|^{p+1}.
\]
Equicontinuity thus bounds the centered integral along $b(u)+sv/\sqrt n$,
$0\le s\le1$, by $o_p(1)$ uniformly. Since $Q$ is positive definite and
$\inf_{u\in\mathcal A}f_u>0$, $\kappa:=\inf_{u\in\mathcal A}\lambda_{\min}(L_u)>0$.
Writing $R_{n,M}=o_p(1)$ for the uniform absolute remainder above,
\[
\inf_{\substack{u\in\mathcal A\\\|v\|=M}}\mathcal Q_{n,u}(v)
\ge \tfrac{\kappa}{2}M^2-M\sup_{u\in\mathcal A}\|S_n(u)\|-R_{n,M}.
\]
Since $\sup_u\|S_n(u)\|=O_p(1)$, the lower bound is positive with arbitrarily
high probability for sufficiently large fixed $M$ as $n\to\infty$.
Convexity and $\mathcal Q_{n,u}(0)=0$ then put all minimizers inside the ball.
For the quadratic part $H_u(v)=v^\top S_n(u)+v^\top L_uv/2$,
\[
v_u=-L_u^{-1}S_n(u),\qquad
H_u(v)-H_u(v_u)=\tfrac12(v-v_u)^\top L_u(v-v_u)
\ge \tfrac{\kappa}{2}\|v-v_u\|^2.
\]
The uniform $o_p(1)$ approximation now yields \eqref{eq:ar2-uniform-bahadur}.

The influence process follows by the continuous map $S\mapsto-L_u^{-1}S(u)$.
All cross-date covariances vanish by predictability and the conditional
mean-zero identity above. With $a_\gamma=Q^{-1}E[X_tw_t(\gamma)]$,
independence of $\varepsilon_t$ and $I_{t-1}$ gives
\[
E[l_{t,u}l_{t,u}^\top]=\frac{u(1-u)}{f_u^2}Q^{-1},\qquad
E[\Psi_t(u,\gamma)^2]=u(1-u)E[(w_t(\gamma)-a_\gamma^\top X_t)^2].
\]
The first matrix is uniformly positive definite. A zero residual variance
would make $e^{\gamma^\top i}$ affine in $i_1$ on all of $\mathbb R^2$
by full support and continuity. Its Hessian $e^{\gamma^\top i}\gamma\gamma^\top$
is nonzero for $\gamma\ne0$. The exponential envelope gives continuity in
$\gamma$; compactness of $\Gamma$ away from zero gives the uniform lower
bound, also for the variance of the centered weight.

For the local paths, set $Z_t^0=(1,Y_{t-1},Y_{t-2})^\top$ and
\[
s_t=\sigma^{-1}\begin{pmatrix}\varepsilon_t Z_t^0\\ \varepsilon_t^2-1\end{pmatrix},\quad
\Delta_n=n^{-1/2}\sum_t s_t,\qquad
\mathcal I=\sigma^{-2}\operatorname{diag}\!\left(E[Z_t^0(Z_t^0)^\top],2\right).
\]
Expanding the Gaussian log likelihood gives
\begin{equation}
\log\frac{dP_{n,d}}{dP_{n,0}}
=-d^\top\Delta_n-\tfrac12d^\top\mathcal I d+o_{P_0}(1).
\label{eq:ar2-lan}
\end{equation}
The ergodic average negative Hessian tends to $\mathcal I$; the cubic
remainder is $O_p(n^{-1/2})$ by Gaussian moments and the stationary initial
log density ratio is $o_p(1)$. The Gaussian innovation moments give
\[
E_0[s_t\mid\mathcal F_{t-1}]=0,\qquad
E_0[s_ts_t^\top\mid\mathcal F_{t-1}]
=\sigma^{-2}\operatorname{diag}\!\left(Z_t^0(Z_t^0)^\top,2\right).
\]
Thus only the same-date terms contribute to the covariance of $\Delta_n$,
giving $\mathcal I$. Adjoining $s_t$ to the score class gives a joint Gaussian
limit with $\Delta\sim N(0,\mathcal I)$, so under the null
\[
L_{n,d}:=\frac{dP_{n,d}}{dP_{n,0}}\rightsquigarrow
L_d:=\exp(-d^\top\Delta-\tfrac12d^\top\mathcal I d),
\qquad E_0 L_{n,d}=E L_d=1.
\]
The likelihood ratios are therefore uniformly integrable and the local
laws are contiguous. For any finite vector $G_n$ of normalized score
coordinates with null limit $G$, joint convergence and uniform
integrability give, for bounded continuous $\varphi$,
\[
E_{n,d}\varphi(G_n)=E_0[L_{n,d}\varphi(G_n)]
\longrightarrow E[L_d\varphi(G)].
\]
For this tilted law, the Gaussian moment generating function gives
\[
E[L_d e^{a^\top G}]
=\exp\!\left\{-a^\top\operatorname{Cov}(G,\Delta)d
+\tfrac12 a^\top\operatorname{Var}(G)a\right\}.
\]
Thus the mean shifts by $-\operatorname{Cov}(G,\Delta)d$ and the covariance
is unchanged. Contiguity transfers stochastic equicontinuity, so these
finite-dimensional laws determine the local process limit.

Write $a_t=d_0+d_1Y_{t-1}+d_2Y_{t-2}$. Uniformly on $\mathcal A$,
\begin{align*}
P_{n,d}\{Y_t\le q^0_{I_{t-1}}(u)\mid I_{t-1}\}-u
&=\Phi\!\left(\frac{\sigma z_u+a_t/\sqrt n}{\sigma-d_\sigma/\sqrt n}\right)-u\\
&=\frac{f_u}{\sqrt n}(a_t+d_\sigma z_u)
  +O\!\left(\frac{1+\|I_{t-1}\|^2}{n}\right).
\end{align*}
Since $E[r_{t,u}\varepsilon_t]=-\phi(z_u)$ and
$E[r_{t,u}(\varepsilon_t^2-1)]=-z_u\phi(z_u)$, the local means of the
normalized scores converge to the tilt shifts. Centering at the local
means therefore restores the null process limit. Gaussian means and
covariances vary continuously on a stable neighborhood, giving the uniform
envelope moments and localized density convergence of
Assumption~\ref{ass:local}(ii). Contiguity transfers equicontinuity,
null $o_p(1)$ remainders, and any established uniform studentizer
consistency; density-based studentizers still require Assumption~\ref{ass:kernel}.
All bounds apply over the fixed finite set of $d$. For Test~2,
\[
-L_u^{-1}\sqrt n E_{n,d}[X_t r_{nt,u}]
=-L_u^{-1}f_u E[X_tY_{t-2}]d_2+o(1)
=-Q^{-1}E[X_tY_{t-2}]d_2+o(1).
\]
Substituting this mean in \eqref{eq:ar2-uniform-bahadur} under $P_{n,d}$
gives the drift in part~3.

Finally, couple the fixed-start state $V_t^\circ$ and the stationary $V_t$
using the same innovations. For $m_n=\lceil c\log n\rceil$ and $c$ large,
\[
V_t^\circ-V_t=C^t(V_0^\circ-V_0),\qquad
\sup_{t\ge m_n}\|V_t^\circ-V_t\|=O_p(n^{-3}).
\]
The first $m_n$ terms contribute $O_p(m_n/\sqrt n)=o_p(1)$ to each
normalized score. On the remaining dates the smooth weight increments
are negligible by the Gaussian envelope. Indicator changes lie in a
threshold strip of width $O(n^{-3})$ on events of arbitrarily high
probability. Its weighted mean is $O(n^{-3})$ by bounded conditional density;
the centered strip increments are uniformly $o_p(1)$ by equicontinuity.
The same bounds apply to the local objective increments, and the smooth
Gaussian likelihood expansion is unchanged up to $o_p(1)$. Thus neither
fixed initial values nor a fixed burn-in changes the stated limits.
\end{proof}
\section*{S5. Additional Monte Carlo Results}
\phantomsection
\makeatletter
\def\@currentlabel{S5}
\makeatother
\label{sec:simulation-supp}

This section complements the linear results in Section~\ref{sec:simulation}
with the nonlinear experiments, complete size comparisons, and additional
power results for the three aggregation rules. It also gives implementation
details for all comparisons.

\subsection*{S5.1. Implementation details}
\phantomsection
\makeatletter
\def\@currentlabel{S5.1}
\makeatother
\label{app:simulation-implementation}

\paragraph{Nonlinear design.}
The nonlinear experiment replaces the first lag in the linear AR(2) model by
$s_\rho(Y_{t-1})$, where
\begin{equation}
s_\rho(y)=\operatorname{sgn}(y)|y|^\rho,
\qquad
Y_t=\mu_0+\mu_1s_\rho(Y_{t-1})+\mu_2Y_{t-2}+\sigma\varepsilon_t,
\label{eq:simulation-nonlinear}
\end{equation}
with $\varepsilon_t\overset{\mathrm{iid}}{\sim}\mathcal N(0,1)$.
The null values are
$(\mu_0,\mu_1,\rho,\mu_2,\sigma)=(0.5,0.4,0.8,-0.2,1)$.
Its conditional quantile function is
\[
m(I_{t-1},\theta_0(\alpha))
=\beta_0(\alpha)+\mu_1s_\rho(Y_{t-1})+\mu_2Y_{t-2},
\qquad \beta_0(\alpha)=\mu_0+\sigma\Phi^{-1}(\alpha).
\]
Test~1 fixes this full null quantile curve, and Test~2 tests
$H_0:\rho=0.8$ while estimating the remaining quantile coefficients under
that restriction. The corresponding linear tests fix the full null curve
or test $H_0:\mu_2=-0.2$, as described in the main text.

\paragraph{Samples and alternatives.}
Both initial lags are zero, and 300 burn-in observations are discarded.
We use $N$ for the retained series length and $n=N-2$ for the number of
observations entering the statistic. Size is evaluated at
$N\in\{200,500,1000\}$. The power experiments use $N=500$ for Test~1 and
$N=1000$ for Test~2. Each reported individual experiment uses 2,000 Monte
Carlo replications. The main tests use 1,000 independent standard normal
multiplier draws in each replication.

Fixed alternatives change one DGP primitive by $\delta$, holding the other
primitives and the tested null at their baseline values. These deviations
are not scaled by $n^{-1/2}$. In the linear model, each primitive uses 13
equally spaced deviations on $[-0.2,0.2]$. In the nonlinear Test~1 experiment,
the numbers of deviations for $(\mu_0,\mu_1,\rho,\mu_2,\sigma)$ are
$(13,21,23,13,19)$, over $[-0.2,0.2]$, $[-0.3,0.3]$, $[-0.7,0.2]$,
$[-0.2,0.2]$, and $[-0.2,0.3]$, respectively. The nonlinear Test~2
experiment uses 27 deviations on $[-0.7,0.2]$.
The nonlinear deviation grids are not all equally spaced.

\paragraph{Restricted estimation and weights.}
For Test~2 in the linear model, quantile regression fits
$Y_t+0.2Y_{t-2}$ on $x_t=(1,Y_{t-1})^\top$. In the nonlinear model, it fits
$Y_t$ on $x_t=(1,s_{0.8}(Y_{t-1}),Y_{t-2})^\top$, where
$s_\rho(y)=\operatorname{sgn}(y)|y|^\rho$. Each coefficient may vary across
quantiles. The intercept absorbs $\mu_0+\sigma\Phi^{-1}(\alpha)$, so
$\mu_0$ and $\sigma$ are not estimated separately at each quantile.
The nuisance curve is fitted at 100 equally spaced nodes in $[0.15,0.85]$.
Piecewise linear interpolation evaluates it between nodes, and linear
extrapolation is used outside the fitting interval.

Test~1 uses nine equally spaced test quantiles in $[0.1,0.9]$, and Test~2
uses seven in $[0.2,0.8]$. If $E_k$ denotes $k$ equally spaced points on
$[-1,1]$, including the endpoints, the direction grids are $E_5^2$ for the
linear Test~1 experiment and $E_5^2\setminus\{(0,0)^\top\}$ for the linear
Test~2 experiment. These grids contain 25 and 24 directions. Both nonlinear
experiments use $E_{10}^2$, which contains 100 directions and excludes zero.

With $I_{t-1}=(Y_{t-1},Y_{t-2})^\top$, Test~1 uses
$w_t(\gamma)=\exp(\gamma^\top I_{t-1})$. Test~2 first standardizes each lag
using its sample mean and standard deviation with divisor $n$. An additive
$10^{-5}$ is included in the standard deviation. If $z_t$ denotes the
resulting vector, the centered weight is
\[
 w_t(\gamma)=\exp(\gamma^\top z_t)
             -n^{-1}\sum_{s=1}^n\exp(\gamma^\top z_s).
\]
The regressions and quantile indicators use the original data units.

\paragraph{Density estimation and correction for pre-estimation.}
In both tests, the density calculation draws $J=1000$ auxiliary indices
$U_j$ uniformly from $[0.01,0.99]$. With Gaussian kernel $\phi$ and
$h=\max\{10^{-4},(\log n/n)^{1/5}\}$, the implemented estimator is
\[
 \widehat f_{t\alpha}
 =\frac{0.98}{Jh}\sum_{j=1}^J
   \phi\!\left(\frac{\widehat q_t(\alpha)-\widehat q_t(U_j)}{h}\right).
\]
Test~1 evaluates the fixed null curve directly. Test~2 uses the fitted curve
described above, including its endpoint extrapolation at auxiliary indices.
The auxiliary indices are random and are distinct from the deterministic
test and fitting grids.

For Test~2, write
$\psi_{t\alpha}=\mathbf 1\{Y_t\leq\widehat q_t(\alpha)\}-\alpha$ and
\begin{align*}
 L_\alpha&=\frac1n\sum_t\widehat f_{t\alpha}x_tx_t^\top,
 &\ell_{t\alpha}&=-L_\alpha^{-1}x_t\psi_{t\alpha},\\
 A_{2,\gamma\alpha}&=\frac1n\sum_t
               w_t(\gamma)\widehat f_{t\alpha}x_t^\top.
\end{align*}
The observed numerator is the plug-in moment
$n^{-1}\sum_t w_t(\gamma)\psi_{t\alpha}$. The variance and multiplier
process instead use the complete corrected score
\[
 \xi_{t,\gamma\alpha}
   =w_t(\gamma)\psi_{t\alpha}+A_{2,\gamma\alpha}\ell_{t\alpha},
 \qquad
 V_{\gamma\alpha}=\frac1n\sum_t\xi_{t,\gamma\alpha}^2.
\]
The same observed $V_{\gamma\alpha}$ studentizes the statistic and the null
and shifted multiplier processes. Nuisance estimates are held fixed within
these bootstrap calculations.

For the shift used in penalty selection, let $g_{t\alpha}$ be the derivative
of the quantile with respect to the target parameter. In the linear model,
$g_{t\alpha}=Y_{t-2}$. In the nonlinear model,
\[
g_{t\alpha}=\widehat\beta_1(\alpha)s_{0.8}(Y_{t-1})\log|Y_{t-1}|,
\]
with the continuous extension at zero. The unit reference direction uses
\begin{align*}
 d_{\gamma\alpha}&=A_{1,\gamma\alpha}+A_{2,\gamma\alpha}H_\alpha,\\
 A_{1,\gamma\alpha}&=\frac1n\sum_t
             w_t(\gamma)\widehat f_{t\alpha}g_{t\alpha},
 &H_\alpha&=-L_\alpha^{-1}\frac1n\sum_t
             \widehat f_{t\alpha}x_tg_{t\alpha}.
\end{align*}
In Test~1, the score is simply $w_t(\gamma)\psi_{t\alpha}$, evaluated at the
null curve. Its observed variance is
$\alpha(1-\alpha)n^{-1}\sum_t w_t(\gamma)^2$.
For multiplier draw $b$, the bootstrap variance is
$\alpha(1-\alpha)n^{-1}\sum_t w_t(\gamma)^2\omega_{tb}^2$.
The shift is the weighted fitted-density average of the null quantile
gradient in the chosen reference direction.

\paragraph{Aggregation and penalty selection.}
The candidate penalties are $\{0,0.1,\ldots,1\}$ in both linear experiments,
$\{0,0.1,\ldots,2\}$ in nonlinear Test~1, and
$\{0,0.1,\ldots,0.9\}$ in nonlinear Test~2. The reference class contains
one direction: $(0.5,0.5,0.5)^\top$ in the coordinates
$(\beta_0,\mu_1,\mu_2)$ for linear Test~1,
$(0.2,0.2,0.2,0.2)^\top$ in $(\beta_0,\mu_1,\rho,\mu_2)$ for nonlinear
Test~1, and a unit direction in the target parameter for Test~2. These
reference directions determine the shifted bootstrap criterion. They are
different from the fixed DGP deviations used to evaluate power.

For a nonnegative studentized process $Q_{\gamma\alpha}$ and $A$ test
quantiles, the three implemented aggregation rules are
\begin{align*}
 T^{\mathrm{CvM\text{-}KS}}_\lambda(Q)
 &=\max_\gamma\left\{A^{-1}\sum_\alpha Q_{\gamma\alpha}^2
                         -\lambda\|\gamma\|_1\right\},\\
 T^{\mathrm{KS\text{-}KS}}_{\eta,\lambda}(Q)
 &=\max_{\gamma,\alpha}
       \{Q_{\gamma\alpha}-\eta|\alpha|-\lambda\|\gamma\|_1\},\\
 T^{\mathrm{KS\text{-}CvM}}_\lambda(Q)
 &=A^{-1}\sum_\alpha
       \left[\max_\gamma\{Q_{\gamma\alpha}-\lambda\|\gamma\|_1\}\right]^2.
\end{align*}
KS--KS additionally searches $\eta\in\{0,0.05,\ldots,0.25\}$.
None of these implementations applies a positive-part truncation.
Thus the aggregation comparisons also change the way penalties enter the
statistic, rather than varying only the order of maximization and averaging.

At each nominal level, each candidate penalty has its own empirical null
critical value, obtained with the higher-quantile convention. The adaptive
choice maximizes the shifted bootstrap rejection frequency. The same
multiplier sequences are reused across quantiles, directions, penalties,
and the null and shifted processes. Ties within $10^{-12}$ are resolved
by the smallest penalty in lexicographic order. The unpenalized procedure
sets every penalty to zero. The reported multiplier $p$-value is
$(1+\#\{T_b^0\geq T_{\mathrm{obs}}\})/(B+1)$, with $B=1000$.

\paragraph{Comparison methods.}
The CFV--KS benchmark adapts the subsampling approach of
\citet{chernozhukov2005} to the linear parameter restrictions. The EV-style
CvM benchmark adapts the quantile-hit process and subsampling approach of
\citet{escanciano2010} to the restrictions in both DGPs.

CFV--KS uses recentered contiguous-block subsampling. The EV-style
statistic integrates a squared quantile-hit process with Gaussian kernel
$K(I_{t-1},I_{s-1})=\exp(-\|I_{t-1}-I_{s-1}\|^2/2)$ on raw lag information.
It uses uncentered contiguous-block subsampling, with permitted nuisance
parameters re-estimated under the target restriction within each block.
Both methods use all contiguous blocks of length $\lfloor n^{0.6}\rfloor$.
This gives 458 blocks when $N=500$ and 936 when $N=1000$.
Their $p$-values are empirical upper-tail proportions from these blocks.
The EV-style comparison therefore concerns this implementation, rather
than a replication of the recentering and tuning choices in
\citet{escanciano2010}.

The coefficient tests use unrestricted QR in the linear model and profile
NQR in the nonlinear model. The exponent in NQR is estimated separately
at each quantile over $[0.01,2]$, using a 129-point profile grid followed
by numerical refinement. Sup-Wald tests the full quantile coefficient
vector over the nine Test~1 quantiles. Sup-$t$ tests only $\mu_2$ or $\rho$
over the seven Test~2 quantiles. Calibration uses 1,000 shared Gaussian
multiplier draws with Bartlett dependence and block length
$\max\{2,\lfloor n^{1/3}\rfloor\}$. The observed and bootstrap statistics
use the same HAC sandwich covariance. These tests use the plus-one
multiplier $p$-value.

\paragraph{Monte Carlo uncertainty and matched samples.}
Adaptive and unpenalized procedures and the three aggregation rules are
evaluated on matched samples. Within each DGP, test, and sample size, power
alternatives share innovation, auxiliary-density, and multiplier random
streams. The known-curve parameter panels reuse one null experiment per
DGP, so their repeated zero points are not independent size estimates.

With 2,000 replications, the binomial Monte Carlo standard error is about
0.0049 at probability 0.05 and 0.0067 at probability 0.10.

\subsection*{S5.2. Complete size results}
\phantomsection
\makeatletter
\def\@currentlabel{S5.2}
\makeatother
\label{app:simulation-size}

Table~\ref{tab:size-nonlinear} reports the nonlinear results for both tests.
The corresponding linear results appear in Table~\ref{tab:size_comparison}
in the main text.
The CvM--KS null frequencies at $N=500$ for Test~1 and $N=1000$ for Test~2
are the zero points of the corresponding power curves.

\begin{table}[htbp]
\centering
\setstretch{1}
\footnotesize
\caption{Empirical size of the three aggregation schemes in the nonlinear AR(2) model}
\label{tab:size-nonlinear}
\setlength{\tabcolsep}{3pt}
\begin{tabular}{rrcccccc}
\toprule
& & \multicolumn{3}{c}{Adaptive} & \multicolumn{3}{c}{Unpenalized} \\
\cmidrule(lr){3-5}\cmidrule(lr){6-8}
$N$ & Nominal & CvM--KS & KS--KS & KS--CvM & CvM--KS & KS--KS & KS--CvM \\
\midrule
\multicolumn{8}{l}{\textit{Panel A. Full null curve (Test 1)}} \\
200 & 0.05 & 0.0595 & 0.0430 & 0.0595 & 0.0580 & 0.0385 & 0.0575 \\
 & 0.10 & 0.1165 & 0.0875 & 0.1145 & 0.1170 & 0.0855 & 0.1095 \\
\addlinespace
500 & 0.05 & 0.0520 & 0.0525 & 0.0560 & 0.0540 & 0.0485 & 0.0545 \\
 & 0.10 & 0.1110 & 0.1045 & 0.1070 & 0.1080 & 0.1045 & 0.1035 \\
\addlinespace
1000 & 0.05 & 0.0545 & 0.0430 & 0.0545 & 0.0495 & 0.0470 & 0.0485 \\
 & 0.10 & 0.1130 & 0.1030 & 0.0950 & 0.1120 & 0.0970 & 0.0990 \\
\midrule
\multicolumn{8}{l}{\textit{Panel B. Single target with pre-estimation (Test 2)}} \\
200 & 0.05 & 0.2035 & 0.6740 & 0.5190 & 0.1665 & 0.5940 & 0.4055 \\
 & 0.10 & 0.3575 & 0.7975 & 0.6930 & 0.3095 & 0.7280 & 0.5690 \\
\addlinespace
500 & 0.05 & 0.0685 & 0.2995 & 0.1880 & 0.0590 & 0.2420 & 0.1295 \\
 & 0.10 & 0.1470 & 0.4195 & 0.3235 & 0.1220 & 0.3410 & 0.2505 \\
\addlinespace
1000 & 0.05 & 0.0655 & 0.1950 & 0.1240 & 0.0600 & 0.1600 & 0.0990 \\
 & 0.10 & 0.1350 & 0.2835 & 0.2255 & 0.1205 & 0.2465 & 0.1880 \\
\bottomrule
\end{tabular}
\par\smallskip
\begin{minipage}{.96\linewidth}\footnotesize
\textit{Notes:} Entries are rejection proportions from 2,000 replications,
reported to four decimals, with $n=N-2$.
\end{minipage}
\end{table}

Without pre-estimation, all three aggregation rules have rejection
frequencies generally close to the nominal levels in the nonlinear design.
The nonlinear Test~2 experiment exhibits substantial over-rejection at
$N=200$. At nominal level 0.10, CvM--KS rejects with frequencies 0.3575
and 0.3095 for the adaptive and unpenalized procedures. These decrease
to 0.1350 and 0.1205 at $N=1000$. The other two aggregations reject
considerably more often under the null. At $N=1000$ and nominal level
0.10, their adaptive and unpenalized frequencies are 0.2835 and 0.2465
for KS--KS, and 0.2255 and 0.1880 for KS--CvM. Hence higher rejection
frequencies under alternatives do not by themselves favor these rules.
\clearpage

Table~\ref{tab:comparison-size} collects the size results for the comparison
methods. At nominal level 0.10, the Test~1
Sup-Wald benchmark rejects with frequencies 0.2200 and 0.4085 in the
linear and nonlinear designs. In Test~2, QR Sup-$t$ is close to nominal
size at 0.1050, while NQR Sup-$t$ rejects at 0.1230. Nonlinear EV-style
CvM is conservative at 0.0675. These differences in actual size are
relevant when comparing the power curves.

\begin{table}[!htbp]
\centering
\setstretch{1}
\footnotesize
\setlength{\tabcolsep}{6pt}
\caption{Empirical size of the comparison methods}
\label{tab:comparison-size}
\begin{tabular}{@{}ll l rr@{}}
\toprule
 & & & \multicolumn{2}{c}{Nominal level} \\
\cmidrule(l){4-5}
DGP & $N$ & Method & 5\% & 10\% \\
\midrule
\multicolumn{5}{@{}l}{\textit{Panel A. Joint restrictions}} \\
\addlinespace[2pt]
Linear & 200 & QR Sup-Wald & 0.2590 & 0.3545 \\
 & 500 & CFV--KS & 0.0350 & 0.0785 \\
 & & EV-style CvM & 0.0510 & 0.1030 \\
 & & QR Sup-Wald & 0.1420 & 0.2200 \\
 & 1000 & QR Sup-Wald & 0.1090 & 0.1755 \\
\addlinespace
Nonlinear & 200 & NQR Sup-Wald & 0.5425 & 0.6135 \\
 & 500 & EV-style CvM & 0.0425 & 0.0900 \\
 & & NQR Sup-Wald & 0.3265 & 0.4085 \\
 & 1000 & NQR Sup-Wald & 0.1760 & 0.2545 \\
\addlinespace[5pt]
\multicolumn{5}{@{}l}{\textit{Panel B. Single target (Test 2)}} \\
\addlinespace[2pt]
Linear & 200 & CFV--KS & 0.0760 & 0.1410 \\
 & & EV-style CvM & 0.0425 & 0.0755 \\
 & & QR Sup-$t$ & 0.0845 & 0.1460 \\
 & 500 & CFV--KS & 0.0740 & 0.1300 \\
 & & EV-style CvM & 0.0405 & 0.0810 \\
 & & QR Sup-$t$ & 0.0655 & 0.1200 \\
 & 1000 & CFV--KS & 0.0665 & 0.1195 \\
 & & EV-style CvM & 0.0490 & 0.0875 \\
 & & QR Sup-$t$ & 0.0535 & 0.1050 \\
\addlinespace
Nonlinear & 200 & EV-style CvM & 0.0145 & 0.0305 \\
 & & NQR Sup-$t$ & 0.1525 & 0.2075 \\
 & 500 & EV-style CvM & 0.0290 & 0.0575 \\
 & & NQR Sup-$t$ & 0.1150 & 0.1785 \\
 & 1000 & EV-style CvM & 0.0360 & 0.0675 \\
 & & NQR Sup-$t$ & 0.0815 & 0.1230 \\
\bottomrule
\end{tabular}
\par\smallskip
\begin{minipage}{\textwidth}\footnotesize
\textit{Notes:} Rates are rejection proportions from 2,000 replications.
The effective sample size is $n=N-2$. In Panel A, the $N=500$ rows
give the null points of the Test~1 power curves. The $N=200$ and
$N=1000$ rows are additional joint coefficient tests over nine quantiles
in $[0.1,0.9]$. In Panel B, the $N=1000$ rows give the null points of
the Test~2 power curves. All coefficient tests use unrestricted QR or NQR estimates.
\end{minipage}
\end{table}
\clearpage

\subsection*{S5.3. Power comparisons in the nonlinear model}
\phantomsection
\makeatletter
\def\@currentlabel{S5.3}
\makeatother
\label{app:simulation-power}

Figures~\ref{fig:nonlinear-known} and~\ref{fig:pre-full} report the
nonlinear power comparisons at nominal level 0.10. Together with the
linear results in Figures~\ref{fig:linear-known} and~\ref{fig:linear-qr}
of the main text, they cover every parameter deviation in the power
experiments. Each curve retains its complete deviation range and its
actual null rejection frequency at $\delta=0$.

Lines join the simulated design points without smoothing. Shading shows the
adaptive minus unpenalized rejection frequency relative to zero on the
same vertical scale, including negative differences below zero. The dotted
horizontal reference marks the nominal level. The matching null results
are reported in Section~\ref{app:simulation-size}.

Without pre-estimation, the effect of penalization varies across the
nonlinear parameter directions. For the exponent at $\delta=-0.7$, the
adaptive and unpenalized CvM--KS rejection frequencies are 0.6145 and
0.6445. This departure favors the unpenalized statistic.
The high rejection frequencies of NQR Sup-Wald under alternatives are
accompanied by substantial over-rejection under the null.

\begin{figure}[htbp]
\centering
\includegraphics[width=\linewidth]{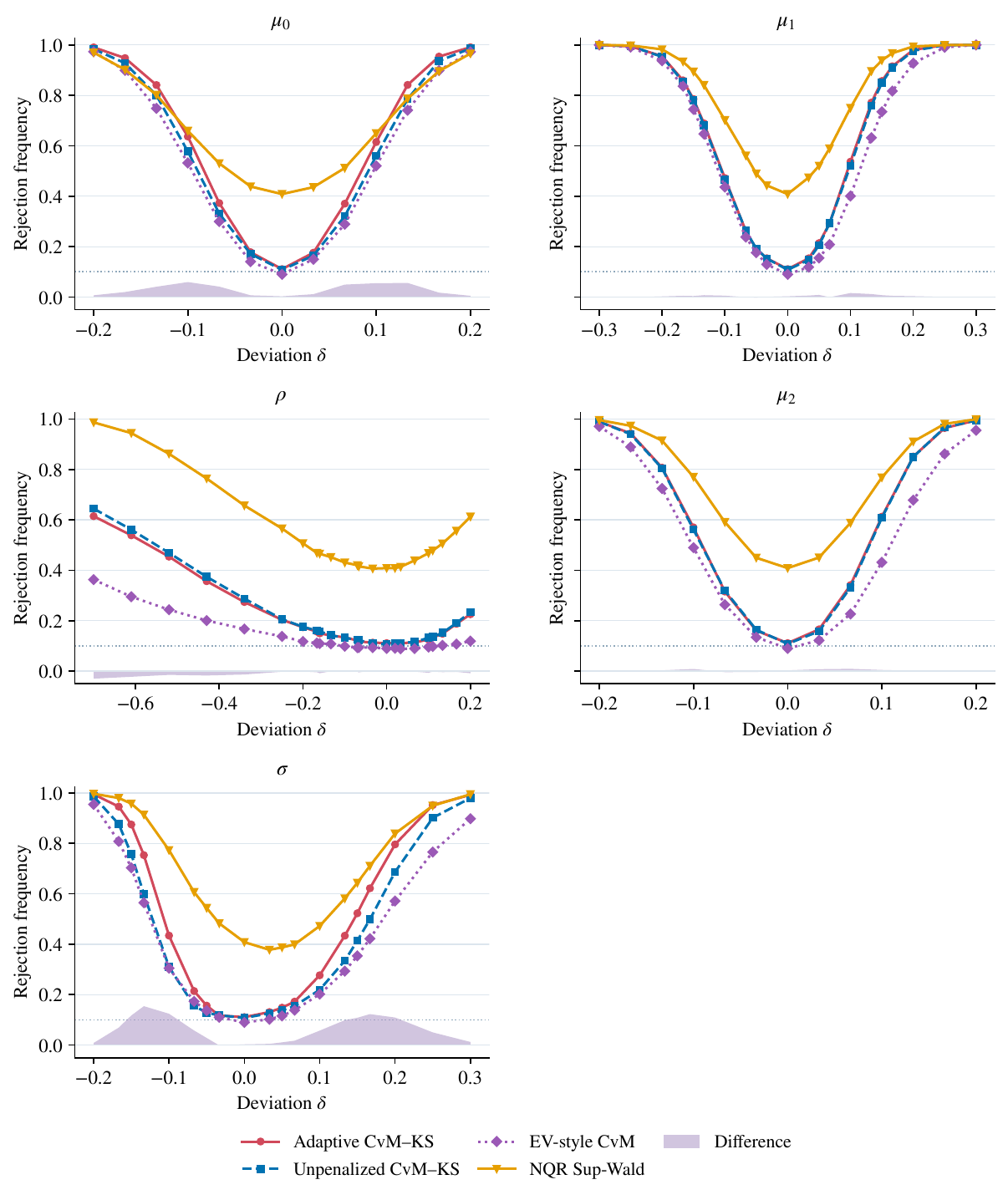}
\caption{Test~1 in the nonlinear AR(2) model. Rejection frequencies at
$N=500$ and nominal level 0.10 for deviations in all five DGP primitives.
The baseline exponent is $\rho=0.8$. The panels share one null experiment
and retain their respective deviation ranges. The comparison methods are
NQR Sup-Wald and EV-style CvM.}
\label{fig:nonlinear-known}
\label{fig:power_no_pre2}
\end{figure}
\clearpage

With pre-estimation, the relative performance also varies with the
departure from $\rho=0.8$. At $\delta=-0.3$, the adaptive and unpenalized
CvM--KS rejection frequencies are 0.3705 and 0.3370, compared with 0.6995
for NQR Sup-$t$ and 0.2210 for EV-style CvM.
Test~1 and Test~2 differ in both the hypothesis and the sample size, so
their contrast does not isolate the effect of pre-estimation.

\begin{figure}[htbp]
\centering
\includegraphics[width=.9\linewidth]{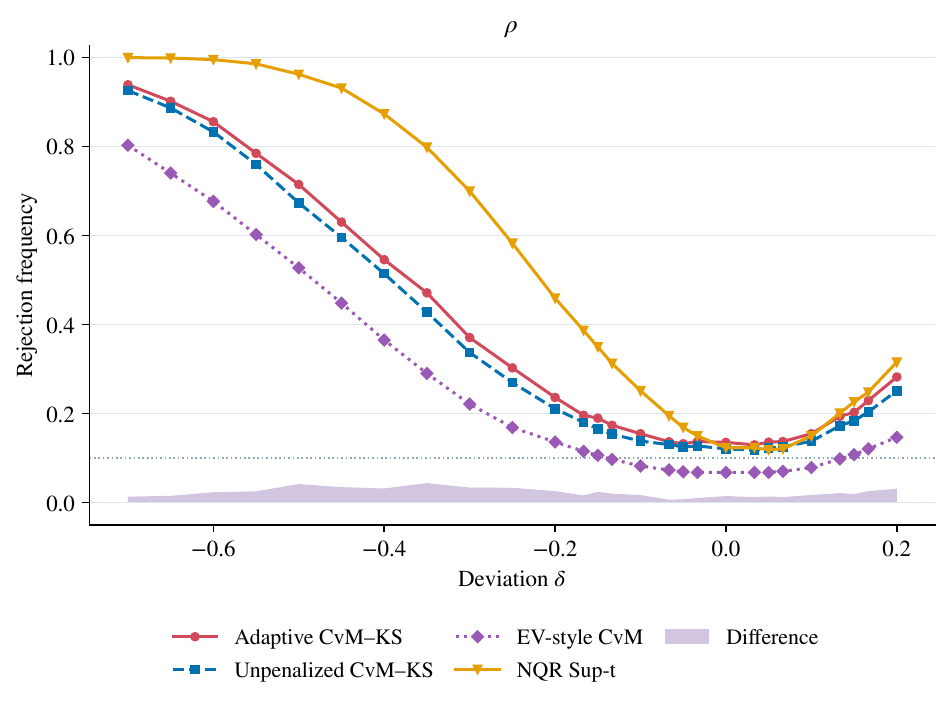}
\caption{Test~2 in the nonlinear AR(2) model. Rejection frequencies for
$H_0:\rho=0.8$ at $N=1000$ and nominal level 0.10, over all 27 deviations.
Nuisance coefficients are estimated by restricted QR. The comparison
methods are NQR Sup-$t$ and EV-style CvM. The adaptive and unpenalized
CvM--KS null frequencies are 0.1350 and 0.1205.}
\label{fig:pre-full}
\label{fig:power_pre2}
\end{figure}
\clearpage

\subsection*{S5.4. Alternative aggregation schemes}
\phantomsection
\makeatletter
\def\@currentlabel{S5.4}
\makeatother
\label{app:simulation-differences}

Figures~\ref{fig:aggregation-linear} and~\ref{fig:aggregation-nonlinear}
add the Test~2 power results for KS--KS and KS--CvM over the same complete
target deviations. The zero points correspond to the original power
samples. Their rejection frequencies under alternatives must be read
alongside the substantial null over-rejection documented above. Shading
uses the same difference convention as in Section~S5.3.

\begin{figure}[htbp]
\centering
\includegraphics[width=\linewidth]{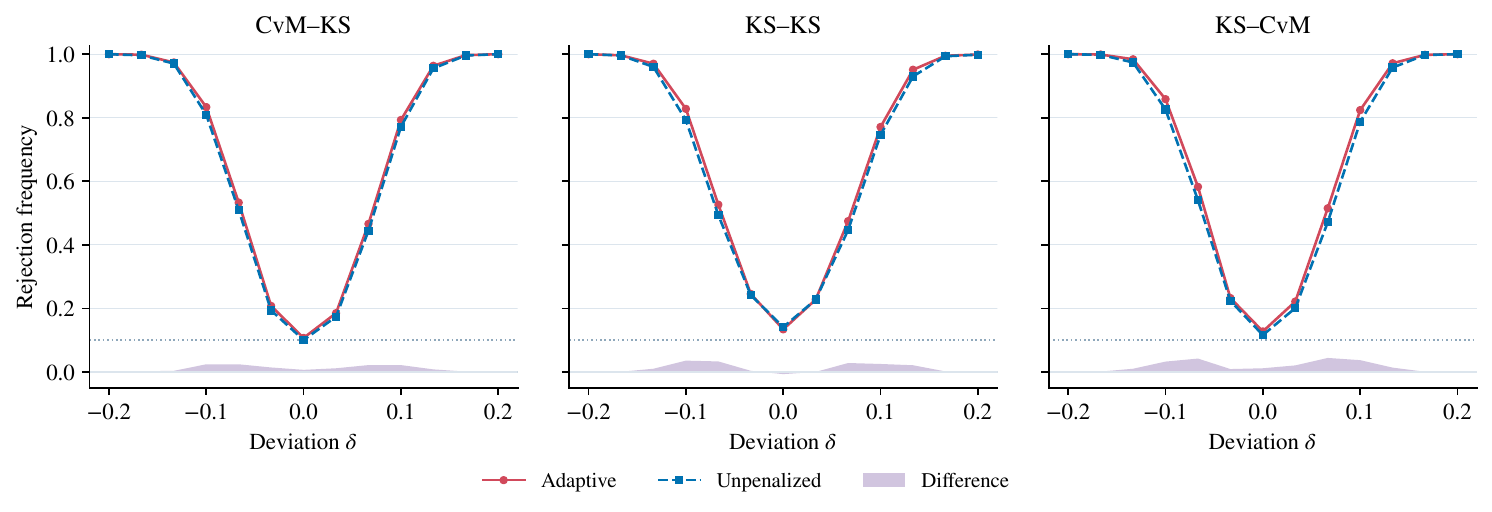}
\caption{Rejection frequencies for the three aggregation rules in linear
Test~2, with $N=1000$ and nominal level 0.10. Each rule uses its implemented
penalty and selection criterion on the same Monte Carlo samples.}
\label{fig:aggregation-linear}
\end{figure}
\clearpage

\begin{figure}[htbp]
\centering
\includegraphics[width=\linewidth]{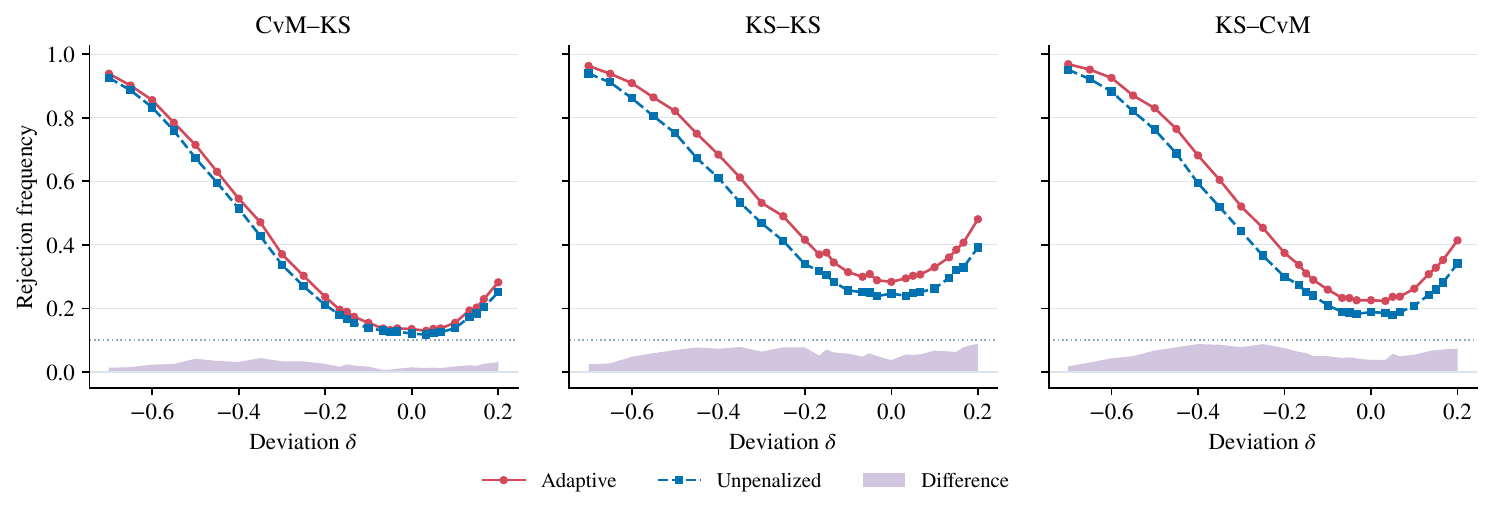}
\caption{Rejection frequencies for the three aggregation rules in
nonlinear Test~2, with $N=1000$ and nominal level 0.10. The full interval
$[-0.7,0.2]$ is shown, including the actual null rejection frequency for
each rule.}
\label{fig:aggregation-nonlinear}
\end{figure}
\clearpage

\section*{S6. Empirical Implementation and Additional Results}
\phantomsection
\makeatletter
\def\@currentlabel{S6}
\makeatother
\label{app:empirical-implementation}

\subsection*{S6.1. Data and estimation}

The growth series is the seasonally adjusted annualized quarterly change
in real GDP. The responses are $Y_t^{(1)}=y_{t+1}$ and
$Y_t^{(4)}=(y_{t+1}+y_{t+2}+y_{t+3}+y_{t+4})/4$.
For $h=1$, NFCI is the last weekly observation in each quarter. For
$h=4$, it is the mean of weekly observations assigned to quarters by
their recorded Friday dates. Regressors are measured at the forecast
origin $t$. The full samples contain 219 and 216 observations,
respectively. The sample ending before 2020 retains only observations
whose entire response window ends by 2019Q4, giving 195 and 192
observations. The OLS reference value is computed within each sample
and treated as a fixed numerical benchmark.

At each fixed candidate $c$, restricted quantile regressions of
$Y_t^{(h)}-c\,\mathrm{NFCI}_t$ on $(1,y_t)$ are fitted at 200 equally
spaced levels in $[0.05,0.95]$. Adding $c\,\mathrm{NFCI}_t$ back to the
fitted values, with linear interpolation and endpoint extrapolation,
gives the restricted fitted curve $\widehat q_t(u)$. The main statistic
averages over 100 equally spaced levels in $[0.10,0.90]$, and the left
tail statistic averages over 25 levels in $[0.05,0.30]$. Both use equal
weights. The direction grid has 40 equally spaced values on each axis
of $[-4,4]^2$ and excludes the origin. The conditioning variables
entering the exponential weights are centered and divided by their
empirical standard deviation, computed with divisor $n$, plus $10^{-8}$.

Let $\widehat\sigma_r$ be the standard deviation of
$Y_t^{(h)}-\widehat q_t(0.5)$, also computed with divisor $n$. The
bandwidth and conditional density estimator are
\[
b_n=\max\{0.9\widehat\sigma_r n^{-1/5},\,0.05\},
\qquad
\widehat f_t^{\,\mathrm{emp}}(\alpha)
=\frac{0.98}{Jb_n}\sum_{j=1}^{J}
K\!\left(
\frac{\widehat q_t(\alpha)-\widehat q_t(U_j)}{b_n}
\right),
\quad J=1000,
\]
where $K$ is the standard Gaussian density and the $U_j$ are independent
uniform draws on $[0.01,0.99]$. The same auxiliary draws are used across
restrictions and numerical variants. The factor $0.98$ is the length
of the auxiliary interval, as in \eqref{eq:density_est}.
We add $10^{-7}I$ to $\widehat L_\alpha$ before inversion.
Restricted quantile regressions are solved by linear programming using
HiGHS. Indicators use the weak inequality in the score definition with
double precision arithmetic. A failed fit or a variance at or below
$10^{-12}$ stops that calculation. All reported fits converged, and all
variances exceeded this threshold.

\subsection*{S6.2. Covariance estimation and bootstrap}

Write $\widehat\Psi_{t,n}(\alpha,\gamma)$ for the estimated corrected
score in Section~\ref{sec:pre_estimation}, evaluated using the fitted
curve and density above. For $h=1$, the ordinary Gaussian multiplier
bootstrap uses
\[
\widehat Z_n^*(\alpha,\gamma)
=\frac{1}{\sqrt n}\sum_{t=1}^{n}
e_t\widehat\Psi_{t,n}(\alpha,\gamma),
\qquad
\widehat s_n^2(\alpha,\gamma)
=\frac1n\sum_{t=1}^{n}\widehat\Psi_{t,n}(\alpha,\gamma)^2,
\]
with independent standard Gaussian $e_t$.

For $h=4$, collect the corrected scores over the evaluation grid in the
vector $\widehat\Psi_{t,n}$ and center them as
\[
\widehat x_t=\widehat\Psi_{t,n}
-\frac1n\sum_{s=1}^{n}\widehat\Psi_{s,n}.
\]
The Bartlett HAC covariance estimate is
\[
\widehat K_\ell
=\frac1n\sum_{t,s=1}^{n}
\left(1-\frac{|t-s|}{\ell}\right)_+
\widehat x_t\widehat x_s^\top.
\]
The block multiplier bootstrap generates $n+\ell-1$ independent
standard Gaussian variables and sets
\[
\xi_t=\frac1{\sqrt\ell}\sum_{j=0}^{\ell-1}e_{t+j},
\qquad
\widehat Z_n^*=\frac1{\sqrt n}\sum_{t=1}^{n}\xi_t\widehat x_t.
\]
The overlapping sums give
$E^*(\xi_t\xi_s)=(1-|t-s|/\ell)_+$, so
$\operatorname{Cov}^*(\widehat Z_n^*)=\widehat K_\ell$ exactly.
The main choice is $\ell=8$, with $\ell=4,12,16$ used for sensitivity
analysis. Thus $\ell$ includes autocovariances through lag $\ell-1$.
The independent multiplier comparisons use either the uncentered scores
as above or the centered scores with $\ell=1$.

At each horizon, the observed process remains
$\sqrt n\,\widehat M_n(\alpha,\gamma)$.
It and the bootstrap process use the same sample variance estimate,
which is held fixed across multiplier draws. For $h=4$, this estimate
is the corresponding diagonal element of $\widehat K_\ell$.
Each draw uses one multiplier sequence for all quantiles, directions,
and penalties. The nuisance estimates remain fixed within the
bootstrap. The shifted process adds
$\widehat d_{\mathrm{pre},n,B}$ from \eqref{eq:dpre-hat} before division
by the same standard deviation.

\subsection*{S6.3. Penalty selection and numerical sensitivity}

The candidate penalties are
$\Lambda_G=\{0,0.1,\ldots,2.0\}$, and the reference class is
$\mathcal B=\{B^*\}$ with $B^*(\alpha)\equiv0.50$.
We use 4,999 multiplier draws for selection and
an independent set of 9,999 draws for calibration. The smallest
maximizer of estimated rejection probability is selected. The
unpenalized comparison sets $\lambda=0$ and uses the same calibration
draws. Critical values use the higher empirical quantile, without
interpolation between order statistics. With $R=9{,}999$, the
reported $p$-value is
\[
\widehat p(\lambda)
=\frac{1+\sum_{r=1}^{R}
\mathbf1\{\widehat T_{n*,\Pi}^{(r)}(\lambda)
\ge\widehat T_{n,\Pi}(\lambda)\}}{R+1}.
\]
The test rejects when the observed statistic exceeds the corresponding
bootstrap critical value. The reported decisions use level $0.10$.

For each horizon, numerical sensitivity checks multiply $b_n$ by
$0.75$ or $1.25$ and increase the number of fitted quantiles from 200
to 400. Each setting is changed separately. These calculations retain
the ordinary multiplier bootstrap for $h=1$ and the main HAC and block
multiplier bootstrap for $h=4$. The sample ending before 2020 is
estimated separately, including its OLS reference value.

For inversion of the common constant restriction, candidate values
start at $-2$ and end at $2$ with spacing $0.05$. An accepted endpoint
triggers an extension of the search range. Observed transitions between
acceptance and rejection are refined to spacing $0.01$, and the penalty
is selected separately at each candidate value. We report the retained
evaluated nodes and their distinct components, without filling gaps
between them. These sets concern a common constant, not a confidence
band for the quantile coefficient curve.

Independent multiplier draws are used to repeat the core comparisons
and selected inversion points. The latter include the nodes adjacent
to observed decision transitions and nodes with $p$-values within
$0.015$ of $0.10$. These comparisons assess numerical sensitivity at
the evaluated points. The Monte Carlo intervals for individual
$p$-values condition on the selected penalty.

\subsection*{S6.4. Sensitivity of the core comparisons}

Tables~\ref{tab:gar_robustness_h1} and~\ref{tab:gar_robustness_h4}
report the comparisons by horizon. For one-quarter growth, the ordinary
multiplier procedure rejects the zero restriction and retains the OLS
benchmark in every reported setting at the 10\% level. This includes
the lower-tail benchmark comparison and the sample ending before 2020.

\begin{table}[!htbp]
\centering
\caption{One-quarter sensitivity results}
\label{tab:gar_robustness_h1}
\begingroup
\footnotesize
\setlength{\tabcolsep}{4pt}
\renewcommand{\arraystretch}{1.08}
\begin{tabular}{llrrrrr}
\toprule
Sample / change & Calibration & $n$ & $c$ & $p_{\mathrm A}$ & $p_{\mathrm U}$ & $\widehat\lambda$ \\
\midrule
\multicolumn{7}{l}{\textit{Zero restriction, $\mathcal T=[0.10,0.90]$}} \\
Full sample & Ordinary & 219 & 0.0000 & 0.0201 & 0.0268 & 0.3 \\
Bandwidth $\times0.75$ & Ordinary & 219 & 0.0000 & 0.0191 & 0.0270 & 0.4 \\
Bandwidth $\times1.25$ & Ordinary & 219 & 0.0000 & 0.0203 & 0.0263 & 0.3 \\
400 fit nodes & Ordinary & 219 & 0.0000 & 0.0204 & 0.0266 & 0.3 \\
Pre-2020 & Ordinary & 195 & 0.0000 & 0.0095 & 0.0412 & 2.0 \\
\addlinespace
\multicolumn{7}{l}{\textit{OLS numerical benchmark, $\mathcal T=[0.10,0.90]$}} \\
Full sample & Ordinary & 219 & -1.2552 & 0.1298 & 0.1238 & 1.3 \\
Bandwidth $\times0.75$ & Ordinary & 219 & -1.2552 & 0.1427 & 0.1240 & 1.8 \\
Bandwidth $\times1.25$ & Ordinary & 219 & -1.2552 & 0.1319 & 0.1236 & 1.4 \\
400 fit nodes & Ordinary & 219 & -1.2552 & 0.1379 & 0.1215 & 1.7 \\
Pre-2020 & Ordinary & 195 & -0.9389 & 0.2515 & 0.5325 & 1.8 \\
\addlinespace
\multicolumn{7}{l}{\textit{OLS numerical benchmark, $\mathcal T=[0.05,0.30]$}} \\
Full sample & Ordinary & 219 & -1.2552 & 0.2870 & 0.3581 & 0.3 \\
Bandwidth $\times0.75$ & Ordinary & 219 & -1.2552 & 0.3075 & 0.3352 & 0.1 \\
Bandwidth $\times1.25$ & Ordinary & 219 & -1.2552 & 0.2891 & 0.3737 & 0.3 \\
400 fit nodes & Ordinary & 219 & -1.2552 & 0.2617 & 0.3378 & 0.4 \\
Pre-2020 & Ordinary & 195 & -0.9389 & 0.1909 & 0.1804 & 0.3 \\
\bottomrule
\end{tabular}
\endgroup
\par\vspace{3pt}
\begin{minipage}{\linewidth}
\footnotesize
Notes: A and U denote adaptive and unpenalized procedures. $c$ is displayed to four decimals but evaluated without rounding. The OLS benchmark is re-estimated in the pre-2020 sample. The displayed changes are assessed separately. All one-quarter rows use the ordinary multiplier bootstrap.
\end{minipage}
\end{table}

For four-quarter growth, treating the corrected scores as serially
uncorrelated produces much smaller $p$-values for the zero restriction.
The ordinary and centered independent multiplier comparisons give
adaptive values of $0.0035$ and $0.0018$, compared with $0.0657$ for
the baseline HAC and block multiplier procedure. The OLS benchmark
over the main range is rejected by both independent multiplier
comparisons, but retained with each of the four reported block lengths.
These differences illustrate the importance of accounting for overlap
when interpreting this application.

The four-quarter lower-tail comparison is less sensitive to the
implementation choices. Its adaptive $p$-values range from $0.0010$
to $0.0324$ across block lengths, and the unpenalized values range
from $0.0011$ to $0.0438$. Both remain below $0.05$ in the sample
ending before 2020, with $p=0.0205$. In contrast, neither procedure
rejects a zero coefficient over the main range in that shorter sample.
Bandwidth changes and the denser fitting grid preserve the baseline
decisions at 10\%, with the dependence treatment held fixed.

\begin{table}[!htbp]
\centering
\caption{Four-quarter sensitivity results}
\label{tab:gar_robustness_h4}
\begingroup
\footnotesize
\setlength{\tabcolsep}{4pt}
\renewcommand{\arraystretch}{1.08}
\begin{tabular}{llrrrrr}
\toprule
Sample / change & Calibration & $n$ & $c$ & $p_{\mathrm A}$ & $p_{\mathrm U}$ & $\widehat\lambda$ \\
\midrule
\multicolumn{7}{l}{\textit{Zero restriction, $\mathcal T=[0.10,0.90]$}} \\
Full sample & Ordinary & 216 & 0.0000 & 0.0035 & 0.0032 & 0.1 \\
Full sample & Centered iid & 216 & 0.0000 & 0.0018 & 0.0016 & 0.1 \\
Full sample & HAC, $\ell=4$ & 216 & 0.0000 & 0.1070 & 0.0840 & 2.0 \\
Full sample & HAC, $\ell=8$ & 216 & 0.0000 & 0.0657 & 0.0540 & 0.2 \\
Full sample & HAC, $\ell=12$ & 216 & 0.0000 & 0.0457 & 0.0341 & 0.4 \\
Full sample & HAC, $\ell=16$ & 216 & 0.0000 & 0.0324 & 0.0251 & 0.4 \\
Bandwidth $\times0.75$ & HAC, $\ell=8$ & 216 & 0.0000 & 0.0647 & 0.0529 & 0.2 \\
Bandwidth $\times1.25$ & HAC, $\ell=8$ & 216 & 0.0000 & 0.0665 & 0.0544 & 0.2 \\
400 fit nodes & HAC, $\ell=8$ & 216 & 0.0000 & 0.0641 & 0.0571 & 0.1 \\
Pre-2020 & HAC, $\ell=8$ & 192 & 0.0000 & 0.1812 & 0.1118 & 1.7 \\
\addlinespace
\multicolumn{7}{l}{\textit{OLS numerical benchmark, $\mathcal T=[0.10,0.90]$}} \\
Full sample & Ordinary & 216 & -0.6009 & 0.0478 & 0.0185 & 0.7 \\
Full sample & Centered iid & 216 & -0.6009 & 0.0609 & 0.0134 & 1.3 \\
Full sample & HAC, $\ell=4$ & 216 & -0.6009 & 0.3134 & 0.2648 & 0.2 \\
Full sample & HAC, $\ell=8$ & 216 & -0.6009 & 0.1890 & 0.2895 & 1.0 \\
Full sample & HAC, $\ell=12$ & 216 & -0.6009 & 0.1451 & 0.2128 & 0.9 \\
Full sample & HAC, $\ell=16$ & 216 & -0.6009 & 0.1584 & 0.1853 & 0.3 \\
Bandwidth $\times0.75$ & HAC, $\ell=8$ & 216 & -0.6009 & 0.2062 & 0.3023 & 1.2 \\
Bandwidth $\times1.25$ & HAC, $\ell=8$ & 216 & -0.6009 & 0.1806 & 0.2775 & 1.1 \\
400 fit nodes & HAC, $\ell=8$ & 216 & -0.6009 & 0.1516 & 0.2566 & 1.3 \\
Pre-2020 & HAC, $\ell=8$ & 192 & -0.5062 & 0.2627 & 0.2425 & 0.7 \\
\addlinespace
\multicolumn{7}{l}{\textit{OLS numerical benchmark, $\mathcal T=[0.05,0.30]$}} \\
Full sample & Ordinary & 216 & -0.6009 & 0.0018 & 0.0017 & 0.1 \\
Full sample & Centered iid & 216 & -0.6009 & 0.0023 & 0.0010 & 1.1 \\
Full sample & HAC, $\ell=4$ & 216 & -0.6009 & 0.0324 & 0.0438 & 0.9 \\
Full sample & HAC, $\ell=8$ & 216 & -0.6009 & 0.0055 & 0.0070 & 0.2 \\
Full sample & HAC, $\ell=12$ & 216 & -0.6009 & 0.0010 & 0.0011 & 0.1 \\
Full sample & HAC, $\ell=16$ & 216 & -0.6009 & 0.0013 & 0.0013 & 0.0 \\
Bandwidth $\times0.75$ & HAC, $\ell=8$ & 216 & -0.6009 & 0.0069 & 0.0069 & 0.0 \\
Bandwidth $\times1.25$ & HAC, $\ell=8$ & 216 & -0.6009 & 0.0070 & 0.0090 & 0.2 \\
400 fit nodes & HAC, $\ell=8$ & 216 & -0.6009 & 0.0062 & 0.0080 & 0.3 \\
Pre-2020 & HAC, $\ell=8$ & 192 & -0.5062 & 0.0205 & 0.0205 & 0.0 \\
\bottomrule
\end{tabular}
\endgroup
\par\vspace{3pt}
\begin{minipage}{\linewidth}
\footnotesize
Notes: $c$ is displayed to four decimals but evaluated without rounding. The displayed changes are assessed separately. HAC uses Bartlett weights with nonzero lags $1,\ldots,\ell-1$ and matching block multipliers. Ordinary and centered iid calibrations are diagnostic comparisons.
\end{minipage}
\end{table}

\FloatBarrier
\subsection*{S6.5. Common constants and numerical sensitivity}
\label{ssec:CI}

Let $\mathcal G_c$ denote the evaluated candidate constants, including
the additional nodes near observed decision transitions. Inverting the
test over the main range gives
\[
\mathcal C_G=\{c\in\mathcal G_c:H_0(c)\text{ is not rejected at level }0.10\}.
\]
Tables~\ref{tab:gar_common_constants_h1} and
\ref{tab:gar_common_constants_h4} report every retained component.
For one-quarter growth, the full-sample adaptive procedure retains a run
from $-1.29$ to $-0.38$ and the isolated node $-0.36$. The unpenalized
procedure retains a single run from $-1.29$ to $-0.27$. Both procedures
retain a single run in the sample ending before 2020.

\begin{table}[!htbp]
\centering
\caption{One-quarter sets of retained common constants}
\label{tab:gar_common_constants_h1}
\begingroup
\footnotesize
\setlength{\tabcolsep}{4pt}
\renewcommand{\arraystretch}{1.08}
\begin{tabular}{llclr}
\toprule
Sample & Calibration & Test & Retained common constants & Nodes \\
\midrule
Full & Ordinary & A & \begin{tabular}[t]{@{}l@{}}$[-1.29,-0.38]\,\cup\,\{-0.36\}$\end{tabular} & 29 \\
Full & Ordinary & U & \begin{tabular}[t]{@{}l@{}}$[-1.29,-0.27]$\end{tabular} & 35 \\
Pre-2020 & Ordinary & A & \begin{tabular}[t]{@{}l@{}}$[-1.15,-0.49]$\end{tabular} & 15 \\
Pre-2020 & Ordinary & U & \begin{tabular}[t]{@{}l@{}}$[-1.23,-0.19]$\end{tabular} & 33 \\
\bottomrule
\end{tabular}
\endgroup
\par\vspace{3pt}
\begin{minipage}{\linewidth}
\footnotesize
Notes: Sets retain the evaluated constants not rejected at the 10\% level. Each bracket denotes a consecutive run of retained evaluated nodes, not every real number between its endpoints. Braces denote a singleton and $\varnothing$ an empty set. The initial spacing is 0.05, with observed transitions refined to 0.01.
\end{minipage}
\end{table}

For four-quarter growth, the baseline adaptive result has eight
components, including isolated retained nodes, and the unpenalized
result has three. The test over the main range therefore retains common
coefficients even though the estimated unrestricted path varies
across quantiles. The components and their endpoints change with
block length. In the sample ending before 2020, the retained runs
span $[-0.76,0.21]$ and $[-0.84,0.02]$, respectively, and both include
zero. Figure~\ref{fig:gar_common_constant_points} displays the
baseline $p$-values at all evaluated candidates.

\begin{table}[!htbp]
\centering
\caption{Four-quarter sets of retained common constants}
\label{tab:gar_common_constants_h4}
\begingroup
\footnotesize
\setlength{\tabcolsep}{4pt}
\renewcommand{\arraystretch}{1.08}
\begin{tabular}{llclr}
\toprule
Sample & Calibration & Test & Retained common constants & Nodes \\
\midrule
Full & Ordinary & A & \begin{tabular}[t]{@{}l@{}}$\{-0.50\}$\end{tabular} & 1 \\
Full & Ordinary & U & \begin{tabular}[t]{@{}l@{}}$\varnothing$\end{tabular} & 0 \\
Full & Centered iid & A & \begin{tabular}[t]{@{}l@{}}$\varnothing$\end{tabular} & 0 \\
Full & Centered iid & U & \begin{tabular}[t]{@{}l@{}}$\varnothing$\end{tabular} & 0 \\
Full & HAC, $\ell=4$ & A & \begin{tabular}[t]{@{}l@{}}$[-0.98,-0.02]\,\cup\,[0.00,0.06]$ \\ $\cup\,\{0.08\}$\end{tabular} & 73 \\
\addlinespace[2pt]
Full & HAC, $\ell=4$ & U & \begin{tabular}[t]{@{}l@{}}$[-1.05,-0.06]$\end{tabular} & 64 \\
Full & HAC, $\ell=8$ & A & \begin{tabular}[t]{@{}l@{}}$[-0.84,-0.83]\,\cup\,[-0.80,-0.79]$ \\ $\cup\,[-0.77,-0.50]\,\cup\,\{-0.48\}$ \\ $\cup\,[-0.44,-0.42]\,\cup\,\{-0.13\}$ \\ $\cup\,[-0.10,-0.09]\,\cup\,[-0.06,-0.05]$\end{tabular} & 25 \\
\addlinespace[2pt]
Full & HAC, $\ell=8$ & U & \begin{tabular}[t]{@{}l@{}}$\{-0.96\}\,\cup\,[-0.94,-0.28]$ \\ $\cup\,\{-0.26\}$\end{tabular} & 45 \\
\addlinespace[2pt]
Full & HAC, $\ell=12$ & A & \begin{tabular}[t]{@{}l@{}}$\{-0.85\}\,\cup\,[-0.81,-0.80]$ \\ $\cup\,[-0.75,-0.52]\,\cup\,[-0.10,-0.09]$\end{tabular} & 13 \\
\addlinespace[2pt]
Full & HAC, $\ell=12$ & U & \begin{tabular}[t]{@{}l@{}}$[-0.89,-0.41]$\end{tabular} & 33 \\
Full & HAC, $\ell=16$ & A & \begin{tabular}[t]{@{}l@{}}$\{-0.85\}\,\cup\,[-0.76,-0.52]$ \\ $\cup\,[-0.50,-0.48]\,\cup\,\{-0.42\}$\end{tabular} & 14 \\
\addlinespace[2pt]
Full & HAC, $\ell=16$ & U & \begin{tabular}[t]{@{}l@{}}$[-0.85,-0.48]\,\cup\,\{-0.46\}$ \\ $\cup\,[-0.44,-0.43]$\end{tabular} & 25 \\
\addlinespace[2pt]
Pre-2020 & HAC, $\ell=8$ & A & \begin{tabular}[t]{@{}l@{}}$[-0.76,0.21]$\end{tabular} & 26 \\
Pre-2020 & HAC, $\ell=8$ & U & \begin{tabular}[t]{@{}l@{}}$[-0.84,0.02]$\end{tabular} & 27 \\
\bottomrule
\end{tabular}
\endgroup
\par\vspace{3pt}
\begin{minipage}{\linewidth}
\footnotesize
Notes: The node and notation conventions in Table~\ref{tab:gar_common_constants_h1} apply.
\end{minipage}
\end{table}

\begin{figure}[!htbp]
\centering
\includegraphics[width=\linewidth]{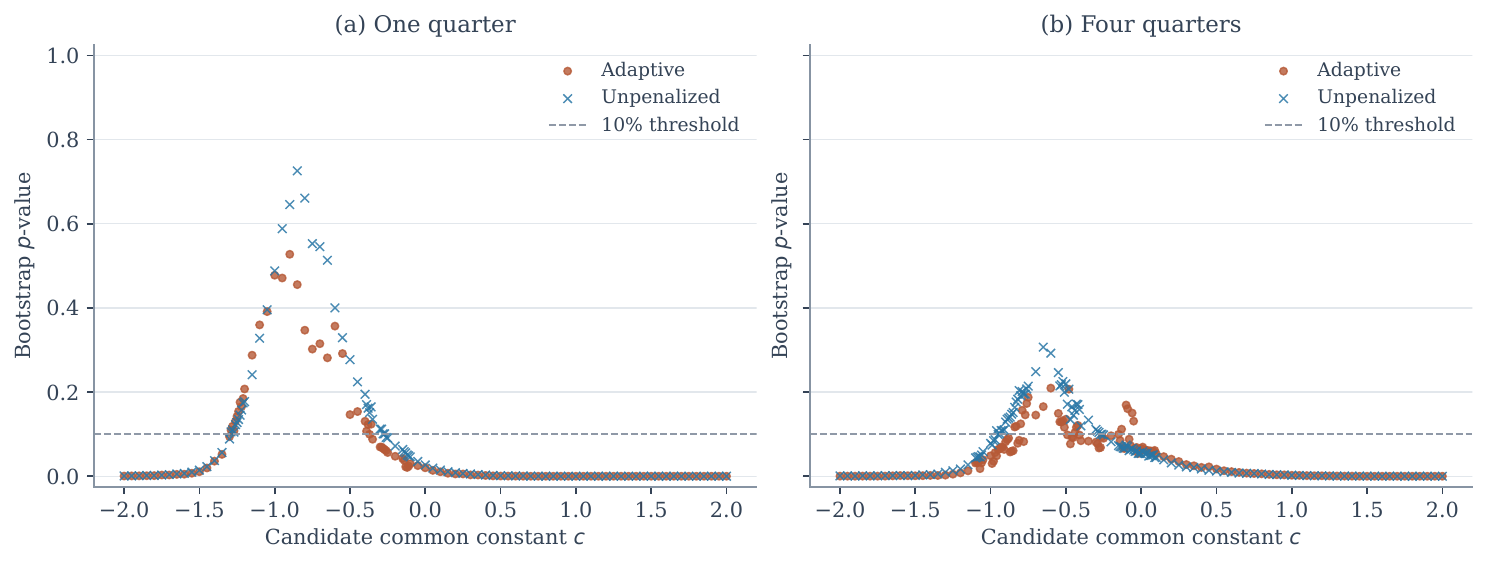}
\caption{Bootstrap $p$-values for the evaluated common constants. The one-quarter panel uses ordinary multipliers and the four-quarter panel uses HAC with block length $\ell=8$. Every evaluated point is shown. The dashed line marks the 10\% threshold.}
\label{fig:gar_common_constant_points}
\end{figure}

Independent multiplier draws leave the core decisions unchanged for
both horizons in the full and shorter samples, although the selected
penalty changes in each of these comparisons. The estimated selection
criterion is often nearly flat across several penalties.
The results at selected inversion boundaries are more sensitive:
the adaptive and unpenalized decisions change at 5 and 1 of the 30
one-quarter points, and at 19 and 6 of the 67 four-quarter points.
Table~\ref{tab:gar_multiplier_repeat} summarizes these comparisons.
The stability of the core tests therefore supports their reported
qualitative conclusions, while the detailed fragmentation and
endpoints of the inversion sets warrant a more cautious interpretation.

\begin{table}[!htbp]
\centering
\caption{Sensitivity to independent multiplier draws}
\label{tab:gar_multiplier_repeat}
\begingroup
\small
\setlength{\tabcolsep}{4pt}
\renewcommand{\arraystretch}{1.08}
\begin{tabular}{clrrrr}
\toprule
$h$ & Comparisons & Number & A changes & U changes & $\widehat\lambda$ changes \\
\midrule
1 & Inversion boundary points & 30 & 5 & 1 & 22 \\
1 & Core: full sample & 3 & 0 & 0 & 3 \\
1 & Core: pre-2020 & 3 & 0 & 0 & 3 \\
4 & Inversion boundary points & 67 & 19 & 6 & 62 \\
4 & Other dependence calibrations & 15 & 0 & 0 & 14 \\
4 & Core: full sample & 3 & 0 & 0 & 3 \\
4 & Core: pre-2020 & 3 & 0 & 0 & 3 \\
\bottomrule
\end{tabular}
\endgroup
\par\vspace{3pt}
\begin{minipage}{\linewidth}
\footnotesize
Notes: The selection and calibration draws are regenerated independently while the data and specification remain unchanged. Decision changes refer to the 10\% level. Boundary points were selected near observed acceptance transitions or $p=0.10$, so their change rates do not describe the entire search grid. The one-quarter entries use ordinary multipliers only.
\end{minipage}
\end{table}

\FloatBarrier

\clearpage
\phantomsection
\addcontentsline{toc}{section}{References}
\bibliographystyle{apalike}
\bibliography{references}

\end{document}